\documentclass{nature}

\usepackage{graphicx}
\makeatletter
\let\saved@includegraphics\includegraphics
\AtBeginDocument{\let\includegraphics\saved@includegraphics}
\renewenvironment*{figure}{\@float{figure}}{\end@float}
\makeatother

\usepackage{amsthm,amsmath,amsopn,amsfonts,amssymb,verbatim,color}
\usepackage{subfigure}
\usepackage{bm}
\usepackage{epsfig,slashed}
\usepackage[T1]{fontenc}
\usepackage[colorlinks=true,citecolor=blue,linkcolor=blue,urlcolor=blue]{hyperref}
\usepackage{psfrag}

\usepackage{color}

\title{Quantum Entangled Fractional Topology and Curvatures}

\author{Joel Hutchinson$^{1}$ \& Karyn Le Hur$^1$}

\begin{document}

\newcommand{\tr}{\mathop{\mathrm{Tr}}}
\newcommand{\bsigma}{\boldsymbol{\sigma}}
\newcommand{\re}{\mathop{\mathrm{Re}}}
\newcommand{\im}{\mathop{\mathrm{Im}}}
\renewcommand{\b}[1]{{\boldsymbol{#1}}}
\newcommand{\diag}{\mathrm{diag}}
\newcommand{\sign}{\mathrm{sign}}
\newcommand{\sgn}{\mathop{\mathrm{sgn}}}
\renewcommand{\c}[1]{\mathcal{#1}}
\renewcommand{\d}{\text{\dj}}
\newcommand{\red}{\textcolor{red}}
\newcommand{\green}{\textcolor{green}}
\newcommand{\noteKLH}[1]{{\color{black}  #1}}

\newcommand{\mb}{\bm}
\newcommand{\ua}{\uparrow}
\newcommand{\da}{\downarrow}
\newcommand{\ra}{\rightarrow}
\newcommand{\la}{\leftarrow}
\newcommand{\mc}{\mathcal}
\newcommand{\bs}{\boldsymbol}
\newcommand{\lra}{\leftrightarrow}
\newcommand{\nn}{\nonumber}
\newcommand{\half}{{\textstyle{\frac{1}{2}}}}
\newcommand{\mf}{\mathfrak}
\newcommand{\MF}{\text{MF}}
\newcommand{\IR}{\text{IR}}
\newcommand{\UV}{\text{UV}}
\newcommand{\be}{\begin{equation}}
\newcommand{\ee}{\end{equation}}
\maketitle

\begin{affiliations}
 \item CPHT, CNRS, Ecole Polytechnique, Institut Polytechnique de Paris, \\
Route de Saclay, 91128 Palaiseau, France
\end{affiliations}

\begin{abstract}
Topological spaces find various applications in phases of matter and properties of these systems can be described through an analogy with spin-one-half particles. 
 Here we propose models of interacting spins with applied magnetic fields acting on the Poincar\' e-Bloch sphere to reveal a new class of \noteKLH{topological states} with rational-valued Chern numbers for each spin. We present a geometrical definition of this fractional topology related to the formation of a product state at the north pole and a maximally entangled state at the south pole. We study a driving protocol in time and the spin magnetizations at the poles to reveal the stability of the fractional topological numbers towards various forms of interactions in the adiabatic limit. We elucidate a correspondence between a two-spin system with one-half topological number for each spin and a topological bilayer model on a honeycomb lattice, which describes semimetals with a nodal ring at one Dirac point encircling a region of entanglement and a topological bandgap at the other Dirac point revealing a $\pi$ Berry phase. Such materials belong to a novel topological class with a $\mathbb{Z}_2$ layer symmetry. We discuss the bulk-edge correspondence applied to quantum transport and light-matter coupling. 
 \end{abstract}

In recent years, rising interest in topology travels from mathematics to physics related to advancing quantum science and technology. This allows for the direct observation of the Chern number, a measure that distinguishes topological insulators and superconductors\cite{HasanKane, Zhang}. The properties of these systems can be revealed from the reciprocal or momentum space showing how the topology is already encoded in a spin-1/2 particle or two-state system when equivalently applying a magnetic field that acts radially on the sphere\cite{schleich2011} with polar angle $\theta$ and azimuthal angle $\phi$. Upon adiabatically sweeping from the north to south pole, along a curved path \noteKLH{ with fixed angle $\phi$}, the Chern number $\mathcal{C}$ of this two-state system represented by a vector of Pauli matrices $\b{\sigma}=(\sigma_x,\sigma_y,\sigma_z)$ is equal to one. \noteKLH{ Incredibly this topological quantity can be measured directly from the spin magnetizations at the poles:}\cite{henriet2017, gritsev2012, degrandi2010}
\begin{eqnarray}
\label{eq:chern_dyn}
\mathcal{C}&\equiv&\frac{1}{2\pi}\int_0^{2\pi}d\phi\int_0^\pi d\theta\mathcal{F}_{\phi\theta}\label{eq:chern}\\ \nonumber
&=&\frac{1}{2}\bigg(\langle\sigma_z(\theta=0)\rangle-\langle\sigma_z(\theta=\pi)\rangle\bigg)=1.
\end{eqnarray}
The angles $\theta=0$ and $\theta=\pi$ refer to the north and south poles of the sphere, respectively.  We have introduced the Berry curvature 
\be
\mathcal{F}_{\phi\theta}\equiv\partial_\phi\mathcal{A}_\theta-\partial_\theta\mathcal{A}_\phi,
\ee
and the Berry connection $\b{\mathcal{A}}$, defined from the gradient of the ground state $|\psi\rangle$ according to\cite{berry1984}
\be\label{eq:BerryConnect}
\mathcal{A}_\alpha=i\langle\psi|\partial_\alpha|\psi\rangle.
\ee
The associated Berry phase represents an important foundation of quantum physics\cite{leek1889}. In the quantum Hall effect, such a geometrical description in terms of curvatures plays a key role in the link with electronic transport properties such as the quantum Hall conductivity\cite{Thouless,Haldanecurvature}. Here, the integer Chern number $\mathcal{C}$ of a given spin-1/2 is related to a topological charge  -- the degeneracy point of the Hamiltonian -- contained within the sphere spanned by the magnetic field vector. The spin-1/2 orientation then measures directly this topological charge\cite{schroer2014,roushan2014,korber2020}. A recent experiment\cite{roushan2014}  has studied two spin-1/2s, $\b{\sigma}_1$, $\b{\sigma}_2$, under the influence of the radial fields $\b{H}_1$ and $\b{H}_2$ forming the surface of the sphere. The two spins interact through a transverse coupling $(\sigma_1^x\sigma_2^x+\sigma_1^y\sigma_2^y)$. Their resulting topological phase diagram consists of integer $\mathcal{C}=0$, $1$ and $2$ phases, corresponding to topological charges located outside both spheres, inside one sphere, and inside both spheres respectively. To show the possibility of entangled states with a stable fractional Chern number for each spin, we add a crucial ingredient corresponding to adjustable constant magnetic fields on the sphere. In the following, we introduce a model with two spins $\b{\sigma}^1=(\sigma_1^x,\sigma_1^y,\sigma_1^z)$ and $\b{\sigma}^2=(\sigma_2^x,\sigma_2^y,\sigma_2^z)$ interacting through an Ising coupling, to reveal half-topological numbers for each spin on the sphere. The topology is defined on each sub-system, here a spin-1/2, directly from the poles. We show applications of the spheres with $\mathcal{C}=1/2$ per spin for the characterization of topological semimetallic phases in bilayer honeycomb systems showing one topological Dirac point associated with a $\pi$ Berry phase and another Dirac point revealing a nodal entangled ring.  In the Supplementary Information, we show that this is just one example of a large class of models, including $XY$ couplings and higher numbers of spins that all reveal the same effect. 

\section{Model with Two Spheres}

The Hamiltonian for two spheres reads
\be
\mathcal{H}^{\pm}=-(\b{H}_1\cdot\b{\sigma}^1\pm\b{H}_2\cdot\b{\sigma}^2)\pm\tilde{r}f(\theta)\sigma_z^1\sigma_z^2.\label{eq:H_twoqubit}
\ee
The magnetic field ${\bf H}_i$ acts on the same sphere parameterized by $(\theta,\phi)$ and may be distorted along the $\hat{z}$ direction with the addition of the uniform field $M_i$ according to\cite{henriet2017}: 
\be
\b{H}_i=(H\sin\theta\cos\phi,H\sin\theta\sin\phi,H\cos\theta+M_i),
\ee
for $i=1,2$. We show below through energetics arguments that the fields $M_i$ are indeed important to stabilize a fractional Chern number. We also consider a generic $\theta$-dependent coupling $\tilde{r}f(\theta)$ with $\tilde{r}f(\theta)>0$. The $\pm$ denote two distinct classes of models. 
It is important to highlight here that in the case where a spin-1/2 is coupled to an environment, the topological number associated to the spin may vary continuously from $\mathcal{C}=1$ to $\mathcal{C}=0$ dependently on the coupling strength between the two systems. In the present case, we show that the fractional Chern numbers are stable towards smooth deformations
of the geometry and towards the form of the interactions. In experiments, the magnetizations may be measured for each spin independently, such that the Chern number also has a well-defined component corresponding to each subsystem. Therefore, we find it important to first generalize Eqs.~\eqref{eq:chern_dyn} and \eqref{eq:BerryConnect} for subsystem or spin $j$ in the interacting model. The corresponding Chern number $\mathcal{C}^j$ will provide a robust topological number related to the quantum Hall conductivity and will also represent a measure of entanglement. The spin system we consider here provides a nice platform for understanding how topology can be partitioned between subsystems. 

While the eigenstates of the Hamiltonian~\eqref{eq:H_twoqubit} are in general complicated for $\tilde{r}f(\theta)\neq 0$, their $\phi$-dependence is very simple, such that the ground state wavefunction of the system can be written as $|\psi\rangle=\sum_{kl}c_{kl}(\theta)|\Phi_k(\phi)\rangle_1|\Phi_l(\phi)\rangle_2$. In the standard representation of a single spin eigenstate in a radial magnetic field, the ground state is $|\uparrow\rangle$ at the north pole where $\theta=0$, and $e^{i\phi}|\downarrow\rangle$ at the south pole where $\theta=\pi$. \noteKLH{ We will take these states to form our single-spin basis and introduce the standard spinor representation, $|\Phi_+(\phi)\rangle_j=|\uparrow\rangle=
\left(
\begin{smallmatrix}
1 \\
0
\end{smallmatrix}
\right)$
and $|\Phi_-(\phi)\rangle_j=e^{i\phi}|\downarrow\rangle=\left(
\begin{smallmatrix}
0\\
e^{i\phi}
\end{smallmatrix}
\right)$
for the two spins with $j=1,2$.
Therefore, in the wavefunction we have $k,l=\pm$.
}
For $\tilde{r}\rightarrow 0$ and $M_i\rightarrow 0$, the ground state $|\psi\rangle$ then shows $c_{++}(\theta)=\cos^2\frac{\theta}{2}$, $c_{--}(\theta)=\sin^2\frac{\theta}{2}$, $c_{-+}(\theta)=c_{+-}(\theta)=\sin\frac{\theta}{2}\cos\frac{\theta}{2}$ with the normalization equation, $\sum_{kl} |c_{kl}|^2=1$ where $k,l=\pm$. While there are many ways to represent these single-spin states, their relative phase $e^{i\phi}$ is fixed. 
\noteKLH{ At a general level, we have $c_{kl}(\theta)=c^1_{kl}(\theta)c^2_{kl}(\theta)$ such that the wavefunction of the system can be equivalently written as $|\psi\rangle=\sum_{kl}|c^1_{kl}(\theta)\Phi(\phi)_k\rangle_1|c^2_{kl}(\theta)\Phi(\phi)_l\rangle_2$. Then, we introduce the partial derivative symbol $\partial_\alpha^1$, which equally refers to $\partial_{\alpha}^1\mathbb{I}\otimes \mathbb{I}$ when applied on $|\psi_g\rangle$, where $2\times 2$ identity matrices $\mathbb{I}$ mean that the partial derivative acts identically on the two components of a spin-1/2 spinor and through the direct product it acts on the sub-space of one spin-1/2 only (here the first spin). We introduce a similar definition for  $\partial_\alpha^2$ as $\mathbb{I}\otimes\partial_{\alpha}^2\mathbb{I}$ acting on the second spin-1/2. The Berry connection for the $j$th spin is then naturally defined as $\mathcal{A}_\alpha^j\equiv i\langle\psi|\partial_\alpha^j|\psi\rangle$ 
where $\alpha=\phi,\theta$, along with the $j$th Berry curvature $\mathcal{F}^j_{\phi\theta}=\partial_\phi\mathcal{A}^j_\theta-\partial_\theta\mathcal{A}^j_\phi$, and Chern number
\be
\mathcal{C}^j=\frac{1}{2\pi}\int_0^{2\pi}d\phi\int_0^\pi d\theta\mathcal{F}^j_{\phi\theta}.
\ee
The operator $\partial_\alpha^j$ acts on the Hilbert space of the $j$th spin. Here, $\mathcal{A}^j_\theta$ is not uniquely defined, but $\mathcal{C}^j$ still is since $\mathcal{A}^j_\theta = i\sum_{kl}c_{kl}^*(\theta)\partial_\theta^jc_{kl}(\theta)$ with $\partial_\theta^1 c_{kl}(\theta)=c_{kl}^2(\theta)(\partial_\theta c_{kl}^1(\theta))$ and $\partial_\theta^2 c_{kl}(\theta)=c_{kl}^1(\theta)(\partial_\theta c_{kl}^2(\theta))$, so that we can safely summarize that $\partial^j_\phi\mathcal{A}^j_{\theta}=0$. }
From  the relations $\partial_\phi^j |\Phi_-(\phi)\rangle_j = i|\Phi_-(\phi)\rangle_j$ and  $\partial_\phi^j |\Phi_+(\phi)\rangle_j=0$, the Berry connection then reads: 
$\mathcal{A}^1_\phi =-|c_{-+}(\theta)|^2-|c_{--}(\theta)|^2$ and $\mathcal{A}^2_\phi=-|c_{+-}(\theta)|^2-|c_{--}(\theta)|^2$.
Note that product states such as $|\uparrow\rangle_1|\uparrow\rangle_2$ or $|\downarrow\rangle_1|\downarrow\rangle_2$ will contribute $0$ or $-1$ to the Berry connection, while a maximally entangled Einstein-Podolsky-Rosen or Bell state\cite{bell1964}, such as $\frac{1}{\sqrt{2}}(|\uparrow\rangle_1|\downarrow\rangle_2+|\downarrow\rangle_1|\uparrow\rangle_2)$ will give $-1/2$.

In the Supplementary Information, we show that the Chern number for the $j$th spin can be written as
\begin{eqnarray}
\label{Cj}
\mathcal{C}^i&=&-(\mathcal{A}^i_\phi(\pi)-\mathcal{A}^i_\phi(0)).
\end{eqnarray}
This form is gauge invariant as shown in the Supplementary Information through the Stokes' theorem and the introduction of smooth fields. It is interesting to observe that a similar correspondence is useful to describe the `quantized' topological response of one pseudospin-1/2 when coupling with circularly polarized light then referring to quantized circular dichroism of light\cite{Philipp,Nathan,Asteria}. Here, we also show that Eq. \eqref{Cj} defining the topology at the poles only, is related to the charge polarization and the quantum Hall conductivity for the sub-system $j$ itself; see Supplementary Information.
Then, we have the general result
\begin{eqnarray}
\mathcal{C}^1&=&|c_{-+}(\pi)|^2+|c_{--}(\pi)|^2-|c_{-+}(0)|^2-|c_{--}(0)|^2, \\ \nonumber
\mathcal{C}^2&=&|c_{+-}(\pi)|^2+|c_{--}(\pi)|^2-|c_{+-}(0)|^2-|c_{--}(0)|^2.\label{eq:cherncoeff}
\end{eqnarray}
From the Pauli operator $\sigma_z^j = |\uparrow\rangle_{j j}\langle \uparrow | - |\downarrow\rangle_{jj}\langle \downarrow|$, and from the normalization equation of the state $|\psi\rangle$, we also find the
equality $\langle \psi| \sigma_z^j |\psi\rangle = 1+2\mathcal{A}^j_\phi$, leading to
\be
\mathcal{C}^j=\frac{1}{2}\bigg(\langle\sigma_z^j(\theta=0)\rangle-\langle\sigma_z^j(\theta=\pi)\rangle\bigg).\label{eq:chern_mag}
\ee
Eq.~\eqref{eq:cherncoeff} is an interesting generalization of Eq.~\eqref{eq:chern_dyn} because this shows that one can yet define and measure for these interacting models in curved space the topology from the magnetizations of a given spin $j$ at the poles. 

Now, we consider the specific system of interest whose ground state evolves from a product state at $\theta=0$ to an entangled state at $\theta=\pi$: 
\be
|\uparrow\rangle_1|\uparrow\rangle_2\rightarrow \frac{1}{\sqrt{2}}(|\uparrow\rangle_1|\downarrow\rangle_2+|\downarrow\rangle_1|\uparrow\rangle_2).
\label{entangledstate}
\ee 
The non-zero coefficients are $|c_{++}(0)|^2=1$, $|c_{+-}(\pi)|^2=|c_{-+}(\pi)|^2=\frac{1}{2}$, for which
\be
\mathcal{C}^1=\mathcal{C}^2=\frac{1}{2}.\label{eq:chernhalf}
\ee
The presence of entanglement at one pole leads to a fractional Chern number of $1/2$ for each spin. 
This value is in agreement with $\langle\sigma_z^j(\theta=0)\rangle=1$ and with $\langle\sigma_z^j(\theta=\pi)\rangle=0$, reflecting the formation of a maximally entangled Bell pair at the south pole\cite{bell1964}. The norm of each spin effectively shrinks at the south pole, leading to a $\ln 2$ entanglement entropy\cite{entropy}. In the case where the two spins would form a product state that follows the magnetic field, then from $c_{++}(0)=1$ and $c_{--}(\pi)=1$, we verify $\mathcal{C}^j=1$. In the case where the two spins would be entangled at both poles then $\mathcal{C}^j=0$. 



\begin{figure}[b]
	\centering
	\includegraphics[width=0.5\columnwidth]{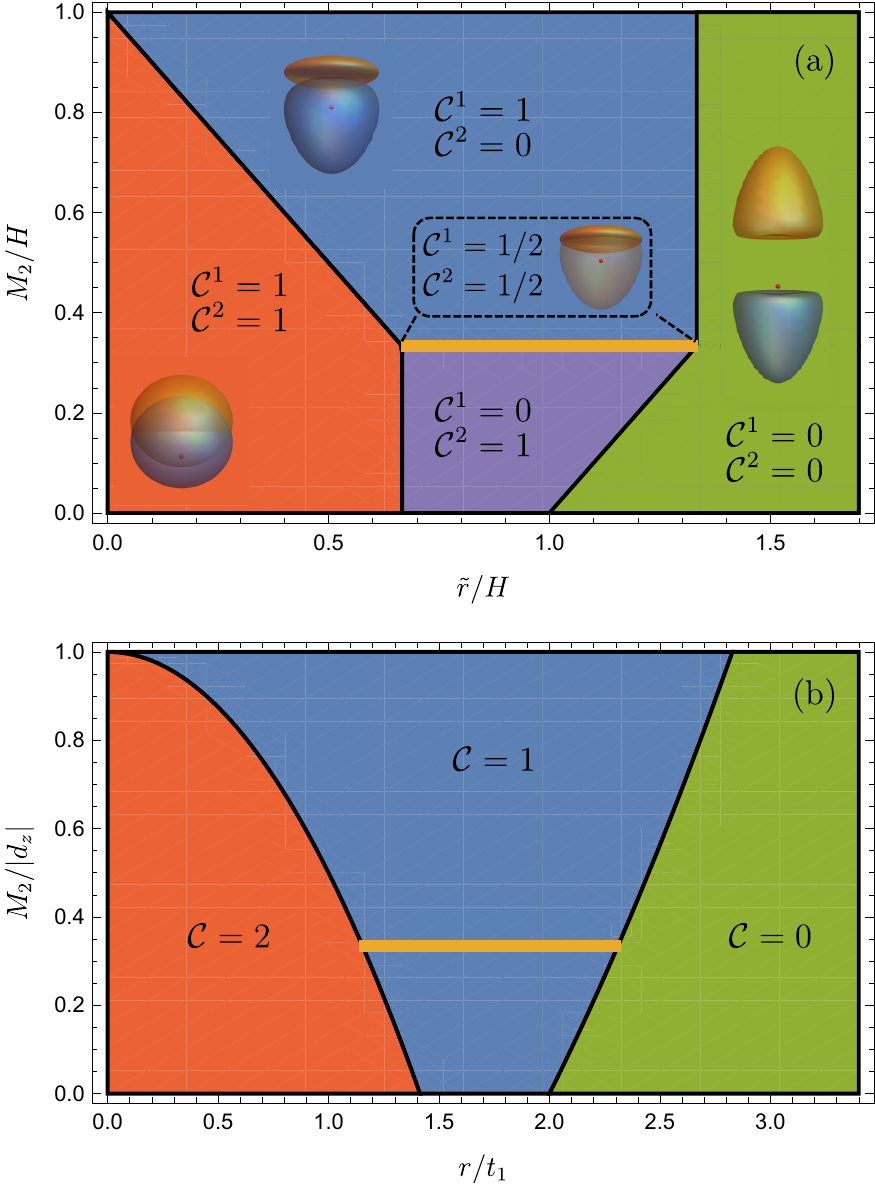}
\caption{a) Spin model: topological phase diagram for the Chern number of each spin in the $(M_2, \tilde{r})$ plane. Here we have set $M_1=H/3$. The gold line at $M_2=M_1$ indicates the symmetric phase $\mathcal{C}^1=\mathcal{C}^2=\frac{1}{2}$. The insets illustrate (adiabatically deformed) spheres corresponding to the parameter space spanned by the effective field for spin 1 (blue) and spin 2 (orange) in each phase. The topological charge at the origin is indicated by the red dot. Along the gold line the effective-field manifold (which is identical for each spin) is in a coherent superposition of containing and not containing the monopole yielding a Chern number of $1/2$ for each spin. b) Lattice model: topological phase diagram in the $(M_2,r$) plane for the total Chern number at half-filling, defined with the two lowest occupied bands. The gold line at $M_2=M_1$ indicates the symmetric phase for which the gap is closed. The parameters $t_1$ and $|d_z|=3\sqrt{3}t_2$ in the Haldane model are shown in the Supplementary Information.}\label{fig:topo_phase}
\end{figure}

\noteKLH{ To show that our model in Eq.~\eqref{eq:H_twoqubit} does indeed fulfill the necessary prerequisites to observe $\mathcal{C}^1=\mathcal{C}^2=\frac{1}{2}$, we study the topological phase diagram which is entirely determined by the energetics at the poles. For clarity, we analyse the $\mathcal{H}^+$ Hamiltonian hereafter (the $\mathcal{H}^-$ Hamiltonian reveals a similar fractional phase).} 
At the poles, the ground state is readily determined, and the resulting topological phase for each spin is shown in Fig.~1a for a constant interaction $f(\theta)=1$. Allowing for a non-constant interaction does not change this phase diagram significantly, though it does open up the intriguing possibility of a direct transition from $\mathcal{C}^1+\mathcal{C}^2=2$ to $\mathcal{C}^1+\mathcal{C}^2=0$ at the solution of $(H-M_2)/f(\pi)=\tilde{r}=(H+M_1)/f(0)$.

In the presence of $\mathbb{Z}_2$ symmetry between the two spins corresponding to $\sigma_z^1 \leftrightarrow \sigma_z^2$ when
$M_1=M_2\equiv M$ in Eq. (\ref{eq:H_twoqubit}), the ground state at the north pole with $\theta=0$, is $|\uparrow\rangle_1|\uparrow\rangle_2$ provided that $\tilde{r}f(0)<H+M$. At the south pole with $\theta=\pi$, the ground state is $|\downarrow\rangle_1|\downarrow\rangle_2$ for $\tilde{r}f(\pi)<H-M$, but it is degenerate between 
the anti-aligned configurations for $\tilde{r}f(\pi)>H-M$. In that case, the presence of the transverse fields in the Hamiltonian
along the path over the sphere will then produce the analogue of resonating valence bonds\cite{Anderson}. Indeed in Sec.~\ref{sec:timedepend}, we will see that the singlet state is decoupled from the rest, while the triplet state $\frac{1}{\sqrt{2}}(|\uparrow\rangle_1|\downarrow\rangle_2 + |\downarrow\rangle_1|\uparrow\rangle_2)$, \noteKLH{ showing the resonance between the states $|\uparrow\rangle_1|\downarrow\rangle_2$ and $|\downarrow\rangle_1|\uparrow\rangle_2$}, is the one adiabatically connected to the $\theta=0$ ground state. As a result, we obtain half-integer Chern numbers \eqref{eq:chernhalf}. For the simple constant interaction $f(\theta)=1$, this occurs within the range 
\be
H-M<\tilde{r}<H+M,
\ee 
indicated by the gold line in Fig.~1a. This line can be considered as a critical point between two distinct topological phases of a given spin. In the limit $M\rightarrow0$, it becomes the quantum critical point between the total-Chern-number 2 and total-Chern-number 0 phases. We find that the $\mathcal{H}^-$ Hamiltonian also contains a line of fractional Chern numbers with $\mathcal{C}^1=-\mathcal{C}^2=\frac{1}{2}$. In the Supplementary
Information, we show that the fractional phase with $C^j=1/2$ can be stabilised and in fact spreads in the presence of an XY coupling. 

We also verify that the fractional Chern number may be generalized for $N>2$ spins; starting from a product state at the north pole, spins may evolve via the transverse field to an entangled state at the south pole with $\mathcal{C}^j=1/2$ for an even number of spins or $\mathcal{C}^j=\frac{N+1}{2N}$ for a frustrated system with an odd number of spins. The spin model of Fig. 3(b) in the Supplementary Information, at the south pole, can be mapped onto the same Majorana fermions as in the Kitaev spin ladder geometry\cite{lehur2017} through
the Jordan-Wigner transformation, providing a relation between $\mathbb{Z}_2$ gauge theories and $\mathcal{C}^j=1/2$ for $N=4$ spins. 

There is another geometric picture we can use to understand the topological nature of these numbers. For a spin-1/2 system, the Chern number counts the number of degeneracy monopoles associated to the topological charges contained within the closed manifold spanned by the magnetic field, in accordance with Gauss' law\cite{bernevig2013}. We can adapt this picture to the case of interacting spins, where the effective magnetic field for each spin depends on the orientation of the other. In a mean-field sense, this would amount to $\b{H}_1^{\rm eff}=-\b{H}_1+\tilde{r}f(\theta)\langle\sigma^2_z\rangle\hat{z}$ with $\hat{z}$ a unit vector along the $z$ axis. Each of the two manifolds spanned by $\b{H}_1^{\rm eff}$, $\b{H}_2^{\rm eff}$ may or may not contain the degeneracy monopole as illustrated by the insets in Fig.~1a, resulting in the different possibilities of $\mathcal{C}_i=0,\pm1$. Thus, $\mathcal{C}_i$, which counts the topological charge of the effective model describing the subsystem, is robust against local perturbations of the effective field. For the entangled case, the manifold spanned by the effective magnetic field on each spin rather consists of a coherent superposition of two geometries: the one that contains the monopole and the one that does not, represented schematically by the inset corresponding to the gold line in Fig.~1a. From Stokes' theorem, the geometry (hemisphere) encircling the topological charge can be related to a pole associated to a $\pi$ Berry phase; see Eq. (33) of Supplementary Information.


Now, we show that this spin-1/2 model can also find applications in topological lattice models.
It is well known that the Haldane model\cite{haldane1988}-- a two-dimensional Chern insulator which has been realized in quantum materials\cite{liu_2015}, graphene\cite{McIver_2019}, cold atoms\cite{jotzu_2014,flaschner_2016} and light systems\cite{haldane_2008, Lu_2014,koch_2010,Le_Hur_2016,Tomoki_2019} -- has a natural pseudspin-1/2 representation due to the $A$ and $B$ sublattices of the honeycomb lattice where the Brillouin zone torus can be mapped onto the parameter space discussed above. It follows that a stack of two Haldane layers may be represented by a two-spin model. 

\section{Lattice model}\label{sec:lattice}

We consider a plane realization of Eq. \eqref{eq:H_twoqubit} consisting of two $AA$ and $BB$-stacked graphene lattices\cite{cheng2019} and show how to find a fractional magnetization representing $\mathcal{C}^j$. Here, 
$\theta=0$ and $\theta=\pi$ map onto the $K$ and $K'$ points of the first Brillouin zone respectively (see Supplementary Information). The spin degrees of freedom now describe the momentum-space sublattice magnetization for each layer $j$. A correspondence between the spheres model and the lattice model, \noteKLH{ that will be developed below Eq. (\ref{eigenstates})}, can be formulated through the identification $\sigma_z^j\leftrightarrow n_{\b{k}B}^j-n_{\b{k}A}^j$. Here, $n_{\b{k}\alpha}^j$ represents the density of particles associated to sublattice $\alpha=A$ or $B$ for a wavevector ${\bf k}$, in a given layer $j$. \noteKLH{ The bilayer system is half-filled.}
The values $M_j$ from the previous section now describe inversion-symmetry breaking Semenoff masses, which may be tuned for each layer\cite{semenoff1984}. We highlight here that from the spheres' formalism, the topology is introduced here through tunable Berry phases in each layer in accordance with the Haldane model. If the two layers have equal fluxes, the model corresponds to the $\mathcal{H}^+$ Hamiltonian, while if they have opposite fluxes, it describes the $\mathcal{H}^-$ Hamiltonian, which is equivalent to the Kane-Mele model\cite{KM}. Here, we will discuss the situation with equal fluxes. 
The mapping suggests that we need an unusual interaction -- one that is local in $\b{k}$-space -- to produce a momentum-dependent Ising interaction. Such interactions have been studied in relation to Weyl semimetals\cite{morimoto2016, meng2019}. In fact, we can achieve the same result with an interlayer coupling $r$ between neighbouring sites.

All of this motivates the following lattice model in momentum-space:
\be
\mathcal{H}=\sum_{\b{k}}(\psi^\dagger_{\b{k}1},\psi^\dagger_{\b{k}2})\mathcal{H}(\b{k})\begin{pmatrix}
\psi_{\b{k}1}\\ \psi_{\b{k}2}\end{pmatrix},
\ee
where $\psi^\dagger_{\b{k}i}\equiv(c^\dagger_{\b{k}Ai}, c^\dagger_{\b{k}Bi})$ and 
\be
\mathcal{H}(\b{k})=\begin{pmatrix}
(\b{d}+M_1\hat{z})\cdot\b{\sigma} & r\mathbb{I}\\
r\mathbb{I} & (\b{d}+M_2\hat{z})\cdot\b{\sigma}
\end{pmatrix},
\ee
is represented in terms of the Pauli matrices $\b{\sigma}$, the $2\times2$ identity matrix $\mathbb{I}$, and the $\b{k}$-dependent vector $\b{d}$ is defined in accordance with the Haldane model in each layer (see the Supplementary Information for details). The indices $i=1,2$ indicate the layer. 

The eigenvalues and eigenvectors of this matrix are readily found at the $K$ and $K'$ points where the gap closes, respectively, for the values of $r$:
\begin{eqnarray}
r_c^+&\equiv&\sqrt{|d_z|+M_1}\sqrt{|d_z|+M_2} \\ \nonumber 
r_c^-&\equiv&\sqrt{|d_z|-M_1}\sqrt{|d_z|-M_2}.
\end{eqnarray}
For the case of asymmetric Semenoff masses $M_1\neq M_2$, the gap closes and reopens at $r_c^-$. When $M_1\neq M_2$, computing the Berry curvature numerically\cite{fukui2005}, we show in Fig.~1b, the phase diagram for the total Chern number $\mathcal{C}$ at half filling defined from the two lowest occupied bands, in agreement with established results\cite{cheng2019}. A topological transition takes place where the Chern number of the second band changes from $1$ to $0$. When the gap closes and reopens at $K$ this number goes to $-1$. The Chern number of the first band (lowest band) remains $1$ throughout. The similarity between Fig.~1a and Fig.~1b suggests that there indeed exists a faithful mapping between the lattice model and the spin model, which has been shown to be certainly valid close to the transition between the phases $\mathcal{C}=2$ and $\mathcal{C}=1$ (starting from the  $\mathcal{C}=1$)\cite{cheng2019}.

Now, we study the (gold) line $M_1=M_2$ where the system shows an additional $\mathbb{Z}_2$ layer symmetry $(1\leftrightarrow 2)$ which is at the origin of the fractional Chern number. This situation describes a special class, where time-reversal and inversion symmetry are not present due to the flux and mass terms, while a $\mathbb{Z}_2$ symmetry is preserved. The result is a nodal ring semimetal where the second and third bands cross as shown in Fig.~2a. The time-reversal invariant version of this case has been discussed\cite{young2015}. The eigenstates at the poles take the simple form
\begin{eqnarray}
\label{eigenstates}
\psi_1&=&\frac{1}{\sqrt{2}}(0,-1,0,1),\;\;
\psi_2=\frac{1}{\sqrt{2}}(0,1,0,1),\nonumber\\
\psi_3&=&\frac{1}{\sqrt{2}}(-1,0,1,0),\;\;
\psi_4=\frac{1}{\sqrt{2}}(1,0,1,0).
\end{eqnarray}
Defining $|\psi_g\rangle\equiv\frac{1}{2}(c^\dagger_{A1}c^\dagger_{B1}-c^\dagger_{A1}c^\dagger_{B2}-c^\dagger_{A2}c^\dagger_{B1}+c^\dagger_{A2}c^\dagger_{B2})|0\rangle$, we see that at $r=r_c^+$, there is a transition in the ground state at $K$ from $c^\dagger_{B1}c^\dagger_{B2}|0\rangle$, with $|0\rangle$ referring to the vacuum state, to $|\psi_g\rangle$. Meanwhile at $K'$, there is a transition at $r=r_c^-$ from $c^\dagger_{A1}c^\dagger_{A2}|0\rangle$ (which is favoured by the Semenoff masses) to $|\psi_g\rangle$ (which is favoured by the interaction). \noteKLH{ Here, we develop the correspondence with the spheres model for specific values of $r$ between $r_c^-$ and $r_c^+$ such that the fractional state can occur. At the $K$ point, the ground state can be written as $c^\dagger_{B1}c^\dagger_{B2}|0\rangle=|\uparrow \uparrow\rangle$ justifying why we define the pseudo-spin magnetization in each plane
as $\sigma_z^j(K) \leftrightarrow n_B^j(K) - n_A^j(K)$ with $n_B^j=c^{\dagger}_{Bj}c_{Bj}$ and similarly for the sublattice $A$. Here,  $\sigma_z^j$ measures the particle-density asymmetry between sublattice $A$ and $B$ resolved for a $k$ value. At the $K$ point, populating an eigenstate $c^\dagger_{A1}c^\dagger_{A2}|0\rangle=|\downarrow \downarrow\rangle$ then requires an additional energy related to $|d_z|$. At the $K'$ point, the nodal ring involves the state $|\psi_g\rangle$. Importantly, the states $c^{\dagger}_{A1}c^{\dagger}_{B1}|0\rangle$ and 
$c^{\dagger}_{A2}c^{\dagger}_{B2}|0\rangle$ do not modify the pseudo-spin magnetization in each plane as they favor an equal particle density on the two sublattices, but they will participate in the entanglement entropy maximum in the nodal ring region. Therefore, from the point of view of the pseudo-spin magnetization at the Dirac points or equivalently at the poles on the sphere then only four states intervene. Explicitly, the correspondence between states in the lattice model and states in the sphere model is given by
\be
c^\dagger_{B1}c^\dagger_{B2}|0\rangle=|\uparrow \uparrow\rangle,\ c^\dagger_{A1}c^\dagger_{A2}|0\rangle=|\downarrow\downarrow\rangle,\ c^\dagger_{B1}c^\dagger_{A2}|0\rangle=|\uparrow \downarrow\rangle,\ c^\dagger_{A1}c^\dagger_{B2}|0\rangle=|\downarrow \uparrow\rangle.
\ee
These are the states that enter in the evaluation of the topological properties. The pseudo-spin magnetic structure around the $K'$ point is therefore related to the reduced wave-function $\frac{1}{\sqrt{2}}(c^\dagger_{A1}c^\dagger_{B2}+c^\dagger_{A2}c^\dagger_{B1})|0\rangle$ in $|\psi_g\rangle$ which corresponds then to the same entangled state as for the two spheres around the south pole. The topological properties of this semimetal can then be described through Eqs. (\ref{Cj}) and (\ref{eq:chern_mag}).} Thus, through the magnetization related to the particle densities on the two sublattices of each layer at the $K$ and $K'$ points, we introduce the lattice version of 
$\mathcal{C}^j$ (Eq.~\eqref{eq:chern_mag})
\begin{eqnarray}
\label{tilde}
\tilde{\mathcal{C}}^j&=&\frac{1}{2}\langle n^j_{KB}-n^j_{KA}-n^j_{K'B}+n^j_{K'A}\rangle\\
&=&\begin{cases}
1 & r<r_c^-\\
1/2 &r_c^-< r<r_c^+\\
0 & r>r_c^+,\\
\end{cases}
\end{eqnarray}
where $j=1,2$ refers to the layer basis. The magnetization for a single layer is shown over the unit cell of the reciprocal lattice in Fig.~2b.

\begin{figure}
	\centering
	\includegraphics[width=0.5\columnwidth]{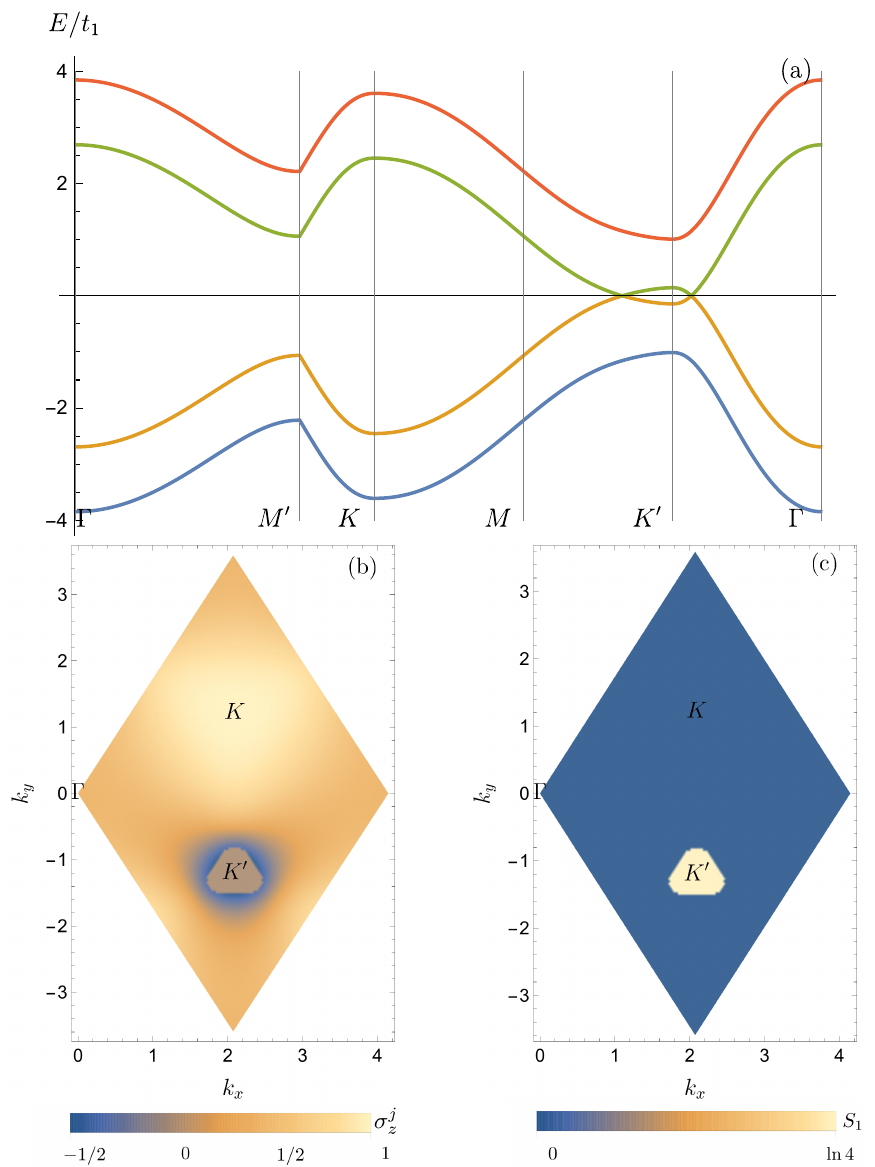}
\caption{a) Two stacked Haldane layers with $|d_z|=\sqrt{3}t_1$, $r=\sqrt{3}t_1/3$ (which is in the range $r_c^-<r<r_c^+$) for the case of symmetric Semenoff masses $M_1=M_2=3\sqrt{3}t_1/4$. Here $t_1$ is the nearest neighbour hopping amplitude. a) Band structure. b) Magnetization $\sigma_z^j\equiv\langle n_{\b{k}B}^j-n_{\b{k}A}^j\rangle$ in one layer over the primitive cell of the reciprocal lattice (at the $K$ point, the light yellow color refers to magnetization $1$, the $K'$ point has a magnetization equal to 0 as a result of entanglement. In the orange region, the magnetization smoothly evolves from $1$ to $-1/2$ in blue). c) Entanglement entropy in the same region (yellow refers to a maximum entropy of $\ln 4$ and in the blue region the entropy is 0).}\label{fig:symm}
\end{figure}
Alternatively, we may represent the ground state at half-filling in terms of the occupancy in each layer (comprising two sub-lattices with a given ket $|ij\rangle$, $i+j=1$, such that $|10\rangle$ refers to sublattice $A$ occupancy and $|01\rangle$ to sublattice $B$ occupancy respectively): $|\psi\rangle=\sum_{i+j+k+l=2}c_{ijkl}|ij\rangle_1|kl\rangle_2$,
from which we get the reduced density matrix $\rho_1$ by tracing out the second layer. 
From this the entanglement entropy
is computed numerically (see the Supplementary Information) and shown for the case of symmetric masses in Fig.~2c. For $r<r_c^-$, the entanglement entropy is identically zero. \noteKLH{ Above $r_c^-$, we verify that the system shows a maximum entanglement entropy of $\ln4$ located in the band crossing region, in agreement with the form of $|\psi_g\rangle$.} One Dirac point is characterized by a nodal ring enclosing the entangled region. Since the two Dirac points map to the two poles on the spheres, this emphasizes the correspondence between the two-spins and the lattice model. 

We highlight here that even though we have a band crossing effect in the nodal ring region, the spheres' formalism allows us to conclude that the topological number defined through Eq. (\ref{Cj}) is yet applicable in this situation showing then that  $\mathcal{C}^j=1/2$  is measurable through the quantum Hall conductivity, with $j$ referring to one layer. From Stokes' theorem on the sphere, it is important to emphasize here that the $\mathcal{C}^j=1/2$ topological number can also be interpreted as a $\pi$ Berry phase encircling just one Dirac point associated to the topology (Eq. (33) of Supplementary Information).
Regarding the bulk-edge correspondence, the edge spectrum in the reciprocal space produces one chiral edge mode as in the quantum Hall effect\cite{Pepper,Halperin,Buttiker} and in the Haldane model\cite{haldane1988}.
We study the edge states of this model in real space using the KWANT code\cite{groth2014} and show that for $M_1=M_2$ this mode is equally distributed between the two planes at the edges as if a charge $e$ in the reciprocal space redistributes as two $e/2$ effective charges in real space, in agreement with the quantum Hall conductivity for $M_1=M_2$. When we progressively deviate from the line $M_1=M_2$, navigating in the blue $C=1$ region of Fig. 1b), then this mode progressively redistributes in one plane only. The nodal ring gives rise to delocalized bulk gapless modes in real space, and yet the robust topology can also be measured from the particles' densities associated to each layer resolved in momentum space at the two Dirac points from Eq. \eqref{tilde}. 
We also find that the layer magnetization number $\tilde{\mathcal{C}}^j$ varies smoothly across the transition, in contrast to the sharp change in $\mathcal{C}^j$ that occurred in the spin model. See Supplementary Information, for further details related to these facts and proofs.

\section{Protocol in Time}\label{sec:timedepend}

Here, we show the occurrence of stable half-topological numbers in a real-time protocol, in the adiabatic limit. We also illustrate energy bands interferometry effects and deviations from these rational values when increasing the speed of the protocol. 
One experimental protocol for measuring $\mathcal{C}^j$ in a spin system is to perform a linear sweep, $\theta=vt$, $t\in[0,\pi/v]$ for some velocity $v$, of the magnetic field along the meridian $\phi=0$, measuring $\langle\sigma^j_z\rangle$ at the endpoints of the path\cite{roushan2014}, i.e. at the north and south poles. Any finite velocity will lead to non-adiabatic transitions via the Landau-Zener-Majorana mechanism\cite{zener1932,landau1932,majorana1932}, which describes a time-dependent two-state model of the form $\mathcal{H}=\lambda t\sigma_z+\Delta\sigma_x$.
The amplitudes for the $|\uparrow\rangle$ and $|\downarrow\rangle$ components of the wavefunction were derived by Zener\cite{zener1932} for the asymptotic case $t\rightarrow\infty$. Here we are actually interested in the values at $t=0$, which are derived in the Supplementary Information. There we also show that the quasi-adiabatic regime of our two-spin system is described by an effective two-state Hamiltonian
\begin{eqnarray}
\mathcal{H}^+_{\rm eff}&=&-[\tilde{r}f(\theta)+H\cos\theta+M]\sigma^z-\sqrt{2}H\sin\theta\sigma_x+(H\cos\theta+M)\mathbb{I},\label{eq:Heff}
\end{eqnarray}
where the basis for the Pauli matrices is now given by two of the triplet states $(1,0)^T=|1,0\rangle$ and $(0,1)^T=|1,-1\rangle$.
We see that the entangled state $|1,0\rangle$ is indeed the unique ground state at $\theta=\pi$ for $\tilde{r}$ sufficiently large. More precisely, the window in which the ground state evolves from the product state at the north pole to the entangled state at the south pole, and therefore has $\mathcal{C}^j=1/2$, is given by 
\be
\frac{H-M}{f(\pi)}<\tilde{r}<\frac{H+M}{f(0)}.
\ee

\begin{figure}
	\centering
	\includegraphics[width=0.5\columnwidth]{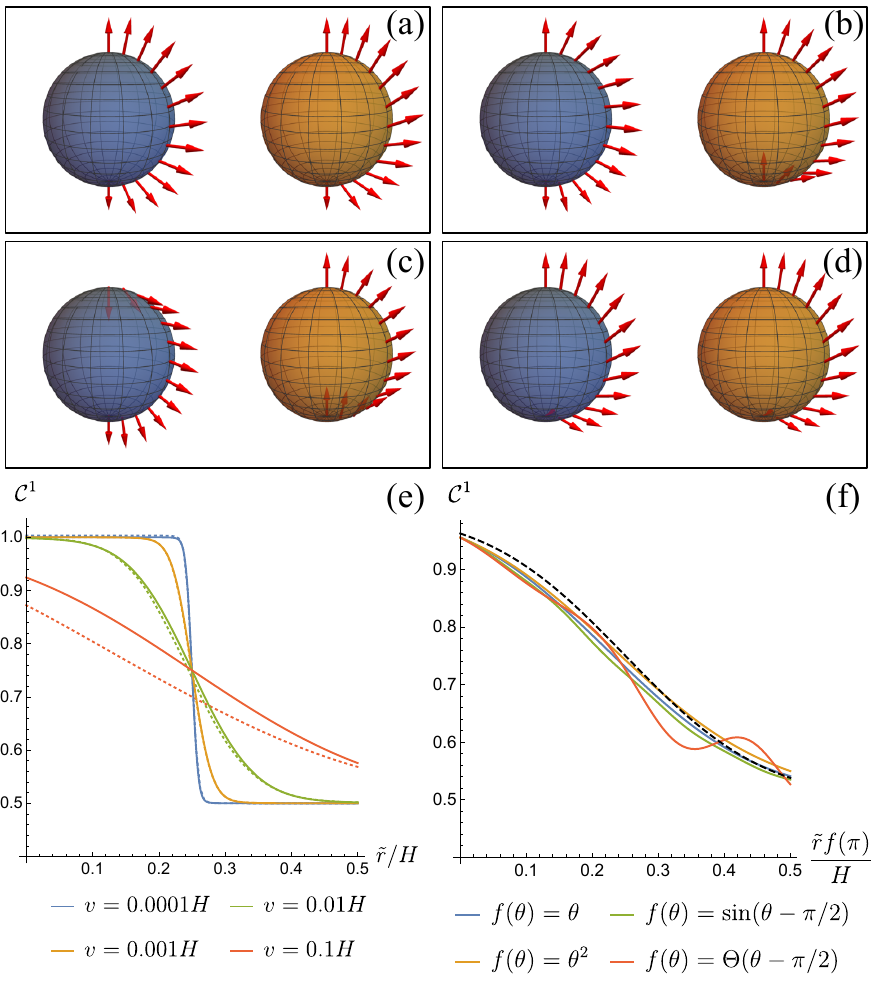}
\caption{a)-d) Spin responses $\langle \b{\sigma}^1\rangle$ (on the blue sphere) and  $\langle \b{\sigma}^2\rangle$ (on the orange sphere) to a sweep protocol of the radial applied magnetic field along a meridian of the sphere with $v=0.0001H$. The time-dependent spin vector shown in red (measured in units of $H$) is determined from the numerical solution of the Schr\"odinger equation. The radius of each sphere is $H$. a)-c) show an asymmetric case with $M_1=H/3$, $M_2=H/2$. a) $\tilde{r}=0.25H$. b) $\tilde{r}=0.9H$. c) $\tilde{r}=1.7H$. d-f) shows the symmetric mass case with $M_1=M_2=3H/4$. d) Spin response for $\tilde{r}=H/3$. In this case, the magnitude of the spin vector vanishes at the south pole. e) Chern number of a single spin versus $\tilde{r}/H$ for different sweep velocities with $f(\theta)=1$. The solid lines show the analytic approximation of Eq.~\eqref{eq:ChernAnalyt}, while the dotted lines show the result from the numerical solution to the Schr\"odinger equation. (f) Numerically determined Chern number of a single spin vs $\tilde{r}f(\pi)/H$ shown by the solid lines for different interactions with $v=0.05H$; $\Theta$ refers to the Heaviside step function. The dashed black line shows the analytic approximation of Eq.~\eqref{eq:ChernAnalyt} which is universal for a given speed $v$.}\label{fig:timeevolve}
\end{figure}


Returning to the dynamics of Eq. \eqref{eq:Heff}, we expand near $\theta=\pi$, such that $t\rightarrow t-\pi/v$. With this new time variable, the important dynamics takes place near $t=0$ such that we approximate $f(\theta)=f(\pi)$ close to the south pole, but we find that relaxing this condition does not affect the result noticeably, as shown in Fig.~3f. We then rotate the Pauli matrices about the $y$-axis. In the rotated basis, the effective Hamiltonian takes the Landau-Zener form, with
\be
\lambda\equiv\sqrt{2}Hv,\;\;\Delta=\tilde{r}f(\pi)-H+M,
\ee
and adiabaticity parameter $\gamma=\Delta^2/\lambda$.
The amplitude for measuring the $|1,-1\rangle$ state is then $\frac{1}{\sqrt{2}}(A(t)-B(t))$, while the amplitude for measuring the entangled state is $\frac{1}{\sqrt{2}}(A(t)+B(t))$. The former results in $\mathcal{C}^j=1$, upon sweeping to the south pole (now at $t=0$) while the latter gives $\mathcal{C}^j=1/2$. The value of $\mathcal{C}^j$ is then related to the coefficient $A(0)$ and $B(0)$ through
\begin{eqnarray}
\mathcal{C}^{j}&\approx&\frac{3}{4}-\frac{1}{4}{\rm Re}(A(0)B^*(0)).
\end{eqnarray}
The product $A(0)B(0)^*$ is evaluated in the Supplementary Information, which yields
\be
\mathcal{C}^{j}\approx\frac{3}{4}+\frac{\pi}{4}{\rm Re}\left(e^{i3\pi/4}e^{-\gamma\pi/4}\frac{\sgn(\Delta)\sqrt{\gamma}}{\Gamma(1/2+i\gamma/4)\Gamma(1-i\gamma/4)}\right),\label{eq:ChernAnalyt}
\ee
in terms of the gamma function $\Gamma(z)$.
We check the adiabatic limit of this formula, $v\rightarrow0$ ($\gamma\rightarrow\infty$) and find
\be
\mathcal{C}^{j}\rightarrow\frac{3}{4}-\frac{1}{4}\sgn(\Delta),
\ee
which gives $1$ for $\tilde{r}<(H-M)/f(\pi)$ and $1/2$ for $\tilde{r}>(H-M)/f(\pi)$. We also study numerically the time evolution of the interacting spins in this protocol (Fig.~3a-d). Our analytic result for $\mathcal{C}^j$ is then compared with the corresponding numerical value in Fig.~3e-f. We see that this formula accurately captures the transition in $\mathcal{C}^j$ for small sweep velocities. 

We also find, by checking many examples, that the shape of the transition is independent of the particular form of time-dependant interaction $f(\theta)$, which for small sweep velocities only shifts the transition point. This is shown in Fig.~3.f where we compare the analytic approximation to the numerical solution of the Schr\"odinger equation for a variety of interactions.  

We also find, by checking many examples, that the shape of the transition is independent of the particular form of time-dependant interaction $f(\theta)$, which for small sweep velocities only shifts the transition point. This is shown in Fig.~3.f where we compare the analytic approximation to the numerical solution of the Schr\"odinger equation for a variety of interactions. \noteKLH{ The fact that the $\mathcal{C}^j$ values are robust to such changes in the form of the interactions is a result of the topological nature of the quantum system.}

\noteKLH{ In addition, we have verified that the additional fractional phases found for $N>2$ spins are also stable from the time evolution of these models in the quantum circuit simulator Cirq\cite{Cirq}, illustrating that these phases can indeed be seen through the action of unitary gates in a generic quantum computer; See Supplementary Information.}

\section{Discussion}

Our analysis shows that one can realize quantum states with fractional topology from the interplay between Berry curvatures and resonating valence bond states\cite{Anderson,haldane1988,KalmeyerLaughlin}. Quantum entanglement between two spins can produce a Chern number of one-half for each spin. We have provided a geometrical and physical interpretation of this result through the derivation of Eq. (\ref{Cj}). We have shown the stability of the fractional Chern number regarding various forms of interactions in the adiabatic limit. We have formulated a correspondence with topological lattice models respecting ${\mathbb Z}_2$ (layer) symmetry, which form nodal ring semimetals in momentum space around a Dirac point. The one-half topological number arises from a $\pi$ Berry phase around the Dirac point that shows the topological band gap and also reveals one protected low-energy edge mode in the reciprocal space. This prediction can be measured from momentum-resolved tunneling i.e. when injecting a charge $e$ resolved in energy and wave-vector\cite{Steinberg}.  In real space, we verify that this mode equally redistributes between the two planes with 1/2 probabilities as if a charge $e$ would equilibrate as two averaged charges $e/2$ in the two layers. It is important to highlight here that for $M_1=M_2$, in the presence of a band-crossing effect around the nodal semimetallic ring, we have shown that the one-half topology of each spin or each plane in the bilayer model can be defined from the spin magnetizations at the poles, the bulk charge polarization and the quantum Hall conductivity which can also be reinterpreted as an effective charge $e/2$. Since the ground state wavefunction is a direct product state on the sphere, defining the operator $\hat{C}^j =1/2(\sigma_{j}^z(0) - \sigma_{j}^z(\pi))$, we obtain the standard
deviation $e\delta C^j = e\sqrt{\langle (\hat{C}^j)^2\rangle - \langle\hat{C}^j \rangle^2}=(e/2)\sqrt{F} =e/2$ and $F=\langle \sigma_z^j(\pi)^2 \rangle-\langle \sigma^j_z(\pi)\rangle^2=1$, which is a result of the formation of an entangled Bell pair at the south pole. Related to circular dichroism of light, we have verified that at the topological Dirac point the response is similar to the Haldane model\cite{Philipp,Nathan} and that in the semimetallic region there is no light response, such that when averaging on both light polarizations the response  at the two Dirac points is also in agreement with a one-half topological number.  Increasing the number of spins can give access to other rational topological numbers as well, \noteKLH{ in relation with various forms of entangled states}. These predictions can be measured with actual developments on quantum systems, entanglement and light-matter coupling. The interpretation of this phase needs to be further studied in relation with the classification table\cite{bernevig2013,young2015}, as well as interaction effects on the lattice directly from the reciprocal space\cite{Philipp,morimoto2016,meng2019}. These spheres' models  may also find applications as light emitters through a quantum dynamo effect\cite{henriet2017}, many-body synchronization sources\cite{sync}, and can be generalized to superconducting systems through the Nambu basis and in networks similar to the Affleck-Kennedy-Lieb-Tasaki architecture\cite{AKLT} for quantum algorithms purposes.

\bibliography{bibliography}

\begin{thebibliography}{10}
\expandafter\ifx\csname url\endcsname\relax
  \def\url#1{\texttt{#1}}\fi
\expandafter\ifx\csname urlprefix\endcsname\relax\def\urlprefix{URL }\fi
\providecommand{\bibinfo}[2]{#2}
\providecommand{\eprint}[2][]{\url{#2}}

\bibitem{HasanKane}
\bibinfo{author}{Hasan, Z.} \& \bibinfo{author}{Kane, C.~L.}
\newblock \bibinfo{title}{Colloquium:topological insulators}.
\newblock \emph{\bibinfo{journal}{Rev. Mod. Phys.}}
  \textbf{\bibinfo{volume}{82}}, \bibinfo{pages}{3045} (\bibinfo{year}{2010}).

\bibitem{Zhang}
\bibinfo{author}{Liang~Qi, X.} \& \bibinfo{author}{Zhang, S.}
\newblock \bibinfo{title}{Topological insulators and superconductors}.
\newblock \emph{\bibinfo{journal}{Rev. Mod. Phys.}}
  \textbf{\bibinfo{volume}{83}}, \bibinfo{pages}{1057} (\bibinfo{year}{2011}).

\bibitem{schleich2011}
\bibinfo{author}{Schleich, W.~P.}
\newblock \emph{\bibinfo{title}{Quantum optics in phase space}}
  (\bibinfo{publisher}{John Wiley \& Sons}, \bibinfo{year}{2011}).

\bibitem{henriet2017}
\bibinfo{author}{Henriet, L.}, \bibinfo{author}{Sclocchi, A.},
  \bibinfo{author}{Orth, P.~P.} \& \bibinfo{author}{Le~Hur, K.}
\newblock \bibinfo{title}{Topology of a dissipative spin: Dynamical chern
  number, bath-induced nonadiabaticity, and a quantum dynamo effect}.
\newblock \emph{\bibinfo{journal}{Phys. Rev. B}} \textbf{\bibinfo{volume}{95}},
  \bibinfo{pages}{054307} (\bibinfo{year}{2017}).

\bibitem{gritsev2012}
\bibinfo{author}{Gritsev, V.} \& \bibinfo{author}{Polkovnikov, A.}
\newblock \bibinfo{title}{Dynamical quantum hall effect in the parameter
  space}.
\newblock \emph{\bibinfo{journal}{Proceedings of the National Academy of
  Sciences}} \textbf{\bibinfo{volume}{109}}, \bibinfo{pages}{6457--6462}
  (\bibinfo{year}{2012}).

\bibitem{degrandi2010}
\bibinfo{author}{De~Grandi, C.} \& \bibinfo{author}{Polkovnikov, A.}
\newblock \emph{\bibinfo{title}{Adiabatic Perturbation Theory: From
  Landau--Zener Problem to Quenching Through a Quantum Critical Point}},
  \bibinfo{pages}{75--114} (\bibinfo{publisher}{Springer Berlin Heidelberg},
  \bibinfo{address}{Berlin, Heidelberg}, \bibinfo{year}{2010}).

\bibitem{berry1984}
\bibinfo{author}{Berry, M.~V.}
\newblock \bibinfo{title}{Quantal phase factors accompanying adiabatic
  changes}.
\newblock \emph{\bibinfo{journal}{Proceedings of the Royal Society of London.
  A. Mathematical and Physical Sciences}} \textbf{\bibinfo{volume}{392}},
  \bibinfo{pages}{45--57} (\bibinfo{year}{1984}).

\bibitem{leek1889}
\bibinfo{author}{Leek, P.~J.} \emph{et~al.}
\newblock \bibinfo{title}{Observation of berry{\textquoteright}s phase in a
  solid-state qubit}.
\newblock \emph{\bibinfo{journal}{Science}} \textbf{\bibinfo{volume}{318}},
  \bibinfo{pages}{1889--1892} (\bibinfo{year}{2007}).

\bibitem{Thouless}
\bibinfo{author}{Thouless, D.}, \bibinfo{author}{Kohmoto, M.},
  \bibinfo{author}{Nightingale, M.~P.} \& \bibinfo{author}{den Nijs, M.}
\newblock \bibinfo{title}{{Quantized Hall Conductance in a Two-Dimensional
  Periodic Potential}}.
\newblock \emph{\bibinfo{journal}{Phys. Rev. Lett.}}
  \textbf{\bibinfo{volume}{49}}, \bibinfo{pages}{405} (\bibinfo{year}{1982}).

\bibitem{Haldanecurvature}
\bibinfo{author}{Haldane, F. D.~M.}
\newblock \bibinfo{title}{Geometrical description of the fractional quantum
  hall effect}.
\newblock \emph{\bibinfo{journal}{Phys. Rev. Lett.}}
  \textbf{\bibinfo{volume}{107}}, \bibinfo{pages}{116801}
  (\bibinfo{year}{2011}).

\bibitem{schroer2014}
\bibinfo{author}{Schroer, M.~D.} \emph{et~al.}
\newblock \bibinfo{title}{Measuring a topological transition in an artificial
  spin-$1/2$ system}.
\newblock \emph{\bibinfo{journal}{Phys. Rev. Lett.}}
  \textbf{\bibinfo{volume}{113}}, \bibinfo{pages}{050402}
  (\bibinfo{year}{2014}).

\bibitem{roushan2014}
\bibinfo{author}{Roushan, P.} \emph{et~al.}
\newblock \bibinfo{title}{Observation of topological transitions in interacting
  quantum circuits}.
\newblock \emph{\bibinfo{journal}{Nature}} \textbf{\bibinfo{volume}{515}},
  \bibinfo{pages}{241 EP --} (\bibinfo{year}{2014}).

\bibitem{korber2020}
\bibinfo{author}{{K{\"o}rber}, S.}, \bibinfo{author}{{Privitera}, L.},
  \bibinfo{author}{{Budich}, J.~C.} \& \bibinfo{author}{{Trauzettel}, B.}
\newblock \bibinfo{title}{{Interacting topological frequency converter}}.
\newblock \emph{\bibinfo{journal}{Phys. Rev. Research}}
  \textbf{\bibinfo{volume}{2}}, \bibinfo{pages}{022023} (\bibinfo{year}{2020}).

\bibitem{bell1964}
\bibinfo{author}{Bell, J.~S.}
\newblock \bibinfo{title}{{On the Einstein Podolsky Rosen paradox}}.
\newblock \emph{\bibinfo{journal}{Physics}} \textbf{\bibinfo{volume}{1}},
  \bibinfo{pages}{195--200} (\bibinfo{year}{1964}).

\bibitem{Philipp}
\bibinfo{author}{Klein, P.}, \bibinfo{author}{Grushin, A.} \&
  \bibinfo{author}{Le~Hur, K.}
\newblock \bibinfo{title}{Interacting stochastic topology and mott transition
  from light response}.
\newblock \emph{\bibinfo{journal}{Phys. Rev. B}}
  \textbf{\bibinfo{volume}{103}}, \bibinfo{pages}{035114}
  (\bibinfo{year}{2021}).

\bibitem{Nathan}
\bibinfo{author}{Tran, D.~T.}, \bibinfo{author}{Dauphin, A.},
  \bibinfo{author}{Grushin, A.~G.}, \bibinfo{author}{Zoller, P.} \&
  \bibinfo{author}{Goldman, N.}
\newblock \bibinfo{title}{Probing topology by "heating": Quantized circular
  dichroism in ultracold atoms}.
\newblock \emph{\bibinfo{journal}{Sciences Advances}}
  \textbf{\bibinfo{volume}{3}}, \bibinfo{pages}{e1701207}
  (\bibinfo{year}{2017}).

\bibitem{Asteria}
\bibinfo{author}{Asteria, L.} \emph{et~al.}
\newblock \bibinfo{title}{Measuring quantized circular dichroism in ultracold
  topological matter}.
\newblock \emph{\bibinfo{journal}{Nature Physics}}
  \textbf{\bibinfo{volume}{15}}, \bibinfo{pages}{449} (\bibinfo{year}{2017}).

\bibitem{entropy}
\bibinfo{author}{Neill, C.~{\it et al}.}
\newblock \bibinfo{title}{{Ergodic dynamics and thermalization in an isolated
  quantum system}}.
\newblock \emph{\bibinfo{journal}{Nature Physics}}
  \textbf{\bibinfo{volume}{1}}, \bibinfo{pages}{1037--1041}
  (\bibinfo{year}{2016}).

\bibitem{Anderson}
\bibinfo{author}{Anderson, P.~W.}
\newblock \bibinfo{title}{Resonating valence bonds: A new kind of insulator?}
\newblock \emph{\bibinfo{journal}{Materials Research Bulletin}}
  \textbf{\bibinfo{volume}{8}}, \bibinfo{pages}{153--160}
  (\bibinfo{year}{1973}).

\bibitem{lehur2017}
\bibinfo{author}{Le~Hur, K.}, \bibinfo{author}{Soret, A.} \&
  \bibinfo{author}{Yang, F.}
\newblock \bibinfo{title}{Majorana spin liquids, topology, and
  superconductivity in ladders}.
\newblock \emph{\bibinfo{journal}{Phys. Rev. B}} \textbf{\bibinfo{volume}{96}},
  \bibinfo{pages}{205109} (\bibinfo{year}{2017}).

\bibitem{bernevig2013}
\bibinfo{author}{Bernevig, B.~A.} \& \bibinfo{author}{Hughes, T.~L.}
\newblock \emph{\bibinfo{title}{Topological insulators and topological
  superconductors}} (\bibinfo{publisher}{Princeton university press},
  \bibinfo{year}{2013}).

\bibitem{haldane1988}
\bibinfo{author}{Haldane, F. D.~M.}
\newblock \bibinfo{title}{Model for a quantum hall effect without landau
  levels: Condensed-matter realization of the "parity anomaly"}.
\newblock \emph{\bibinfo{journal}{Phys. Rev. Lett.}}
  \textbf{\bibinfo{volume}{61}}, \bibinfo{pages}{2015--2018}
  (\bibinfo{year}{1988}).

\bibitem{liu_2015}
\bibinfo{author}{Liu, C.-X.}, \bibinfo{author}{Zhang, S.-C.} \&
  \bibinfo{author}{Qi, X.-L.}
\newblock \bibinfo{title}{The quantum anomalous hall effect}.
\newblock \emph{\bibinfo{journal}{Annual Review of Condensed Matter Physics}}
  \textbf{\bibinfo{volume}{7}}, \bibinfo{pages}{301--321}
  (\bibinfo{year}{2016}).

\bibitem{McIver_2019}
\bibinfo{author}{McIver, J.~W.} \emph{et~al.}
\newblock \bibinfo{title}{Light-induced anomalous hall effect in graphene}.
\newblock \emph{\bibinfo{journal}{Nature Physics}}  (\bibinfo{year}{2019}).

\bibitem{jotzu_2014}
\bibinfo{author}{Jotzu, G.} \emph{et~al.}
\newblock \bibinfo{title}{Experimental realization of the topological haldane
  model with ultracold fermions}.
\newblock \emph{\bibinfo{journal}{Nature}} \textbf{\bibinfo{volume}{515}},
  \bibinfo{pages}{237--240} (\bibinfo{year}{2014}).

\bibitem{flaschner_2016}
\bibinfo{author}{Flaschner, N.} \emph{et~al.}
\newblock \bibinfo{title}{Experimental reconstruction of the berry curvature in
  a floquet bloch band}.
\newblock \emph{\bibinfo{journal}{Science}} \textbf{\bibinfo{volume}{352}},
  \bibinfo{pages}{1091--1094} (\bibinfo{year}{2016}).

\bibitem{haldane_2008}
\bibinfo{author}{Haldane, F. D.~M.} \& \bibinfo{author}{Raghu, S.}
\newblock \bibinfo{title}{Possible realization of directional optical
  waveguides in photonic crystals with broken time-reversal symmetry}.
\newblock \emph{\bibinfo{journal}{Phys. Rev. Lett.}}
  \textbf{\bibinfo{volume}{100}}, \bibinfo{pages}{013904}
  (\bibinfo{year}{2008}).

\bibitem{Lu_2014}
\bibinfo{author}{Lu, L.}, \bibinfo{author}{Joannopoulos, J.~D.} \&
  \bibinfo{author}{Soljacic, M.}
\newblock \bibinfo{title}{Topological photonics}.
\newblock \emph{\bibinfo{journal}{Nature Photonics}}
  \textbf{\bibinfo{volume}{8}}, \bibinfo{pages}{821--829}
  (\bibinfo{year}{2014}).

\bibitem{koch_2010}
\bibinfo{author}{Koch, J.}, \bibinfo{author}{Houck, A.~A.},
  \bibinfo{author}{Le~Hur, K.} \& \bibinfo{author}{Girvin, S.~M.}
\newblock \bibinfo{title}{Time-reversal-symmetry breaking in circuit-qed-based
  photon lattices}.
\newblock \emph{\bibinfo{journal}{Phys. Rev. A}} \textbf{\bibinfo{volume}{82}},
  \bibinfo{pages}{043811} (\bibinfo{year}{2010}).

\bibitem{Le_Hur_2016}
\bibinfo{author}{Le~Hur, K.} \emph{et~al.}
\newblock \bibinfo{title}{Many-body quantum electrodynamics networks:
  Non-equilibrium condensed matter physics with light}.
\newblock \emph{\bibinfo{journal}{Comptes Rendus Physique}}
  \textbf{\bibinfo{volume}{17}}, \bibinfo{pages}{808--835}
  (\bibinfo{year}{2016}).

\bibitem{Tomoki_2019}
\bibinfo{author}{Ozawa, T.} \emph{et~al.}
\newblock \bibinfo{title}{Topological photonics}.
\newblock \emph{\bibinfo{journal}{Rev. Mod. Phys.}}
  \textbf{\bibinfo{volume}{91}}, \bibinfo{pages}{015006}
  (\bibinfo{year}{2019}).

\bibitem{cheng2019}
\bibinfo{author}{Cheng, P.} \emph{et~al.}
\newblock \bibinfo{title}{Topological proximity effects in a haldane graphene
  bilayer system}.
\newblock \emph{\bibinfo{journal}{Phys. Rev. B}}
  \textbf{\bibinfo{volume}{100}}, \bibinfo{pages}{081107}
  (\bibinfo{year}{2019}).

\bibitem{semenoff1984}
\bibinfo{author}{Semenoff, G.~W.}
\newblock \bibinfo{title}{Condensed-matter simulation of a three-dimensional
  anomaly}.
\newblock \emph{\bibinfo{journal}{Phys. Rev. Lett.}}
  \textbf{\bibinfo{volume}{53}}, \bibinfo{pages}{2449--2452}
  (\bibinfo{year}{1984}).

\bibitem{KM}
\bibinfo{author}{Kane, C.~L.} \& \bibinfo{author}{Mele, E.}
\newblock \bibinfo{title}{Quantum spin hall effect in graphene}.
\newblock \emph{\bibinfo{journal}{Phys. Rev. Lett.}}
  \textbf{\bibinfo{volume}{95}}, \bibinfo{pages}{226801}
  (\bibinfo{year}{2005}).

\bibitem{morimoto2016}
\bibinfo{author}{Morimoto, T.} \& \bibinfo{author}{Nagaosa, N.}
\newblock \bibinfo{title}{Weyl mott insulator}.
\newblock \emph{\bibinfo{journal}{Scientific Reports}}
  \textbf{\bibinfo{volume}{6}}, \bibinfo{pages}{19853} (\bibinfo{year}{2016}).

\bibitem{meng2019}
\bibinfo{author}{Meng, T.} \& \bibinfo{author}{Budich, J.~C.}
\newblock \bibinfo{title}{Unpaired weyl nodes from long-ranged interactions:
  Fate of quantum anomalies}.
\newblock \emph{\bibinfo{journal}{Phys. Rev. Lett.}}
  \textbf{\bibinfo{volume}{122}}, \bibinfo{pages}{046402}
  (\bibinfo{year}{2019}).

\bibitem{fukui2005}
\bibinfo{author}{Fukui, T.}, \bibinfo{author}{Hatsugai, Y.} \&
  \bibinfo{author}{Suzuki, H.}
\newblock \bibinfo{title}{Chern numbers in discretized brillouin zone:
  Efficient method of computing (spin) hall conductances}.
\newblock \emph{\bibinfo{journal}{Journal of the Physical Society of Japan}}
  \textbf{\bibinfo{volume}{74}}, \bibinfo{pages}{1674--1677}
  (\bibinfo{year}{2005}).

\bibitem{young2015}
\bibinfo{author}{Young, S.~M.} \& \bibinfo{author}{Kane, C.~L.}
\newblock \bibinfo{title}{Dirac semimetals in two dimensions}.
\newblock \emph{\bibinfo{journal}{Phys. Rev. Lett.}}
  \textbf{\bibinfo{volume}{115}}, \bibinfo{pages}{126803}
  (\bibinfo{year}{2015}).

\bibitem{Pepper}
\bibinfo{author}{Klitzing, K.~v.}, \bibinfo{author}{Dorda, G.} \&
  \bibinfo{author}{Pepper, M.}
\newblock \bibinfo{title}{New method for high-accuracy determination of the
  fine-structure constant based on quantized hall resistance}.
\newblock \emph{\bibinfo{journal}{Phys. Rev. Lett.}}
  \textbf{\bibinfo{volume}{45}}, \bibinfo{pages}{494--497}
  (\bibinfo{year}{1980}).

\bibitem{Halperin}
\bibinfo{author}{Halperin, B.~I.}
\newblock \bibinfo{title}{Quantized hall conductance, current-carrying edge
  states, and the existence of extended states in a two-dimensional disordered
  potential}.
\newblock \emph{\bibinfo{journal}{Phys. Rev. B}} \textbf{\bibinfo{volume}{25}},
  \bibinfo{pages}{2185} (\bibinfo{year}{1982}).

\bibitem{Buttiker}
\bibinfo{author}{B\"uttiker, M.}
\newblock \bibinfo{title}{Absence of backscattering in the quantum hall effect
  in multiprobe conductors}.
\newblock \emph{\bibinfo{journal}{Phys. Rev. B}} \textbf{\bibinfo{volume}{38}},
  \bibinfo{pages}{9375} (\bibinfo{year}{1988}).

\bibitem{groth2014}
\bibinfo{author}{Groth, C.~W.}, \bibinfo{author}{Wimmer, M.},
  \bibinfo{author}{Akhmerov, A.~R.} \& \bibinfo{author}{Waintal, X.}
\newblock \bibinfo{title}{Kwant: a software package for quantum transport}.
\newblock \emph{\bibinfo{journal}{New Journal of Physics}}
  \textbf{\bibinfo{volume}{16}}, \bibinfo{pages}{063065}
  (\bibinfo{year}{2014}).

\bibitem{zener1932}
\bibinfo{author}{Zener, C.} \& \bibinfo{author}{Fowler, R.~H.}
\newblock \bibinfo{title}{Non-adiabatic crossing of energy levels}.
\newblock \emph{\bibinfo{journal}{Proceedings of the Royal Society of London.
  Series A, Containing Papers of a Mathematical and Physical Character}}
  \textbf{\bibinfo{volume}{137}}, \bibinfo{pages}{696--702}
  (\bibinfo{year}{1932}).

\bibitem{landau1932}
\bibinfo{author}{Landau, L.}
\newblock \bibinfo{title}{Zur theorie der energieubertragung i}.
\newblock \emph{\bibinfo{journal}{Z. Sowjetunion}}
  \textbf{\bibinfo{volume}{1}}, \bibinfo{pages}{88--95} (\bibinfo{year}{1932}).

\bibitem{majorana1932}
\bibinfo{author}{Majorana, E.}
\newblock \bibinfo{title}{Atomi orientati in campo magnetico variabile}.
\newblock \emph{\bibinfo{journal}{Il Nuovo Cimento (1924-1942)}}
  \textbf{\bibinfo{volume}{9}}, \bibinfo{pages}{43--50} (\bibinfo{year}{1932}).

\bibitem{Cirq}
\bibinfo{author}{Contributors, T.~C.}
\newblock \bibinfo{title}{Cirq, a python framework for creating, editing, and
  invoking noisy intermediate scale quantum (nisq) circuits}
  \bibinfo{note}{\url{https://github.com/quantumlib/Cirq}}.

\bibitem{KalmeyerLaughlin}
\bibinfo{author}{Kalmeyer, V.} \& \bibinfo{author}{Laughlin, R.~B.}
\newblock \bibinfo{title}{Equivalence of the resonating-valence-bond and
  fractional quantum hall states}.
\newblock \emph{\bibinfo{journal}{Phys. Rev. Lett.}}
  \textbf{\bibinfo{volume}{59}}, \bibinfo{pages}{2095} (\bibinfo{year}{1987}).

\bibitem{Steinberg}
\bibinfo{author}{Steinberg, H.} \emph{et~al.}
\newblock \bibinfo{title}{Charge fractionalization in quantum wires}.
\newblock \emph{\bibinfo{journal}{Nature Physics}}
  \textbf{\bibinfo{volume}{4}}, \bibinfo{pages}{116--119}
  (\bibinfo{year}{2008}).

\bibitem{sync}
\bibinfo{author}{Pizzi, A.}, \bibinfo{author}{Dolcini, F.} \&
  \bibinfo{author}{Le~Hur, K.}
\newblock \bibinfo{title}{Quench-induced dynamical phase transitions and
  pi-synchronization in the bose-hubbard model}.
\newblock \emph{\bibinfo{journal}{Phys. Rev. B}} \textbf{\bibinfo{volume}{99}},
  \bibinfo{pages}{094301} (\bibinfo{year}{2019}).

\bibitem{AKLT}
\bibinfo{author}{Affleck, I.}, \bibinfo{author}{Kennedy, T.},
  \bibinfo{author}{Lieb, E.~H.} \& \bibinfo{author}{Tasaki, H.}
\newblock \bibinfo{title}{Rigorous results on valence-bond ground states in
  antiferromagnets}.
\newblock \emph{\bibinfo{journal}{Phys. Rev. Lett.}}
  \textbf{\bibinfo{volume}{59}}, \bibinfo{pages}{799} (\bibinfo{year}{1987}).

\end{thebibliography}


\begin{thebibliography}{0}%
\makeatletter
\providecommand \@ifxundefined [1]{%
 \@ifx{#1\undefined}
}%
\providecommand \@ifnum [1]{%
 \ifnum #1\expandafter \@firstoftwo
 \else \expandafter \@secondoftwo
 \fi
}%
\providecommand \@ifx [1]{%
 \ifx #1\expandafter \@firstoftwo
 \else \expandafter \@secondoftwo
 \fi
}%
\providecommand \natexlab [1]{#1}%
\providecommand \enquote  [1]{``#1''}%
\providecommand \bibnamefont  [1]{#1}%
\providecommand \bibfnamefont [1]{#1}%
\providecommand \citenamefont [1]{#1}%
\providecommand \href@noop [0]{\@secondoftwo}%
\providecommand \href [0]{\begingroup \@sanitize@url \@href}%
\providecommand \@href[1]{\@@startlink{#1}\@@href}%
\providecommand \@@href[1]{\endgroup#1\@@endlink}%
\providecommand \@sanitize@url [0]{\catcode `\\12\catcode `\$12\catcode
  `\&12\catcode `\#12\catcode `\^12\catcode `\_12\catcode `\%12\relax}%
\providecommand \@@startlink[1]{}%
\providecommand \@@endlink[0]{}%
\providecommand \url  [0]{\begingroup\@sanitize@url \@url }%
\providecommand \@url [1]{\endgroup\@href {#1}{\urlprefix }}%
\providecommand \urlprefix  [0]{URL }%
\providecommand \Eprint [0]{\href }%
\providecommand \doibase [0]{http://dx.doi.org/}%
\providecommand \selectlanguage [0]{\@gobble}%
\providecommand \bibinfo  [0]{\@secondoftwo}%
\providecommand \bibfield  [0]{\@secondoftwo}%
\providecommand \translation [1]{[#1]}%
\providecommand \BibitemOpen [0]{}%
\providecommand \bibitemStop [0]{}%
\providecommand \bibitemNoStop [0]{.\EOS\space}%
\providecommand \EOS [0]{\spacefactor3000\relax}%
\providecommand \BibitemShut  [1]{\csname bibitem#1\endcsname}%
\let\auto@bib@innerbib\@empty
\end{thebibliography}%


\begin{addendum}
 \item We acknowledge discussions with Monika Aidelsburger, Loic Henriet, Philipp Klein and Joseph Maciejko and at Cambridge. KLH also acknowledges discussions in the phys-math class PHY105 and solid-state physics class PHY552A at Ecole Polytechnique, and quantum classes both at
 Ecole Polytechnique and Yale. This work was supported jointly by the Natural Sciences and Engineering Research Council of Canada (NSERC) as well as the French ANR BOCA (JH and KLH). The research on lattices and ultra-cold atoms was funded by the Deutsche Forschungsgemeinschaft
(DFG, German Research Foundation) via Research Unit
FOR 2414 under project number 277974659 (KLH). KLH also acknowledges funding from NSF in USA, through DMR-0803200 on Entanglement Theory in Many-Body Quantum Systems.
 \item[Competing Interests] The authors declare that they have no
competing financial interests.
\item[Additional Information] Supplementary Information is available for this paper.
\item[Author details and Contributions] The two authors (joel.hutchinson@polytechnique.edu and karyn.le-hur@polytechnique.edu) have contributed to the elaboration of ideas and the establishment of results. They have also participated in the writing of the manuscript. 
 \item[Correspondence] The corresponding author is karyn.le-hur@polytechnique.edu.
 
{\bf Methods}
 The methodology begins from general quantum arguments to show the possibility of a fractional Chern number for an interacting spin-1/2 particle, leading to Eqs. (\ref{Cj}), (\ref{eq:chern_mag}) and (\ref{eq:chernhalf}). Then, we analyze
 the ground state energetics of a particular model and show how to observe a Chern number 1/2. Furthermore, we formulate a mathematical correspondence between the spin-1/2 and topological bilayer lattice models.  We find a relation between 
 the Chern number measurement and the quantum Hall conductivity, the polarization and the light response in a given plane. We perform numerical evaluations in the bilayer model of the Berry curvature, magnetization, entanglement entropy as well as the band structure in a finite and infinite system. For the time-dependent protocol, we check, through numerical evaluation of the Schr\"odinger equation, that our results are very similar for various forms of spin interaction in curved space. We also study the effect of increasing the speed of the protocol related to Landau-Zener-Majorana interferometry effects. 
 
 In the Supplementary Information, we present in the first section two proofs for the gauge invariance of Eq. (\ref{Cj}) and show from the smooth fields that it is related to a quantum Hall conductivity $\sigma_{xy}=\frac{1}{2}\frac{e^2}{h}$ on one plane and to a $\pi$ Berry phase around one Dirac point. We also discuss applications to the class of wavefunctions we study. In section 2, we consider generalized models with transverse coupling. In section 3, we study models with higher numbers of spins. In section 4, we show definitions on the Haldane model and bilayer system, develop the notations for the entanglement entropy calculation, and show the edge modes and local density of states of a ribbon geometry.  Lastly, in section 5, \noteKLH{ we present results on the time evolution of the systems. We derive the transition amplitudes for the time-dependent protocol associated to the Landau-Zener-Majorana dynamics. We also verify the possibility of other fractional topological states in time for the situation with $N>2$ spins using the Cirq algorithm \cite{Cirq}.}
 
 \end{addendum}


\end{document}


\newcommand{\tr}{\mathop{\mathrm{Tr}}}
\newcommand{\bsigma}{\boldsymbol{\sigma}}
\newcommand{\re}{\mathop{\mathrm{Re}}}
\newcommand{\im}{\mathop{\mathrm{Im}}}
\renewcommand{\b}[1]{{\boldsymbol{#1}}}
\newcommand{\diag}{\mathrm{diag}}
\newcommand{\sign}{\mathrm{sign}}
\newcommand{\sgn}{\mathop{\mathrm{sgn}}}
\renewcommand{\c}[1]{\mathcal{#1}}
\renewcommand{\d}{\text{\dj}}
\newcommand{\red}{\textcolor{red}}
\newcommand{\green}{\textcolor{green}}
\newcommand{\noteKLH}[1]{{\color{black}  #1}}

\newcommand{\mb}{\bm}
\newcommand{\ua}{\uparrow}
\newcommand{\da}{\downarrow}
\newcommand{\ra}{\rightarrow}
\newcommand{\la}{\leftarrow}
\newcommand{\mc}{\mathcal}
\newcommand{\bs}{\boldsymbol}
\newcommand{\lra}{\leftrightarrow}
\newcommand{\nn}{\nonumber}
\newcommand{\half}{{\textstyle{\frac{1}{2}}}}
\newcommand{\mf}{\mathfrak}
\newcommand{\MF}{\text{MF}}
\newcommand{\IR}{\text{IR}}
\newcommand{\UV}{\text{UV}}
\newcommand{\be}{\begin{equation}}
\newcommand{\ee}{\end{equation}}

\newtheorem{theorem}{Theorem}[section]
\newtheorem{corollary}{Corollary}[theorem]
\newtheorem{lemma}[theorem]{Lemma}

\DeclareGraphicsExtensions{.png}

\title{Supplementary Information: Quantum Entangled Fractional Topology and Curvatures}

\author{Joel Hutchinson$^{1}$ \& Karyn Le Hur$^1$ \\
CPHT, CNRS, Ecole Polytechnique, Institut Polytechnique de Paris, \\
Route de Saclay, 91128 Palaiseau, France}

\date{}

\maketitle

\begin{abstract}

In this Supplementary Information, we provide two alternative derivations of Eqs. (6) and (7) in the manuscript and through smooth fields on the sphere we derive the relations between the topological geometrical responses and measurable observables in transport properties. We also discuss the relation between the one-half topological number on each sphere (associated to a plane) and a $\pi$ Berry phase at one pole (Dirac point). We show the  gauge invariance associated to the topological response (Sec.~\ref{sec:6proof}). We then discuss generalized spin models with $XY$ coupling (Sec.~\ref{sec:anisotropic}) and larger numbers of spins (Sec.~\ref{sec:nspins}) \noteKLH{ which allow us to identify other forms of entangled spin states}. Then, we provide definitions on the lattice geometry and on the calculation of entanglement entropy, as well as results for the edge modes and local density of states  for a finite system (Sec.~\ref{app:haldane}). We also discuss the measure of the topological number with a driving protocol in time and study transition amplitudes in time (Sec.~\ref{app:LZ}). \noteKLH{ We address a time-evolution protocol for a situation with $N>2$ spins to show the stability of these fractional topological numbers for larger systems.}

\end{abstract}

\section{Gauge invariance and non-quantization of $\mathcal{C}^j$}\label{sec:6proof}

First in Sec.~\ref{sec:topo_resp}, we provide a geometrical proof of Eqs.~(6) and (7) in the article from vector theorems. This illustrates the intriguing fact that the topological response can be encoded in the poles of the Bloch sphere, which holds for both the regular Chern number $C$ and the partial Chern number $C^j$. In Sec.~\ref{sec:transport}, we show the relation between the smooth fields in Stokes' theorem and quantum transport properties related to the charge polarization and quantum Hall conductivity. 
Then, in Sec.~\ref{VPW} we provide a proof based on the specific class of wavefunctions we study. In Sec.~\ref{cases2}, we verify our arguments for various forms of wave-functions.

\subsection{Topological Response}\label{sec:topo_resp}
In the article, the Berry connection $\b{\mathcal{A}}^j$ is well-defined on the abstract parameter space $\{\theta,\phi\}$. The goal in this section is to reinterpret $\b{\mathcal{A}}^j$ as a vector on the surface of a sphere $S^2$, and then to make use of Stokes's theorem in three dimensions to evaluate the Chern number. We start from Eq.~(1) of the article:
\be
C=-\frac{1}{2 \pi} \int_{S^{2}}d^{2}\mathbf{n}\cdot(\nabla \times \b{\mathcal{A}}),\label{eq:cherncurl}
\ee
Here, $\b{\mathcal{A}}$ is computed in an arbitrary gauge, $\b{n}$ is the normal vector to the sphere, and we have used 
\be
\label{2}
\mathcal{F}_{\phi\theta}d\theta d\phi=(\partial_\phi\mathcal{A}_\theta-\partial_\theta\mathcal{A}_\phi)d\theta d\phi=-(\nabla\times\b{\mathcal{A}})\cdot d^2\b{r}.
\ee  
Note that we have dropped the $j$ superscript since this proof applies equally to the Chern number $C$ and the partial Chern number $C^j$. 

We decompose the hemisphere into north and south hemispheres demarcated by a fixed $\theta=\theta_c$ (which need not be at the equator) such that
\be
C=-\frac{1}{2\pi} \int_{\rm north}d^{2} \b{n}\cdot( \nabla\times \b{\mathcal{A}})-\frac{1}{2 \pi} \int_{\rm south} d^{2}\b{n}\cdot(\nabla \times \b{\mathcal{A}}).\label{eq:northsouth}
\ee
In mapping the space $\{\theta,\phi\}$ to $S^2$, we have to take special care of the poles because while $\b{\mathcal{A}}(0,\phi)$ and $\b{\mathcal{A}}(\pi,\phi)$ have well-defined $\phi$-components that contribute to the Chern number, any smooth vector field on $S^2$ must have vanishing $\phi$-components at the poles. The case where these components do vanish and the mapping is faithful, is precisely the case when the Chern number is trivial. The use of Stokes's theorem requires that we have a smooth smooth vector field over the relevant manifold. Here, we would like to show the form of this smooth field as well as the form of $C$ in terms of the Berry connections at the poles. We hypothesize that we can build a piecewise smooth field $\b{\mathcal{A}}'$ on the north and south hemispheres, such that
\be
\nabla\times \b{\mathcal{A}}'=\nabla\times \b{\mathcal{A}}
\ee
on each hemisphere, for all values of the azimuthal and polar angles. Now, we show the form of this field, as follows. Looking at Eq.~\eqref{eq:cherncurl}, we can subtract infinitesimally small areas encircling the poles to define a surface $S^{2'}$ which is no longer a closed manifold. Since these areas are infinitesimally small and the Berry curvature is finite, this will not affect the Chern number so we can write
\be
C=-\frac{1}{2 \pi} \int_{S^{2'}} d^{2}\mathbf{n}\cdot(\nabla \times \b{\mathcal{A}}).
\ee
The surface $S^{2'}$ can be decomposed into a north (north') hemisphere defined by $0<\theta<\theta_c$ and south (south') hemisphere defined by $\theta_c<\theta<\pi$ on which the field $\b{\mathcal{A}}$ is smooth, such that
\be
C=-\frac{1}{2 \pi} \int_{\rm north'}d^{2}\mathbf{n}\cdot(\nabla \times \b{\mathcal{A}})-\frac{1}{2 \pi} \int_{\rm south'}d^{2}\mathbf{n}\cdot(\nabla \times \b{\mathcal{A}}).
\ee
On north', we have from Stokes' theorem:
\be
-\frac{1}{2 \pi} \int_{\rm north'}d^{2}\mathbf{n}\cdot(\nabla \times \b{\mathcal{A}})=-\frac{1}{2 \pi} \int_{0}^{2 \pi} d \phi \mathcal{A}_{\phi}\left(\theta_{c}^-,\phi\right)+\frac{1}{2 \pi} \int_{0}^{2 \pi} d \phi \mathcal{A}_{\phi}(0),\label{eq:north'}
\ee
where $\theta_c^-$ means we approach $\theta_c$ from the north. This form assumes that the field is uniquely defined on the boundary path at the north pole with ${\mathcal{A}}_{\phi}(0)={\mathcal{A}}_{\phi}(0,\phi).$ The right-hand side then corresponds to the two boundary paths encircling north'. Similarly, we have for south':
\be
-\frac{1}{2 \pi} \int_{\rm south'}d^{2}\mathbf{n}\cdot(\nabla \times \b{\mathcal{A}})=+\frac{1}{2 \pi} \int_{0}^{2 \pi} d \phi \mathcal{A}_{\phi}\left(\theta_{c}^+,\phi\right)-\frac{1}{2 \pi} \int_{0}^{2 \pi} d \phi \mathcal{A}_{\phi}(\pi).\label{eq:south'}
\ee
Again, $\theta_c^+$ means we approach $\theta_c$ from the south, and the field is uniquely defined on the boundary path at the pole with $\mathcal{A}_{\phi}(\pi)=\mathcal{A}_{\phi}(\phi,\pi).$ These expressions suggest a natural definition for the smooth fields on the full north and south hemispheres:
\begin{eqnarray}
\label{cases}
\b{\mathcal{A}}'(\theta,\phi)&\equiv&\begin{cases}
\b{\mathcal{A}}(\theta,\phi)-\mathcal{A}_\phi(0)\hat{\phi} & \theta<\theta_c\\
\b{\mathcal{A}}(\theta,\phi)-\mathcal{A}_\phi(\pi)\hat{\phi} & \theta>\theta_c,
\end{cases}
\end{eqnarray}
where $\hat{\phi}$ refers to the unit vector associated with the azimuthal angle.  We satisfy $\nabla\times\b{\mathcal{A}}'=\nabla\times\b{\mathcal{A}}$ over each hemisphere, and also ensures that
\be
C=-\frac{1}{2 \pi} \int_{S^{2}}d^{2}\mathbf{n}\cdot(\nabla \times \b{\mathcal{A}}').
\ee
Note that these identities further imply that if we fix $\theta$ and perform a closed path in $\phi$ on a given hemisphere, then
\be
\oint \b{\mathcal{A}}' \cdot d\b{\ell}=\oint \b{\mathcal{A}} \cdot d\b{\ell}-\oint \b{\mathcal{A}}_{\phi}(\rm pole) \cdot d \b{\ell},
\ee
which can be viewed as the integral of a flux through the disk at fixed $\theta$. The last term then places the pole information inside a cylinder with an infinitesimally small radius inside the hemisphere.

Returning to Eqs.~\eqref{eq:north'}, and~\eqref{eq:south'}, we show the form of $C$ in terms of Berry connections as follows. Suppose we move the boundary very close to the north pole such that $\theta_{c}\rightarrow0$ (the same relation would be obtained with $\theta_{c} \rightarrow \pi$), then
\begin{eqnarray}
-\frac{1}{2 \pi} \int_{\rm north'}d^{2}\mathbf{n}\cdot(\nabla \times \b{\mathcal{A}})&=&-\frac{1}{2 \pi} \int_{0}^{2 \pi} d \phi \mathcal{A}_{\phi}\left(0\right)+\frac{1}{2 \pi} \int_{0}^{2 \pi} d \phi \mathcal{A}_{\phi}(0)=0,\\
-\frac{1}{2 \pi} \int_{\rm south'}d^{2}\mathbf{n}\cdot(\nabla \times \b{\mathcal{A}})&=&+\frac{1}{2 \pi} \int_{0}^{2 \pi} d \phi \mathcal{A}_{\phi}\left(0\right)-\frac{1}{2 \pi} \int_{0}^{2 \pi} d \phi \mathcal{A}_{\phi}(\pi).
\end{eqnarray}
Therefore, by summing these two lines, we obtain
\be
\label{CA}
C=\mathcal{A}_{\phi}(0)-\mathcal{A}_{\phi}(\pi).
\ee
This corresponds to Eq. (7) in the article which shows that the topology can be encoded through the poles. 
\subsection{Smooth Fields and Topological Observables}\label{sec:transport}

Here, we show a relation between the formulae (\ref{eq:north'})-(\ref{cases}), (\ref{CA}) and quantum transport properties associated to the quantum Hall conductivity and the charge polarization, directly from the surface of the sphere. 

We apply an electric field ${\bf E}=E{\bf e}_{x_{\parallel}}$ in real space along the axis associated to the (wave-vector) coordinate $k_{\parallel}$ in the reciprocal space. Then,  the velocity of the particle $v_{\parallel}$ satisfies  $\hbar k_{\parallel}=mv_{\parallel}$ and from Newton's equation with the force ${\bf F}=e{\bf E}$ we have $ma_{\parallel}=\hbar \dot{k}_{\parallel}=\hbar \partial k_{\parallel}/\partial t = eE$, where $m$ refers classically to the mass of a particle and ${\bf a}$ to the acceleration. Quantum mechanically, we have used the de Broglie principle. Then, we introduce the map $(k_{\parallel},k_{\perp})=f(\theta,\phi)$ such that the two Dirac points now correspond to the poles. The direction of the electric field here corresponds to a line crossing the two Dirac points, from north to south poles, e.g. at azimuthal angle $\phi=0$ and characterized through $\theta(t)$. In flat space, we can fix effectively the lattice spacing such that the (relative) distance between the poles is equal to $\pi$ as for the unit sphere. For this situation, we observe for a charge $e$ that $\dot{\theta}=(e/\hbar)E$ and $\hbar=h/(2\pi)$ being the reduced Planck constant. At time $t=0$, the particle is at the north pole and at small time $dt$, the particle has moved to the position $d\theta$, such that we have $\int_0^{\theta(t)} d\theta = \theta(t)=\int_0^t dt \dot{\theta} = e E t/\hbar$. The velocity of a particle in real space can be written as $\dot{\bf r}=(\dot{x}_{\parallel},\dot{x}_{\perp})$ and below we study one component of the velocity vector, and more precisely the {\it transverse} component to the electric field which is related to topological properties. To show the relation between the transport from the reciprocal space and the Berry fields, one can start from quantum mechanics laws. We have a simple relation between real and reciprocal or wave-vector k-space from the Parseval-Plancherel relation 
\begin{equation}
\label{ParsevalPlancherel}
 \frac{e}{T}\int_0^T dt \frac{d\langle x\rangle}{dt} = e\frac{(\langle x\rangle(T)-\langle x\rangle(0))}{T} = \frac{e}{T} \int \left(\psi^*(T,k) i \frac{\partial \psi}{\partial k}-\psi^*(0,k) i \frac{\partial \psi}{\partial k}\right)\frac{dk}{2\pi},
\end{equation}
with here $x=x_{\perp}$, $k=k_{\perp}$ and $T$ referring to the final time in the protocol. The left-hand side has the dimension of an (averaged) current density in a one-dimensional pumping protocol and on the right-hand side we introduce the wave-function of interest, e.g., the ground-state wavefunction associated to a given sub-system $j$ resolved in the reciprocal space. On the sphere, we study the perpendicular electron current 
\begin{equation}
J_{\perp}^e=\frac{e}{T}\int_0^T dt \frac{d\langle x_{\perp}\rangle}{dt} = \frac{e}{T}\left(\langle x_{\perp}\rangle(T)-\langle x_{\perp}\rangle(0)\right) = \frac{e}{T}\oint \frac{d\phi}{2\pi} \left(\psi^*(T,\phi) i\frac{\partial \psi}{\partial \phi}-\psi^*(0,\phi) i \frac{\partial \psi}{\partial \phi}\right)
\end{equation}
and the time $T$ is related to the angle $\theta$ through the equality $\theta=e E T/\hbar = v^* T$, referring to a specific point along the path at $\phi=0$.
From the analogy between the left and right-hand side, we can define a current density from the reciprocal space.

Then, we define for a fixed angle $\phi$ 
\begin{equation}
J_{\phi}(\theta,\phi) = \frac{i e}{4\pi T}\left(\psi^{*}\frac{\partial}{\partial \phi}\psi - \frac{\partial\psi^*}{\partial\phi}\psi\right) = \frac{ie}{2\pi T}\psi^*\frac{\partial}{\partial \phi}\psi
\end{equation}
such that 
\begin{equation}
J_{\perp}^e(\theta) = \oint d\phi \left(J_{\phi}(\theta,\phi) - J_{\phi}(0,\phi)\right).
\end{equation}
We identify the important relation
\begin{equation}
J_{\phi}(\theta,\phi) = \frac{e}{2\pi T}{\mathcal A}_{\phi}(\theta,\phi)
\end{equation}
Therefore, we observe a relation between the  transverse current density and the smooth fields:
\begin{equation}
\label{transport}
J_{\perp}^e(\theta) = \frac{e}{2\pi T} \oint d\phi {\mathcal A}'_{\phi,\theta<\theta_c}(\theta,\phi).
\end{equation}
Here, $\mathcal {A}'_{\phi,\theta<\theta_c}(\theta,\phi)$ refers to the azimuthal angle component related to the smooth field defined in (\ref{cases}) with an angle $\theta$ in the north hemisphere such that $\theta<\theta_c$. This shows a relation between the perpendicular current 
and the smooth fields.

These definitions agree with general definitions, which are also applicable for many-body systems, between current density and charge polarization in a pump-geometry. The current density in a one-dimensional topological pump geometry reads
\begin{equation}
\label{currentdensity}
j=n e v = e\int_{BZ} \frac{dq}{2\pi} v(q),
\end{equation}
with the anomalous velocity $v(q)=-\Omega_{qt}$ being related to the Berry curvature and $BZ$ refers to the Brillouin zone. We have defined $F_{\mu\nu}=F_{\phi\theta}$ and here we have the relation $\Omega_{qt}=- v^* F_{\phi\theta} =-v^* \partial \mathcal{A}_{\phi}(\theta,\phi)/\partial\theta$ with $v^*=eE/\hbar$ reproducing the Karplus-Luttinger velocity ${\bf v}=(e/\hbar){\bf E}\times {\bf F}$ with $|{\bf v}|=(e/\hbar)E|F_{\phi\theta}|$. The left-hand side measures the current density in real space with $n$ being the density on the lattice. The polarization associated to the charge $e$ is 
\begin{equation}
\Delta P = -e\int_0^T dt \int_{BZ} \frac{dq}{2\pi} \Omega_{qt}
\end{equation}
and
\begin{equation}
\Omega_{\mu\nu} = \partial_\mu \mathcal{A}_{\nu} - \partial_\mu \mathcal{A}_{\nu}.
\end{equation}
For the sphere geometry, the polarization takes the form
\begin{equation}
\Delta P = e\oint \frac{d \phi}{2\pi}\left(\mathcal{A}_{\phi}(\theta,\phi) - \mathcal{A}_{\phi}(0,\phi)\right) = e\oint \frac{d\phi}{2\pi} \mathcal{A}'_{\phi,\theta<\theta_c}(\theta,\phi)
\end{equation}
and here we use periodic boundary conditions in the direction of the azimuthal angle. This relation is in principle independent of the value of the electric field, but we assume the adiabatic limit i.e. the polarization is defined from the ground state.
We also have the relation between current and polarization
\begin{equation}
\Delta P = \int_0^T dt j
\end{equation}
and in the present protocol the current density is uniformly distributed in the time period $\Delta P = j T = J_{\perp} T$, such that any choice of $T$ will measure the same polarization. 

To define a gauge-invariant current density, we can now take the point of view of the south pole and take into account the current associated to a charge which has possibly traveled from south pole to the measurement angle 
$\theta$ during the protocole. Due to the application of the electric field, this charge now corresponds to a charge $-e$. We verify below that this protocole is also equivalent to measure directly the flow of a charge $e$ from north to south pole, such that the protocole is in agreement with physical laws. The choice of south pole as an upper limit for the polar angle here comes from the fact that $\theta\in[0;\pi]$ from the definition of the area and topological properties
on the sphere. The current linked to the charge $-e$ is 
\begin{equation}
 J_{\perp}^{-e}(\theta) = \oint d\phi \left(J_{\phi}(\theta,\phi) - J_{\phi}(\pi,\phi)\right) = -\frac{e}{2\pi T} \oint d\phi {\cal A}'_{\phi,\theta>\theta_c}(\theta,\phi),
\end{equation}
with ${\cal A}'_{\phi,\theta>\theta_c}$ now referring to the azimuthal component of the smooth field in (\ref{cases}) with an angle $\theta$ in the south hemisphere such that $\theta>\theta_c$. 
Related to the definition of the smooth fields, we can now choose the measurement angle to be $\theta_c$, such that we observe a relation between the definitions of north and south hemispheres.
Using Eqs. (\ref{eq:north'}) and (\ref{eq:south'}), the current density measured at the position $(\theta_c,\phi)$ now reads
\begin{equation}
|J_\perp| = |J_{\perp}^{e}(\theta) + J_{\perp}^{-e}(\theta)| = \frac{e}{T} C,
\end{equation}
with the identifications ${\mathcal A}'_{\phi,\theta>\theta_c}(\theta,\phi)={\mathcal A}'_{\phi}(\theta_c^+,\phi)$ and ${\mathcal A}'_{\phi,\theta<\theta_c}(\theta,\phi)={\mathcal A}'_{\phi}(\theta_c^-,\phi)$. Then, we obtain
the quantized polarization on a sphere associated to a physical plane
\begin{equation}
|\Delta P| = J _\perp T = e C.
\end{equation}
This analysis is yet valid for multispheres or planes when we refer to the polarization in a {\it given} plane as a response to the electric field and we have not used a specific form of wave-functions in agreement with the generality of Stokes' theorem. 
For the specific case where $\theta_c\rightarrow \pi$, we obtain the relation
\begin{equation}
\oint d\phi \frac{T^*}{e}\left(J_{\phi}(0,\phi)-J_{\phi}(\pi,\phi)\right) = C
\end{equation}
with $T^* = \pi/v^*$. This formula shows that the topological number can be equivalently measured when driving from north to south poles which is already visible from Eq. (\ref{transport}) when setting $\theta\rightarrow\pi$. 

To make a bridge with the quantum Hall conductivity, we start from Eq. (\ref{currentdensity}) in
the reciprocal space where we can define the current density for a fixed value of $q=\phi$
\begin{equation}
| j |_q = e | v(q)|
\end{equation}
with $v(q)=v^* \partial {\mathcal A}_{\phi}(\theta,\phi)/\partial\theta$. 
The quantum Hall conductivity can be then calculated from the reciprocal space integrating this current density on all
the possible values of $\phi$ and $\theta$ associated to the reciprocal space and respecting the measure in flat space $\frac{(dk_{\parallel} dk_{\perp})}{(2\pi)^2}$. Then, we identify the transverse current density 
\begin{equation}
j_{xy} = \frac{e}{2\pi} v^* \left( \mathcal{A}_{\phi}(0) - \mathcal{A}_{\phi}(\pi) \right) = \frac{(eC)}{2\pi} v^* = \frac{e^2}{h}C E,
\end{equation}
and therefore the quantum Hall conductivity $\sigma_{xy}=(e^2/h)C$ defined from Eq. (\ref{CA}). This formula implies for multispheres' or multiplanes' systems, that applying an electric field on a given subsystem $j$, one can now measure $C$ related to the Berry fields at the Dirac points. The related arguments developed in the article then show that the one-half topological numbers can be observed from a charge polarization protocol and from the quantum Hall conductivity. From Eq. (9) in the article, the topological number can also be measured when driving from north to south pole measuring the magnetizations. To the best of our knowledge, the relation between the smooth fields on the sphere, the Berry connections at the Dirac points and quantum transport properties was not previously mentioned in the literature.

For the specific case $C=1/2$, the Stokes'  theorem allows us to show that we may re-interpret all the topology as a $\pi$ Berry phase encircling the Dirac point associated to the edge structure and to the hemisphere encircling the topological charge. More precisely, Eqs. (\ref{eq:north'}) and (\ref{eq:south'}) are then equivalent to
\begin{equation}
C = \frac{1}{2\pi}  \oint d\phi \left( \mathcal{A}'_{\phi}(\theta_c^+,\phi) -  \mathcal{A}'_{\phi}(\theta_c^-,\phi)\right) = \frac{1}{2}.
\end{equation}
If we now move the boundary very close to the north pole encircling the Dirac point, then from our definitions $\mathcal{A}'_{\phi}(\theta_c^-,\phi)=0$ and we obtain
\begin{equation}
 \oint d\phi \mathcal{A}'_{\phi}(\theta_c^+,\phi) =\pi = \mathcal{A}_{\phi}(0) - \mathcal{A}_{\phi}(\pi).
\end{equation}
The left-hand side of this equation agrees with the fact that in the Haldane model a Dirac point is characterized by a $\pi$ Berry phase, and in the present case, this particular point has a well defined energy bandgap such that a local $\pi$ Berry phase interpretation remains meaningful. The right-hand side of this equation emphasizes that one can transport the topological response from one pole to another, as discussed in the previous section such that only the quantity $\mathcal{A}_{\phi}(0) - \mathcal{A}_{\phi}(\pi)$ defined from the poles is gauge-invariant.

The situation is different
for the blue phase of the phase diagram in Fig. 1b) (in the main article) as in that case $C^1=1$ and $C^2=0$ associated with the two planes, implying that we do not have a quantum Hall conductivity $\sigma_{xy}=\frac{1}{2}\frac{e^2}{h}$ per plane. It is important to mention here for a comparison that when
a sphere develops a unit quantized Chern number, from Stokes' theorem, the Berry phases at the two Dirac points can be transported at one pole (or one Dirac point in the plane) when setting $\theta_c\rightarrow 0$ or $\pi$, and in that case the sum of the two Berry phases englobing the two hemispheres reads $2\pi=\pi+\pi$. 
\subsection{Vector Potential and Wavefunction}
\label{VPW}
Here, we discuss gauge arguments relating the north and south poles to the two-particle wavefunction. 

There exists a set of gauge choices for which the ground state is single-valued (i.e. independent of $\phi$) at $\theta=0$; we denote any of these choices by $|\psi_N\rangle$. Likewise, there are gauges for which the ground state is single-valued at the $\theta=\pi$, denoted $|\psi_S\rangle$. We then define $\mathcal{A}^j_{N\phi}(\theta,\phi)$ and $\mathcal{A}^j_{S\phi}(\theta,\phi)$ such that
\be
\mathcal{A}^j_{N\phi}(\theta,\phi)\equiv i\langle\psi_N(\theta,\phi)|\partial^j_\phi|\psi_N(\theta,\phi)\rangle,\;\; \mathcal{A}^j_{S\phi}(\theta,\phi)\equiv i\langle\psi_S(\theta,\phi)|\partial^j_\phi|\psi_S(\theta,\phi)\rangle.
\ee
Note the behaviour of these functions at the poles. By definition, $|\psi_N\rangle$ is independent of $\phi$ at $\theta=0$ and $|\psi_S\rangle$ is independent of $\phi$ at $\theta=\pi$. Thus, 
\begin{eqnarray}
(\partial^1_\phi+\partial^2_\phi)|\psi_N(\theta=0,\phi)\rangle&=&0,\\
(\partial^1_\phi+\partial^2_\phi)|\psi_S(\theta=\pi,\phi)\rangle&=&0,
\end{eqnarray}
so that
\begin{eqnarray}
\mathcal{A}^1_{N\phi}(0,\phi)+\mathcal{A}^2_{N\phi}(0,\phi)&=&0,\label{eq:pole0}\\
\mathcal{A}^1_{S\phi}(\pi,\phi)+\mathcal{A}^2_{S\phi}(\pi,\phi)&=&0.\label{eq:pole1}
\end{eqnarray}
Furthermore, in any particular gauge, the Berry connection (not the wavefunction) is a single-valued function over the entire parameter space, so that
\be
\mathcal{A}^j_{N\phi}(\pi,\phi)=\mathcal{A}^j_{N\phi}(\pi),\;\;\mathcal{A}^j_{S\phi}(0,\phi)=\mathcal{A}^j_{S\phi}(0).\label{eq:pole2}
\ee
This is easily checked from the form of wavefunction used in the article.

Without saying anything about the symmetry of the wavefunction, we know that $\vec{A}^j_N(\theta,\phi)$ and $\vec{A}^j_S(\theta,\phi)$ are related by a gauge transformation. This relation works just as it does for the standard Berry connection:
\begin{lemma}\label{lem}
For $|\tilde{\psi}\rangle$ and $|\psi\rangle$ in the class of wavefunctions considered in the article, if $|\tilde{\psi}\rangle=e^{i\chi(\theta,\phi)}|\psi\rangle$, then the corresponding partial Berry connections have azimuthal components related via
\be
\tilde{\mathcal{A}}^j_\phi(\theta,\phi)=-\partial_\phi\chi_j(\phi)+\mathcal{A}^j_\phi(\theta,\phi),
\ee
where $\chi(\theta,\phi)=\chi_0(\theta)+\chi_1(\phi)+\chi_2(\phi)$, and $j=1,2$.  
\end{lemma}
\begin{proof}
Here, we refer to the class of wavefunctions in the article:
\be
|\psi\rangle=\sum_{k l} c_{k l}(\theta)\left|\Phi_{k}(\phi)\right\rangle_{1}\left|\Phi_{l}(\phi)\right\rangle_{2}.
\ee
Since the wavefunction must take this form at all points on the sphere, we need only consider the set of gauge transformations that preserve this form (as pointed out in the article, the $j$th Berry connection is well-defined as long as we stay in this sector), which means
\be
|\psi\rangle\rightarrow|\tilde{\psi}\rangle\equiv\sum_{kl}e^{i\chi(\theta)}c_{kl}(\theta)e^{i\chi_1(\phi)}|\Phi_k(\phi)\rangle_1e^{i\chi_2(\phi)}|\Phi_l(\phi)\rangle_2.
\ee
Note that the decomposition of $\chi$ into $\chi_1$ and $\chi_2$ is not unique. Each different decomposition should be regarded as a different gauge choice.
It suffices to just consider the $\phi$ component of the Berry connection:
\begin{eqnarray}
\tilde{\mathcal{A}}^{j}_{\phi}(\theta,\phi)&=&i\langle\tilde{\psi}|\partial^j_\phi|\tilde{\psi}\rangle\\
&=&-\sum_{kl}|c_{kl}(\theta)|^2\partial_\phi\chi_j(\phi)+ \mathcal{A}^{j}_{\phi}(\theta,\phi)\\
&=&-\partial_\phi\chi_j(\phi) + \mathcal{A}^{j}_{\phi}(\theta,\phi).\label{eq:gaugetrans}
\end{eqnarray}
In the last line we used the normalization of the wavefunction $\sum_{kl}|c_{kl}(\theta)|^2=1$.
\end{proof}

\begin{corollary}
There exists a north gauge such that $\mathcal{A}^j_{N\phi}(0,\phi)=0$ and a south gauge such that $\mathcal{A}^j_{S\phi}(\pi,\phi)=0$.
\end{corollary} 
\begin{proof}
For any given north gauge, $|\psi_N\rangle$ with Berry connections $\mathcal{A}^j_{N\phi}$, we know that $|\psi_N(\theta=0)\rangle$ is independent of $\phi$ by definition. Choose 
\begin{eqnarray}
\chi_1(\phi)&=&-\chi_2(\phi)=-\int d\phi\mathcal{A}_{N\phi}^2(\theta=0,\phi).\label{eq:cor_chi}
\end{eqnarray}
Then using Eq.~\eqref{eq:gaugetrans}, the new connection given by this gauge transform is
\begin{eqnarray}
\tilde{\mathcal{A}}^{1}_{\phi}(\theta=0,\phi)&=&\mathcal{A}^{1}_{\phi}(\theta=0,\phi)+\mathcal{A}^{2}_{\phi}(\theta=0,\phi) = 0\\
\tilde{\mathcal{A}}^{2}_{\phi}(\theta=0,\phi)&=&\mathcal{A}^{2}_{\phi}(\theta=0,\phi)-\mathcal{A}^{2}_{\phi}(\theta=0,\phi) = 0.
\end{eqnarray}
The first line is zero by Eq.~\eqref{eq:pole0}. Since $\chi=\chi_1+\chi_2=0$, we have not left the north gauge sector. The same construction can be used for the south gauge, so indeed we can always find north and south gauges such that
\be
\mathcal{A}^j_{N\phi}(0,\phi)=0,\;\;\mathcal{A}^j_{S\phi}(\pi,\phi)=0.\label{eq:cor_poles}
\ee
Note that these particular gauges are the ones used to define $\b{\mathcal{A}}'$ in the previous section.
\end{proof}
Lemma~\ref{lem} gives a simple expression for the partial Chern number:
\begin{eqnarray}
\mathcal{C}^j
&=&-\frac{1}{2\pi}(\chi^j(2\pi)-\chi^j(0)).
\end{eqnarray}
For the traditional Chern number $\mathcal{C}$, integer quantization follows from the fact that $\chi(\phi)=\chi(\phi+2\pi)+2\pi n$ for integer $n$. The same condition applies here, only now it's $\chi_1(\phi)+\chi_2(\phi)=\chi_1(\phi+2\pi)+\chi_2(\phi+2\pi)+2\pi n$, so that fractional values of $\mathcal{C}^j$ are allowed.

There is another way to write $C^j$. Starting from Eq.~\eqref{eq:northsouth} where the sphere has been split along a line of constant $\theta=\theta_c$, and using the particular gauges of Eq.~\eqref{eq:cor_poles}, we have
\begin{eqnarray}
C^j&=&-\frac{1}{2\pi} \int_{\rm north}d^{2} \b{n}\cdot( \nabla\times \b{\mathcal{A}}^j_N)-\frac{1}{2 \pi} \int_{\rm south} d^{2}\b{n}\cdot(\nabla \times \b{\mathcal{A}}^j_S)\\
&=& -\frac{1}{2\pi}\int_0^{2\pi}d\phi(\mathcal{A}^j_{N\phi}(\theta_c,\phi)-\mathcal{A}^j_{S\phi}(\theta_c,\phi)),
\end{eqnarray}
where we applied Stokes's theorem since the Berry connections in the north and south gauges are smooth over their respective hemisphere.
Since the lemma guarantees that the difference between the two Berry connections of different gauges is independent of $\theta$, we may set $\theta=\pi$ (we could also choose $\theta=0$, the proof goes the same either way):
\begin{eqnarray}
\mathcal{C}^j&=&-\frac{1}{2\pi}\int_0^{2\pi}d\phi(\mathcal{A}^j_{N\phi}(\pi,\phi)-\mathcal{A}^j_{S\phi}(\pi,\phi))\\
&=&-\frac{1}{2\pi}\int_0^{2\pi}d\phi\mathcal{A}^j_{N\phi}(\pi,\phi)\\
&=&-\mathcal{A}^j_{N\phi}(\pi),\label{eq:ChernN}
\end{eqnarray}
where we used Eq.~\eqref{eq:pole2} in the third line. This is possibly the simplest expression for the Chern number, but it requires computing everything in a specifically defined gauge. We prefer to write an expression that is explicitly gauge-independent relating to the geometrical argument of the preceding section. We can do this by adding zero to the above expression in the form
\be
\mathcal{C}^j=-\mathcal{A}^j_{N\phi}(\pi)+\mathcal{A}^j_{N\phi}(0).
\ee
Now we have an expression that involves the difference between a Berry connection at two different angles. But from the lemma, we know that such a difference is independent of gauge.
Therefore we again obtain
\be
\mathcal{C}^j=-\mathcal{A}^j_{\phi}(\pi)+\mathcal{A}^j_{\phi}(0)\label{eq:final}
\ee
in any gauge.

\subsection{Examples}
\label{cases2}

It's worthwhile checking some of the above equations in some simple examples.
\subsubsection{Product state}
Consider the $\tilde{r}=0$ case and $M_i=0$ where the ground state is the product state. The north and south gauges are:
\begin{eqnarray}
\label{North}
|\psi_N\rangle&=&\begin{pmatrix}
\cos(\theta/2) \\
e^{i\phi}\sin(\theta/2)\\
\end{pmatrix}
\otimes
\begin{pmatrix}
\cos(\theta/2) \\
e^{i\phi}\sin(\theta/2)
\end{pmatrix},
\end{eqnarray}
\begin{eqnarray}
\label{South}
|\psi_S\rangle&=&\begin{pmatrix}
e^{-i\phi}\cos(\theta/2) \\
\sin(\theta/2)\\
\end{pmatrix}
\otimes
\begin{pmatrix}
e^{-i\phi}\cos(\theta/2) \\
\sin(\theta/2)
\end{pmatrix},
\end{eqnarray}
so that $\mathcal{A}^1=\mathcal{A}^2$ and
\be
\mathcal{A}^j_{N\phi}=-\sin^2(\theta/2),\;\; \mathcal{A}^j_{S\phi}=\cos^2(\theta/2).
\ee
In the article, we present the proof with $\mathcal{A}^j_{N\phi}=-\sin^2(\theta/2)$ implying that $\mathcal{A}^j_{N\phi}=0$ at the north pole and $\mathcal{A}^j_{N\phi}=-1$ at the south pole. At the equator,
we also find  $\mathcal{A}^j_{N\phi}=-1/2$.

In either case, the Berry curvature is $\mathcal{F}_{\phi\theta}d\theta d\phi=\frac{\sin\theta}{2}d\theta d\phi$, which gives partial Chern numbers of
\be
\mathcal{C}^j=\frac{1}{2\pi}\int_0^{2\pi}\int_0^\pi d\phi d\theta \frac{\sin\theta}{2} = 1.
\ee
This equation naturally relates to the magnetizations at the poles since
\be
\mathcal{C}^j= \int_0^{\pi} d\theta \frac{\sin\theta}{2} = \frac{1}{2}[-\cos\theta]_0^{\pi} = \frac{1}{2}\left(\langle \sigma_z(\theta=0)\rangle - \langle \sigma_z(\theta=\pi)\rangle\right) = +1.
\ee

If we move to a different gauge
\begin{eqnarray}
|\psi\rangle
&=&e^{i\chi_0(\theta)}\begin{pmatrix}
e^{i\chi_1(\phi)}\cos(\theta/2) \nonumber \\
e^{i(\chi_1(\phi)+\phi)}\sin(\theta/2)\\
\end{pmatrix}
\otimes
\begin{pmatrix}
e^{i\chi_2(\phi)}\cos(\theta/2) \nonumber \\
e^{i(\chi_2(\phi)+\phi)}\sin(\theta/2)
\end{pmatrix},
\end{eqnarray}
then 
\begin{eqnarray}
\mathcal{A}^j
&=&-\partial_\phi\chi_j(\phi)+\mathcal{A}^j_{N\phi},
\end{eqnarray}
in agreement with Eq.~\eqref{eq:gaugetrans}.

Finally, at the poles, we have
\begin{eqnarray}
\mathcal{A}^1_{N\phi}(\theta=0)=0,&\;\;& \mathcal{A}^1_{N\phi}(\theta=\pi)=-1,\\
\mathcal{A}^2_{S\phi}(\theta=0)=1,&\;\;& \mathcal{A}^2_{S\phi}(\theta=\pi)=0.
\end{eqnarray}
So that both Eq.~\eqref{eq:ChernN} and Eq.~\eqref{eq:final} give $\mathcal{C}^j=1$.
\subsubsection{Fractional topology state}
Now consider the case where the ground state evolves from the product state $|\uparrow\rangle_1|\uparrow\rangle_2$ to the entangled state $\frac{1}{\sqrt{2}}(|\uparrow\rangle_1|\downarrow\rangle_2+|\downarrow\rangle_1|\uparrow\rangle_2)$. In this case, we don't know the form of the wavefunction over the whole sphere, but we can still check the poles. The gauge used in the article corresponds to a north gauge since the ground state is single-valued at the north pole
\be
|\psi_N(\theta=0)\rangle=|\uparrow\rangle_1|\uparrow\rangle_2,\label{eq:north_entangle1}
\ee 
while at the south pole it is given by
\be
|\psi_N(\theta=\pi)\rangle=\frac{1}{\sqrt{2}}\bigg(|\uparrow\rangle_1(e^{i\phi}|\downarrow\rangle_2)+(e^{i\phi}|\downarrow\rangle_1)|\uparrow\rangle_2\bigg),\label{eq:north_entangle2}
\ee
so that
\begin{eqnarray}
\mathcal{A}^1_{N\phi}(\theta=0)=0,&\;\;& \mathcal{A}^1_{N\phi}(\theta=\pi)=-1/2,\\
\mathcal{A}^2_{N\phi}(\theta=0)=0,&\;\;& \mathcal{A}^2_{N\phi}(\theta=\pi)=-1/2.
\end{eqnarray}
Both Eqs.~\eqref{eq:ChernN} and \eqref{eq:final} then give $\mathcal{C}^1=\mathcal{C}^2=1/2$. This leads to the physical interpretation of the fractional Chern number in Fig. 1 of the article if we 
use the arguments of the preceding sub-section. The system is topologically in a superposition of two geometries, one enclosing the topological charge and another geometry which is topologically trivial. 

On the other hand, we can find a south gauge by multiplying $|\psi_N\rangle$ by $e^{-i\phi}$, i.e. $\chi=-\phi$. Now we have the additional gauge freedom of choosing the decomposition into $\chi_1(\phi)$ and $\chi_2(\phi)$. 
For example, suppose we choose $\chi_1=-\phi$, $\chi_2=0$. Then
\begin{eqnarray}
|\psi_S(\theta=0)\rangle&=&e^{-i\phi}|\uparrow\rangle_1|\uparrow\rangle_2,\\
|\psi_S(\theta=\pi)\rangle&=&\frac{1}{\sqrt{2}}\bigg((e^{-i\phi}|\uparrow\rangle_1)(e^{i\phi}|\downarrow\rangle_2)+|\downarrow\rangle_1|\uparrow\rangle_2\bigg),
\end{eqnarray}
so that
\begin{eqnarray}
\mathcal{A}^1_{S\phi}(\theta=0)=1,&\;\;& \mathcal{A}^1_{S\phi}(\theta=\pi)=1/2,\\
\mathcal{A}^2_{S\phi}(\theta=0)=0,&\;\;& \mathcal{A}^2_{S\phi}(\theta=\pi)=-1/2,
\end{eqnarray}
which of course gives $\mathcal{C}^1=\mathcal{C}^2=1/2$. Note that $\mathcal{A}^1_{S\phi}(\theta=\pi)+\mathcal{A}^2_{S\phi}(\theta=\pi)=0$ in accordance with Eq.~\eqref{eq:pole1}, but this is not the particular south gauge that satisfies Eq.~\eqref{eq:cor_poles}. To obtain that gauge we follow the construction in Eq.~\eqref{eq:cor_chi}:
\begin{eqnarray}
\chi_1(\phi)=-\chi_2(\phi)&=&-\int d\phi \mathcal{A}^2_{S\phi}(\theta=\pi)\\
&=&\frac{\phi}{2} +c,
\end{eqnarray}
where $c$ is an arbitrary integration constant. With this transformation, the wavefunction becomes
\begin{eqnarray}
|\psi_S(\theta=0)\rangle&\rightarrow&(e^{-i\phi/2}|\uparrow\rangle_1)(e^{-i\phi/2}|\uparrow\rangle_2),\\
|\psi_S(\theta=\pi)\rangle&\rightarrow&\frac{1}{\sqrt{2}}\bigg((e^{-i\phi/2}|\uparrow\rangle_1)(e^{i\phi/2}|\downarrow\rangle_2)+(e^{i\phi/2}|\downarrow\rangle_1)(e^{-i\phi/2}|\uparrow\rangle_2)\bigg),
\end{eqnarray}
for which
\begin{eqnarray}
\mathcal{A}^1_{S\phi}(\theta=0)=1/2,&\;\;& \mathcal{A}^1_{S\phi}(\theta=\pi)=0,\\
\mathcal{A}^2_{S\phi}(\theta=0)=1/2,&\;\;& \mathcal{A}^2_{S\phi}(\theta=\pi)=0,
\end{eqnarray}
which satisfies both Eq.~\eqref{eq:cor_poles} and $\mathcal{C}^1=\mathcal{C}^2=1/2$. This choice of gauge leads to the physical picture of Fig. 2 (top) in the article where the massive Dirac point
carries the $\pi$ Berry phase and the semimetal ring region participates in the entanglement entropy, and leads to a ``zero'' topological response. 

These examples have all used north or south gauges, but it's important to emphasize that Eq.~\eqref{eq:final} holds for any gauge. As a final example, consider starting with \eqref{eq:north_entangle1} and \eqref{eq:north_entangle2}, and applying the transform $\chi_1=-\phi/3$, $\chi_2=\phi/2$. The wavefunction
\begin{eqnarray}
|\psi(\theta=0)\rangle&=&(e^{-i\phi/3}|\uparrow\rangle_1)(e^{i\phi/2}|\uparrow\rangle_2),\\
|\psi(\theta=\pi)\rangle&=&\frac{1}{\sqrt{2}}\bigg((e^{-i\phi/3}|\uparrow\rangle_1)(e^{3i\phi/2}|\downarrow\rangle_2)+(e^{2i\phi/3}|\downarrow\rangle_1)(e^{i\phi/2}|\uparrow\rangle_2)\bigg),
\end{eqnarray}
is clearly not single-valued at any pole, but its Berry connections
\begin{eqnarray}
\mathcal{A}^1_{\phi}(\theta=0)=1/3,&\;\;& \mathcal{A}^1_{\phi}(\theta=\pi)=-1/6,\\
\mathcal{A}^2_{\phi}(\theta=0)=-1/2,&\;\;& \mathcal{A}^2_{\phi}(\theta=\pi)=-1,
\end{eqnarray}
still give $\mathcal{C}^1=\mathcal{C}^2=1/2$.


\section{Generalized interactions}\label{sec:anisotropic}

We now consider a general spin model with both in-plane and out-of-plane coupling:
\be
\mathcal{H}=-\b{H}_1\cdot\b{\sigma}_1-\b{H}_2\cdot\b{\sigma}_2+r_z\sigma_1^z\sigma_2^z+r_{xy}(\sigma_1^x\sigma_2^x+\sigma_1^y\sigma_2^y).
\ee
We anticipate that fractional invariants occur only along the line of symmetry $M_1=M_2$, and we set $\b{H}_1=\b{H}_2=(H_x,0,H_z)$, focusing on sweeps along the $\phi=0$ meridian. This Hamiltonian admits a nice singlet-triplet representation:
\be
\mathcal{H}=-(2r_{xy}+r_z)|00\rangle\langle00| +\mathcal{H}_{\rm trip},\label{eq:xyst}
\ee
where
\be
\mathcal{H}_{\rm trip}=\begin{pmatrix}
-2H_z+r_z & -\sqrt{2}H_x & 0 \\
-\sqrt{2}H_x & 2r_{xy}-r_z & -\sqrt{2}H_x\\
0 & -\sqrt{2}H_x & 2H_z+r_z
\end{pmatrix},\label{eq:xyt}
\ee
in the $(|1,1\rangle, |1,0\rangle, |1,-1\rangle)$ basis. Here $H_x=H\sin\theta$ and $H_z=H\cos\theta+M$. Regardless of the spin coupling, we see that the singlet state is decoupled from the triplet states. If it is not the ground state at the north pole, then it will not contribute to the sweep across the sphere. At the poles, we have the energies shown in table~\ref{table:energies}.
\begin{table}[h]
\centering
\begin{tabular}{ l | c | c | c | c | c | }
\centering                   
 Pole & $|0,0\rangle$ & $|1,1\rangle$ & $|1,0\rangle$ & $|1,-1\rangle$\\
 \hline
North & $-2r_{xy}-r_z$ & $-2(H+M)+r_z$ & $2r_{xy}-r_z$ & $2(H+M)+r_z$ \\
South & $-2r_{xy}-r_z$ & $-2(M-H)+r_z$ & $2r_{xy}-r_z$ & $2(M-H)+r_z$ \\
\end{tabular}
\caption{Energies at the poles for each of the singlet and triplet states.}\label{table:energies}
\end{table}

We choose the definition $H,M>0$. Of the 16 possible transitions between these four states, six are ruled out by the singlet-triplet decoupling, and five are ruled out by energetics. The remaining five are
\begin{eqnarray}
|0,0\rangle &\rightarrow& |0,0\rangle\;\;\Rightarrow\mathcal{C}^j=0\nonumber\\
|1,1\rangle &\rightarrow& |1,1\rangle\;\;\Rightarrow\mathcal{C}^j=0\nonumber\\
|1,1\rangle &\rightarrow& |1,0\rangle\;\;\Rightarrow\mathcal{C}^j=1/2\nonumber\\
|1,1\rangle &\rightarrow& |1,-1\rangle\;\;\Rightarrow\mathcal{C}^j=1\nonumber\\
|1,0\rangle &\rightarrow& |1,0\rangle\;\;\Rightarrow\mathcal{C}^j=0.
\end{eqnarray}

Let's focus on the $\mathcal{C}^j=1/2$ transition. At the north pole, this requires
\be
r_z-H-M<r_{xy}<H+M-r_z,
\ee
while at the south pole, this requires
\be
r_{xy}<r_z+H-M,\;\;r_{xy}<r_z-H+M.
\ee
Generically then, we find the fractional Chern number in the region
\be\label{eq:halfregion}
r_z-H-M<r_{xy}<{\rm min}\{H+M-r_z,r_z-|H-M|\},
\ee

in agreement with the $r_{xy}=0$ case studied in the main text. We now study other special cases of this model.
\begin{itemize}
%
\item XY model: $r_z=0$. In this case, the condition for $\mathcal{C}^j=1/2$ becomes
	\be
	-H-M<r_{xy}<-|H-M|.
	\ee
	We see that the XY model admits a fractional Chern number, but only for ferromagnetic in-plane coupling. This system was studied experimentally in Ref.~[12] in the article for the case $M=0$, which did not allow for the observation of a fractional Chern number. Without the inversion symmetry-breaking term, the fractional Chern phase collapses to a point at $r_{xy}=-H$ and therefore was not observed in that experiment.
	
\item Heisenberg model: $r_z=r_{xy}\equiv r$. Condition \eqref{eq:halfregion} yields three conditions
	\be
	H+M>0,\;\; r<\frac{H+M}{2},\;\; |H-M|<0.
	\ee
	The last one is impossible, so we see that the Heisenberg model does not admit a fractional Chern number. However, anisotropy rescues the $\mathcal{C}^j=1/2$ phase. For example, if we take $r_{xy}=r=-r_z$, then condition \eqref{eq:halfregion} becomes
	\be
	-\frac{(H+M)}{2}<r<\frac{|H-M|}{2},	
	\ee
	which is readily satisfied. 
		
\item Inversion symmetric: $M=0$. The region \eqref{eq:halfregion} collapses to a truncated line.
	\be
	r_{xy}=r_z-H<0.
	\ee
\end{itemize}

The topological phase diagram in the plane of the couplings is shown in Supplementary Fig.~\ref{fig:anis_phase}.
For $M<H$, the fractional Chern number phase forms a wedge between the $\mathcal{C}^j=0$ and $\mathcal{C}^j=1$ phases, while for $M>H$, the $\mathcal{C}^j=1$ phase is taken over by $\mathcal{C}^j=0$.

In the main article, we found that we could extend the fractional Chern region by breaking inversion symmetry, either through the inclusion of a constant offset $M\neq0$, or an asymmetric coupling $r_z\rightarrow r_zf(\theta)$. If we do the same in the generic anisotropic model (with $r_{xy}\rightarrow r_{xy}g(\theta)$),  we can also get an extended region of $\mathcal{C}^j=1/2$. It is redundant to include both inversion-symmetry breaking mechanisms, so we set $M=0$ in the following. In that case, the condition for a $|1,1\rangle\rightarrow|1,0\rangle$ transition relaxes to
\begin{eqnarray}
r_zf(0)-H&<&r_{xy}g(0)<H-r_zf(0),\\
r_{xy}g(\pi)&<&r_zf(\pi)-H.
\end{eqnarray}
However, the $\theta$-dependent interactions also open up the possibility of a transition from  $|1,0\rangle$ to $|1,-1\rangle$ (the transition $|1,0\rangle\rightarrow|1,1\rangle$ is forbidden for $M=0$). This occurs when
\begin{eqnarray}
r_{xy}g(0)&<&{\rm min}\{0,-H+r_zf(0)\}\\
H+r_zf(\pi)&<&r_{xy}g(\pi).
\end{eqnarray}

Some examples of interactions that break inversion symmetry are:
\begin{itemize}
\item $f(\theta)=g(\theta)=\pi-|\pi-\theta|$:
	The fractional Chern number appears in the entire half-plane
	\be
	r_{xy}<r_z-\frac{H}{\pi}.
	\ee
\item $f(\theta)=g(\theta)=\sin(\theta-\pi/2)$:
	Intriguingly, the fractional Chern number appears in two separated domains:
	\be
	-H-r_z<r_{xy}<r_z-H,
	\ee
	and 
	\be
	r_{xy}>{\rm max}\{0,r_z+H\}.
	\ee
	This unusual phase diagram is shown in Supplementary Fig.~\ref{fig:phasesin}.
	
\end{itemize}


\section{$N$ spins}\label{sec:nspins}

Here we discuss the fractional phase for more than two spins. Without specifying any details of the model, we know that the fractional Chern number arises when the ground state changes from a ferromagnet at the north pole to a degenerate antiferromagnet at the south pole. This can be achieved through inversion-symmetry-breaking masses or $\theta$-dependent interactions as we have seen. The key observation in the case of two spins was that the presence of an infinitesimal transverse field breaks the ground-state degeneracy near the south pole and favours the entangled state. Specifically, in the subspace of south-pole ground states for the Ising-coupled model with $M_1=M_2$: $\mathcal{D}\equiv\{|\uparrow\downarrow\rangle,|\downarrow\uparrow\rangle\}$, with energy $E_D=-r_z=-\tilde{r}$, we apply second-order degenerate perturbation theory
\be
\mathcal{H}_{\rm eff}=P\mathcal{H}P+P\mathcal{H}'\frac{1-P}{E_D-H_0}\mathcal{H}'P,
\ee
where $P\equiv\sum_{\alpha\in\mathcal{D}}|\alpha\rangle\langle\alpha|$ is the projection operator on the south-pole ground-state subspace, $H_0=(1-P){\cal H}(1-P)$ and
\be
\mathcal{H}'\equiv -H\sin\theta(\sigma_1^x+\sigma_2^x),\label{eq:Hx}
\ee
In the Hilbert space formed by $\{ |\uparrow \downarrow\rangle , |\downarrow \uparrow\rangle\}$, this yields the effective perturbation
\be
\mathcal{H}_{\rm eff}=-\frac{H^2\sin^2\theta\tilde{r}}{\tilde{r}^2-(H-M)^2}\begin{pmatrix}
1 & 1 \\
1 & 1\\
\end{pmatrix},
\ee  
whose unique ground state $\frac{1}{\sqrt{2}}(|\uparrow\downarrow\rangle+|\downarrow\uparrow\rangle)$. In $\mathcal{H}_{\rm eff}$, we  take into account the two possible states $|\uparrow \uparrow\rangle$ and 
$|\downarrow \downarrow\rangle$ in $(1-P)$ summing their two contributions. 

We can generalize this reasoning to chains with more spins. For four spins, all antiferromagnetically coupled, the south-pole ground states are the two N\'eel states. The degeneracy is preserved at second order in the perturbation, but fourth order spin-flip terms will choose the ground state $\frac{1}{\sqrt{2}}(|\uparrow\downarrow\uparrow\downarrow\rangle+|\downarrow\uparrow\downarrow\uparrow\rangle)$. This is a maximally entangled state, and an analogue of the Greenberger-Horne-Zeilinger state which again has $\mathcal{C}^j=1/2$ for all $j$. A chain of $2N$ spins requires $2N$ orders of perturbation theory to lift the degeneracy, so the gap near $\theta=\pi$ will be reduced as $(H\sin\theta)^{2N}$.

Let us now consider a different four-spin model corresponding to two Chern one-half systems with a weak transverse coupling. We would like to know if the fractional invariant is robust to this coupling. In the absence of the perturbative coupling, the ground state at the south pole is 
\be
\mathcal{D}=\{|\uparrow\downarrow\uparrow\downarrow\rangle, |\uparrow\downarrow\downarrow\uparrow\rangle, |\downarrow\uparrow\uparrow\downarrow\rangle, |\downarrow\uparrow\downarrow\uparrow\rangle\}.
\ee
If we couple the two systems at just one site via a term $\mathcal{H}'=r'\sigma^x_2\sigma^x_3$ with $r'=r_x$ as shown in Supplementary Fig.~\ref{fig:Nspins} (a), then the second order perturbation is
\be
\mathcal{H}_{\rm eff}=-\frac{r'^2}{4}{\rm diag}\bigg\{\frac{1}{\tilde{r}}, \frac{1}{\tilde{r}+H-M}, \frac{1}{\tilde{r}-H+M}, \frac{1}{\tilde{r}}\bigg\}.
\ee
In Supplementary Fig.~\ref{fig:Nspins} (b), the top left spin corresponds to spin 2 which refers to the first spin in the definition of of $\mathcal{D}$, the top right spin corresponds to spin 3 which refers to the second spin
in the definition of of $\mathcal{D}$, the bottom right spin corresponds to spin 4 referring to the third spin in the definition of $\mathcal{D}$ and the bottom left spin corresponds to spin 1 referring to the fourth
spin in the definition of $\mathcal{D}$. The two antiferromagnetic N\' eel ordered states on the square then correspond to the first and fourth states in $\mathcal{D}$. 

For $\tilde{r}>H-M$, this favours the state $|\downarrow\uparrow\uparrow\downarrow\rangle$ with an energy shift of $\frac{-r'^2}{4(\tilde{r}-H+M)}$. This is due to virtual excitations of the of the fully polarized state aligned with the magnetic field. In this case the fractional Chern number is destroyed and $\mathcal{C}^j=\{1,0,0,1\}$. This is because fractional $\mathcal{C}^j$ state is only protected by exchange symmetry of the spins. However, we can also construct a perturbative coupling that respects this symmetry, for example, by adding $r'\sigma^x_1\sigma_4^x$ as shown in Fig.~\ref{fig:Nspins} (b) of the Supplementary Material. In this case we have
\be
\mathcal{H}_{\rm eff}=-\frac{r'^2}{2}\begin{pmatrix}
\frac{1}{\tilde{r}} & 0 & 0 & \frac{1}{\tilde{r}}\\
0 & \frac{\tilde{r}}{\tilde{r}^2-(H-M)^2} & \frac{\tilde{r}}{\tilde{r}^2-(H-M)^2} & 0\\
0 & \frac{\tilde{r}}{\tilde{r}^2-(H-M)^2} & \frac{\tilde{r}}{\tilde{r}^2-(H-M)^2} & 0\\
\frac{1}{\tilde{r}} & 0 & 0 & \frac{1}{\tilde{r}}\\
\end{pmatrix}.
\ee
Upon diagonalization, this gives the unique ground state $\frac{1}{\sqrt{2}}(|\uparrow\downarrow\downarrow\uparrow\rangle+|\downarrow\uparrow\uparrow\downarrow\rangle)$, with $\mathcal{C}^j=1/2$.

Other fractional Chern numbers can appear if we introduce frustration in the system without breaking spin-exchange symmetry. For three antiferromagnetically coupled spins as shown in Fig.~\ref{fig:Nspins} (c) of the Supplementary Material, the ground state at $\theta=\pi$ is
\be
\mathcal{D}=\{|\downarrow\downarrow\uparrow\rangle,|\downarrow\uparrow\downarrow\rangle,|\uparrow\downarrow\downarrow\rangle\}.
\ee 
The transverse field $-H\sin\theta\sum_{i=1}^3\sigma_i^x$, yields the second-order perturbation
\begin{eqnarray}
\mathcal{H}_{\rm eff}&=&-\frac{H^2\sin^2\theta}{(H-M)(H-M-2\tilde{r})}\begin{pmatrix}
\frac{(H-M-4\tilde{r})}{2} & -\tilde{r} & -\tilde{r} \\
-\tilde{r} & \frac{(H-M-4\tilde{r})}{2} & -\tilde{r}\\
-\tilde{r} & -\tilde{r} & \frac{(H-M-4\tilde{r})}{2}\\
\end{pmatrix}.\nonumber\\
\end{eqnarray}
Upon diagonalizing, the ground state is $\frac{1}{\sqrt{3}}(|\downarrow\downarrow\uparrow\rangle+|\downarrow\uparrow\downarrow\rangle+|\uparrow\downarrow\downarrow\rangle)$. The corresponding partial Chern numbers are $\mathcal{C}^j=\frac{2}{3}$ for all $j$. 

This reasoning can be applied to higher numbers of spin as well. For odd $N>1$ spins, the $N$-fold degenerate space of frustrated antiferromagnets at the south pole is 
\be
\mathcal{D}=\{|\downarrow\downarrow\uparrow\downarrow\ldots\uparrow\rangle,|\uparrow\downarrow\downarrow\uparrow\downarrow\ldots\downarrow\rangle,\ldots,|\downarrow\uparrow\ldots\uparrow\downarrow\rangle\}.
\ee
Each of these states can be written as $|\alpha\rangle$ where $\alpha$ indexes the site location of the ferromagnetic pair. At second order in perturbation theory, for $N>3$, the transverse field has a diagonal contribution
\begin{eqnarray}
\langle\alpha|\mathcal{H}_{\rm eff}|\alpha\rangle&=&H^2\sin^2\theta\bigg(\frac{\tilde{r}}{2(H-M)(H-M-2\tilde{r})}+\frac{(H-M)-\tilde{r}(2N-8)}{2(4\tilde{r}^2-(H-M)^2)}\bigg),
\end{eqnarray}
from flipping a single spin twice. The only off-diagonal contributions come from flipping one spin in the ferromagnetic pair along with its neighbour outside the pair. This is equivalent to shifting the pair by two sites:
\begin{eqnarray}
\langle\alpha-2|\mathcal{H}_{\rm eff}|\alpha\rangle&=&H^2\sin^2\theta\bigg(\frac{\tilde{r}}{(H-M)(H-M-2\tilde{r})}\bigg)\nonumber\\
&=&\langle\alpha+2|\mathcal{H}_{\rm eff}|\alpha\rangle.
\end{eqnarray}
Thus the effective Hamiltonian describes a single particle hopping on a lattice. This is easily solved by Fourier transform, from which we get the ground state $\frac{1}{\sqrt{N}}\sum_{\alpha=1}^N|\alpha\rangle$. The corresponding spin expectation value for each site $j$ is
\be
\langle \sigma_j^z\rangle=-\frac{1}{N},
\ee
which gives the sequence of rational partial Chern numbers 
\be
\mathcal{C}^j=\frac{N+1}{2N},
\ee
with integer total Chern number $\mathcal{C}_{\rm tot}=\frac{N+1}{2}$. Note that in the large $N$ limit, the even and odd sectors converge to give $\mathcal{C}^j=1/2$.


\section{Monolayer and bilayer Haldane models on the honeycomb lattice}\label{app:haldane}

We employ the following definitions in our lattice model. We set the lattice spacing to $a=1$. The  honeycomb graphene Bravais lattice consists of $A$ and $B$ sites with primitive vectors
\begin{eqnarray}
\b{u}_1=\frac{1}{2}(3,\sqrt{3}),\;\;
\b{u}_2=\frac{1}{2}(3,-\sqrt{3}),
\end{eqnarray}
nearest-neighbour vectors
\begin{eqnarray}
\b{a}_1=\frac{1}{2}(1,\sqrt{3}),\;\;
\b{a}_2=\frac{1}{2}(1,-\sqrt{3}),\;\;
\b{a}_3=(-1,0),\nonumber\\
\end{eqnarray}
and next-nearest-neighbour vectors
\begin{eqnarray}
\b{b}_1=\frac{1}{2}(-3,\sqrt{3}),\;\;\nonumber
\b{b}_2=\frac{1}{2}(3,\sqrt{3}),\;\;\nonumber
\b{b}_3=(0,-\sqrt{3})\nonumber\\
\end{eqnarray} 
as shown in Supplementary Fig.~\ref{fig:Hex}.

The reciprocal lattice has a primitive cell defined by
\begin{eqnarray}
\b{v}_1=\frac{2\pi}{3}(1,\sqrt{3}),\;\;
\b{v}_2=\frac{2\pi}{3}(1,-\sqrt{3}).
\end{eqnarray}
The diamond formed by $\b{v}_1$, $\b{v}_2$ result in Fig. 2b) and 2c) in the article. 
Some important points in the Brillouin zone are
\begin{eqnarray}
K&=&\frac{2\pi}{3}\left(1,\frac{1}{\sqrt{3}}\right),\;\; K'=\frac{2\pi}{3}\left(1,-\frac{1}{\sqrt{3}}\right)\nonumber\\
M&=&\frac{2\pi}{3}(1,0),\;\;M'=\frac{\pi}{3}(1,\sqrt{3}),
\end{eqnarray}
as shown in Supplementary Fig.~\ref{fig:Hex2}.

For a single layer $i$, we start with the tight-binding Hamiltonian for graphene with nearest-neighbour hopping $t_1$ and Semenoff mass $M_i$, which is given by
\begin{eqnarray}
H_1 &=&t_1\sum_{\b{r}_A}\sum_{i=1}^3(c^\dagger_{B}(\b{r}_A+\b{a}_i)c_{A}(\b{r}_A) +h.c.)+ M_i\left(\sum_{\b{r}_A}c^\dagger_A(\b{r}_A)c_A(\b{r}_A)-\sum_{\b{r}_B}c^\dagger_B(\b{r}_B)c_B(\b{r}_B)\right).\nonumber\\
\end{eqnarray}
To construct the Haldane model, we add next-nearest-neighbour hopping $t_2$ with flux $\phi$ oriented as in Supplementary Fig.~\ref{fig:Hex3}, via the term
\begin{eqnarray}
H_2&=&t_2\sum_{i=1}^3\bigg(\sum_{\b{r}_A}c_A^\dagger(\b{r}_A)c_A(\b{r}_A+\b{b}_i)e^{i\phi}+\sum_{\b{r}_B}c_B^\dagger(\b{r}_B)c_B(\b{r}_B+\b{b}_i)e^{-i\phi}\bigg)+h.c.
\end{eqnarray}

Fourier transforming $H_1+H_2$ gives the single-layer Hamiltonian
\begin{eqnarray}
H&=&\sum_{\b{k}}\begin{pmatrix} c^\dagger_{A\b{k}} & c^\dagger_{B \b{k}}\end{pmatrix} h(\b{k}) \begin{pmatrix} c_{A\b{k}} \\ c_{B\b{k}}\end{pmatrix}\nonumber\\
h(\b{k})&=&\b{d}(\b{k})\cdot\b{\sigma}+\epsilon\mathbb{I},
\end{eqnarray}
where 
\begin{eqnarray}
d_x(\b{k})&=&-t_1\sum_{i=1}^3\cos(\b{k}\cdot\b{a}_i)\label{eq:dx}\\
d_y(\b{k})&=&-t_1\sum_{i=1}^3\sin(\b{k}\cdot\b{a}_i)\\
d_z(\b{k})&=&-2t_2\sin\phi\sum_{i=1}^3\sin(\b{k}\cdot\b{b}_i)+M\\
\epsilon&=&-2t_2\cos\phi\sum_{i=1}^3\cos(\b{k}\cdot\b{b}_i).\label{eq:dz}
\end{eqnarray}
At the Dirac points $K$, $K'$, we have
\be
\b{d}=d_z=\pm3\sqrt{3}t_2\sin\phi. 
\ee

Our bilayer model consists of two copies of the Haldane model for which we fix $\phi=\pi/2$ for simplicity and couple them with an interlayer hopping $r$:
\be
\mathcal{H}(\b{k})=\begin{pmatrix}
(\b{d}+M_1\hat{z})\cdot\b{\sigma} & r\mathbb{I}\\
r\mathbb{I} & (\b{d}+M_2\hat{z})\cdot\b{\sigma}
\end{pmatrix}.
\ee
{\it This corresponds to Eq. (14) in the article.}

From this we get the energy spectrum at the $K$ point,
\begin{eqnarray}\label{eq:bands1}
E_1(K)&=&-\frac{1}{2}(2|d_z|+M_1+M_2+\sqrt{(M_1-M_2)^2+4r^2})\\
E_2(K)&=&-\frac{1}{2}(2|d_z|+M_1+M_2-\sqrt{(M_1-M_2)^2+4r^2}),\\
E_3(K)&=&\frac{1}{2}(2|d_z|+M_1+M_2-\sqrt{(M_1-M_2)^2+4r^2}),\\
E_4(K)&=&\frac{1}{2}(2|d_z|+M_1+M_2+\sqrt{(M_1-M_2)^2+4r^2}).\label{eq:bands4}
\end{eqnarray}
The eigenvectors corresponding to these bands are
\begin{eqnarray}
\psi_1&=&N_x(0,x,0,1),\;\;
\psi_2=N_y(0,-y,0,1),\nonumber\\
\psi_3&=&N_y(y,0,1,0),\;\;
\psi_4=N_x(-x,0,1,0),\label{eq:bands_eigen}
\end{eqnarray}
where
\begin{eqnarray}
x&\equiv&\frac{1}{2r}(M_2-M_1-\sqrt{(M_1-M_2)^2+4r^2)},\\
y&\equiv&\frac{1}{2r}(M_1-M_2-\sqrt{(M_1-M_2)^2+4r^2)},\\
N_x&=&(1+x^2)^{-1/2},\\
N_y&=&(1+y^2)^{-1/2}.
\end{eqnarray}
We may obtain the energies at the $K'$ point through the transformation $|d_z|$ with $-|d_z|$ in Eqs.~\eqref{eq:bands1}-\eqref{eq:bands4}.

In the bilayer model, we may represent the ground state at half-filling in terms of the occupancy of each layer:
\be
|\psi\rangle=\sum_{i+j+k+l=2}c_{ijkl}|ij\rangle_1|kl\rangle_2.
\ee
In this representation, a ket $|ij\rangle_n$ is defined for layer $n$ with the two sublattices $A$ and $B$ such that $|10\rangle_n$ refers to a state with sublattice $A$ occupied in layer $n$ and $|01\rangle_n$ to a state with sublattice $B$ occupied in layer $n$. The subset of the Hilbert space with each layer half-filled corresponds to the constraint $i+j=1$. We get the reduced density matrix $\rho_1$ by tracing out one layer:
\be
\rho_1={\rm diag}(|c_{0011}|^2, |c_{1100}|^2, \rho_1^{\rm red}),
\ee
where the $2\times2$ block $\rho_1^{\rm red}$ describes the space of states where each layer is half-filled:
\be
\rho_1^{\rm red}=\begin{pmatrix}
|c_{0110}|^2+|c_{0101}|^2 & c^*_{0110}c_{1010}+c^*_{0101}c_{1001}\\
c_{0110}c^*_{1010}+c_{0101}c^*_{1001} & |c_{1010}|^2+|c_{1001}|^2\\
\end{pmatrix}.
\ee
From these coefficients, the entanglement entropy:
\be
S_1=-\rho_1\ln \rho_1,
\ee
is computed numerically.
\subsection{Edge states}
Here, we study the edge states in this bilayer model. We consider a finite ribbon geometry with $30$ sites in the $y$-direction and $50$-sites in the $x$-direction. Boundary conditions are open in $y$ and periodic in $x$. We evaluate the band structure and wavefunctions of the edge modes. We confirm that the topological phase diagram (Fig. 1b in the main text) has a region with two chiral edge modes, a region with one chiral edge mode and a fully gapped phase as shown in Supplementary Fig.~\ref{fig:edge1}. Along the line of symmetry $M_1=M_2$, we find that that the bulk gap closes due to the nodal ring, but a single chiral edge mode persists in the reciprocal space as shown in Supplementary Fig.~\ref{fig:edge2}. By evaluating the probability density of the wavefunction, we find that these modes are indeed uniformly split between the two layers. Deviating very slightly from the critical line $M_1=M_2$, we find that the edge mode remains delocalized in the two planes and then progressively redistributes in one plane only when increasing the mass asymmetry related to the blue region of the phase diagram of Fig. 1b) in the article. In Supplementary Fig.~\ref{DOS}, we also show the density of states coming from the nodal ring region in real space giving rise to additional delocalized modes for $M_1=M_2$.  

\section{Protocoles in time}\label{app:LZ}

\subsection{Single spin-1/2 model in time}

In this paper, we are interested in transition amplitudes of a two-state system at finite times, since the linear sweep protocol on a sphere takes place over a finite time. To that end, it is worth deriving the full time-dependent amplitudes for different states of the spin-1/2 Hamiltonian
\be\label{eq:LZham}
H=\lambda t \sigma_z+\Delta\sigma_x,
\ee
using Ref.~[46] as a guide.
The instantaneous eigenenergies and eigenstates of this system are:
\begin{eqnarray}
E_\pm&=&\pm\sqrt{\lambda^2t^2+\Delta^2}\equiv\pm E\\
|\psi_\pm\rangle&=&\frac{1}{\sqrt{\Delta^2+(E-\lambda t)^2}}\begin{pmatrix}
\lambda t\pm E \\
\Delta
\end{pmatrix}.
\end{eqnarray}
It is important to note that these eigenstates change character as $t$ goes from $-\infty$ to $+\infty$. Since $E\rightarrow\pm\lambda t$ as $t\rightarrow\pm\infty$, we have
\begin{eqnarray}
|\psi_-(-\infty)\rangle&=&\begin{pmatrix} 1 \\ 0\end{pmatrix} = |\psi_+(+\infty)\rangle\label{eq:spinasymp}\\
|\psi_-(+\infty)\rangle&=&\begin{pmatrix} 0 \\ 1\end{pmatrix} = |\psi_+(-\infty)\rangle.
\end{eqnarray}
In other words, if the evolution is adiabatic (i.e. we track the ground state as $t$ increases), then the spin will necessarily flip. The Landau-Zener result says that if the evolution is not completely adiabatic (in a sense we will soon make precise), then there is a significant probability of ending up in the excited state where the spin has not flipped. Also note that exactly at $t=0$, the eigenstates are equal combinations of up and down:
\be
|\psi_\pm(0)\rangle=\frac{1}{\sqrt{2}}\begin{pmatrix} \pm1 \\ \sgn(\Delta)\end{pmatrix}.
\ee

We wish to solve the time-dependent Schr\"odinger equation. We represent the quantum state as
\be
|\Psi(t)\rangle=A(t)|\uparrow\rangle+B(t)|\downarrow\rangle,
\ee
so that we have two coupled differential equations
\begin{eqnarray}
\dot{A}(t)+\frac{i\lambda}{\hbar}A(t)+\frac{i\Delta}{\hbar}B(t)&=&0\label{eq:SE1}\\
\dot{B}(t)+\frac{i\Delta}{\hbar}A(t)-\frac{i\lambda}{\hbar}tB(t)&=&0.\label{eq:SE2}
\end{eqnarray}
Differentiating the second equation gives
\begin{eqnarray}
\ddot{B}(t)+\frac{i\Delta}{\hbar}\dot{A}(t)-\frac{i\lambda}{\hbar}B(t)-\frac{i\lambda}{\hbar}t\dot{B}(t)&=&0.\label{eq:SE4}
\end{eqnarray}
Substituting Eqs. \eqref{eq:SE1} and \eqref{eq:SE2} into this gives
\begin{eqnarray}
\ddot{B}(t)+\left(\frac{\Delta^2}{\hbar}+\frac{\lambda^2t^2}{\hbar^2}-\frac{i\lambda}{\hbar}\right) B(t)=0.\label{eq:A2}
\end{eqnarray}
We can put this differential equation in the form of the Weber equation\footnote{see, for instance, NIST Digital Library of Mathematical Functions, http://dlmf.nist.gov/}, by using the dimensionless quantity
\be
z\equiv\sqrt{\frac{2\lambda}{\hbar}}e^{-i\pi/4}t,\label{eq:zoft}
\ee
so that
\be
\frac{d^2B(z)}{dz^2}+\left(\frac{i\Delta^2}{2\lambda\hbar}+\frac{1}{2}-\frac{z^2}{4}\right)B(z)=0,
\ee
which has the linearly independent solutions
\be
B(z)=c_1D_\nu(z)+c_2D_{-1-\nu}(-iz).\label{eq:Bsol}
\ee
Here we have defined $\nu\equiv i\frac{\Delta^2}{2\lambda\hbar}$, and $D_\nu(z)$ are the parabolic cylinder functions. From Eq.~\eqref{eq:SE2}, we also get the solution for $A$:
\begin{eqnarray}
A(t)&=&\frac{i\hbar}{\Delta}\dot{B}(t)+\frac{\lambda t}{\Delta}B(t)\\
\Rightarrow A(z)&=&\frac{\sqrt{2\lambda\hbar}}{\Delta}e^{i\pi/4}\left(B'(z)+\frac{z}{2}B(z)\right).
\end{eqnarray}

The initial condition ensures that the spin begins in the ground state at $t=-\infty$, which means that $B(t=-\infty)=0$ according to Eq.~\eqref{eq:spinasymp}. The second initial condition ($|A(t=-\infty)|^2=1$) will not be used just yet. 

One has to be careful with the asymptotics of the parabolic cylinder functions, since they have different behaviours depending the direction in which their argument goes to infinity. For $t\rightarrow-\infty$, Eq.~\eqref{eq:zoft} shows that $\arg(z)=3\pi/4$ and $\arg(-iz)=\pi/4$. One can check that the first term in Eq.~\eqref{eq:Bsol} diverges along the former axis, while the second term decays along the later. Thus $c_1=0$. We then have
\be
A(z)=\frac{\sqrt{2\lambda\hbar}}{\Delta}e^{i\pi/4}\left(-ic_2D'_{-1-\nu}(-iz)+\frac{1}{2}c_2D_{-1-\nu}(-iz)\right).
\ee 
Using the identity
\be
D'_n(z)=\frac{z}{2}D_n(z)-D_{n+1}(z),
\ee
this simplifies to
\be
A(z)=\frac{\sqrt{2\lambda\hbar}}{\Delta}e^{i3\pi/4}c_2D_{-\nu}(-iz).
\ee
Instead of solving for $c_2$ using the other initial condition, it is easier to use the probability normalization at time $t=0$:
\be
1=|A(0)|^2+|B(0)|^2,
\ee 
where the parabolic cylinder functions take the analytic form
\be
D_n(0)=\frac{2^{n/2}\sqrt{\pi}}{\Gamma(\frac{1-n}{2})}.\label{eq:Dn0}
\ee
We employ the following identities to simplify the gamma functions:
\begin{eqnarray}
\Gamma(1-x)\Gamma(x)&=&\frac{\pi}{\sin(\pi x)}\label{eq:gamma1}\\
\Gamma(1+x)&=&x\Gamma(x)\label{eq:gamma2}.
\end{eqnarray}
Noting that $\nu$ is purely imaginary and $\Gamma^*(z)=\Gamma(z^*)$, these identities allow us to write
\begin{eqnarray}
|D_{-\nu}(0)|^2=\cos(\pi\nu/2),\\
|D_{-\nu-1}(0)|^2=\frac{\sin(\pi\nu/2)}{\nu}.
\end{eqnarray}

It is also useful to define the parameter
\be
\gamma\equiv\frac{\Delta^2}{\lambda\hbar}=-2i\nu,
\ee
in terms of which the coefficient becomes
\begin{eqnarray}
c_2&=&\left(\frac{2}{\gamma}\cos(i\pi\gamma/4)+\frac{2}{i\gamma}\sin(i\pi\gamma/4)\right)^{-1/2}\\
&=&\sqrt{\frac{\gamma}{2}}e^{-\pi\gamma/8}.
\end{eqnarray} 
So finally we have the complete solutions for the spin up and down amplitudes respectively:
\begin{eqnarray}
A(z)&=&\sgn(\Delta)e^{i3\pi/4}e^{-\pi\gamma/8}D_{-i\gamma/2}(-iz)\label{eq:LZprob1}\\
B(z)&=&\sqrt{\frac{\gamma}{2}}e^{-\pi\gamma/8}D_{-1-i\gamma/2}(-iz).\label{eq:LZprob2}
\end{eqnarray}
Using Eqs.~\eqref{eq:Dn0}, \eqref{eq:LZprob1}, \eqref{eq:LZprob2}, we can also show the useful identity
\be
A(0)B^*(0)=e^{i3\pi/4}e^{-\gamma\pi/4}\frac{\sgn(\Delta)\pi\sqrt{\gamma}}{2\Gamma(1/2+\nu/2)\Gamma(1-\nu/2)}.\label{eq:LZident}
\ee

Using the asymptotic expansions for the parabolic cylinder function, one can show that the probability of a non-adiabatic transition (i.e. spin up at $t=\infty$) is $e^{-\pi\gamma}$. This is the Landau-Zener result, and shows that $\gamma$ is the appropriate adiabaticity parameter. For $\gamma\gg1$, the system remains in the ground state, while for $\gamma\ll1$ the system transitions to the excited state. For our purposes however, we are interested in the case $t=0$ for which we make use of Eqs.~\eqref{eq:LZident}. Supplementary Fig.~\ref{fig:LZ} shows the time-dependence of the transition for small $\gamma$. It is important to note that changes to the distribution of probability only begin very close to $t=0$.

\subsection{Application to the two interacting spins}\label{app:LZ2}
Even though our model,
\be
\mathcal{H}^{\pm}=-(\b{H}_1\cdot\b{\sigma}^1\pm\b{H}_2\cdot\b{\sigma}^2)\pm\tilde{r}f(\theta)\sigma_z^1\sigma_z^2,\label{eq:H_twoqubit}
\ee
is a four-state system, we now show that the dynamics of this model are well captured by Eqs.~\eqref{eq:LZprob1}, \eqref{eq:LZprob2}. We consider the symmetric case $M_1=M_2=M<H$. For simplicity, here it suffices to treat just the $\mathcal{H}^+$ Hamiltonian. 

First, we write $\mathcal{H}^+$ in the singlet-triplet basis 
with ${\bm s}={\bm \sigma}^1+{\bm \sigma}^2$. In this basis, we have
\be
\mathcal{H}^+=-\tilde{r}f(\theta)|0,0\rangle\langle0,0| + \mathcal{H}^+_{\rm trip}.
\ee
Here $|0,0\rangle$ refers to the singlet state and
\be
\mathcal{H}_{\rm trip}^+=\frac{1}{2}\tilde{r}f(\theta)s_z^2-H\sqrt{2}\sin\theta s_x-(H\cos\theta+M) s_z-\tilde{r}f(\theta)\mathbb{I},\label{eq:Htrip}
\ee
where we have made use of the Gell-Mann matrices 
\be
s_x=\begin{pmatrix}
0 & 1 & 0\\
1 & 0 & 1\\
0 & 1 & 0
\end{pmatrix}\;\;
s_z=2\begin{pmatrix}
1 & 0 & 0\\
0 & 0 & 0\\
0 & 0 & -1
\end{pmatrix},
\ee
represented in the triplet basis
\begin{eqnarray}
|1,1\rangle=(1,0,0)^T,\;\;
|1,0\rangle=(0,1,0)^T,\;\;
|1,-1\rangle=(0,0,1)^T.\nonumber\\
\label{eq:basisnew4}
\end{eqnarray}
The singlet component is completely decoupled from the equations of motion, provided we initialize the spins in the ground state at the north pole, which is the triplet state $|1,1\rangle$ for $\tilde{r}<(H+M)/f(0)$; the ground state is $|1,0\rangle$ otherwise. 

The total Chern number is encoded in the spin magnetization value according to
\be
\mathcal{C}=\langle s_z(\theta=0)\rangle-\langle s_z(\theta=\pi)\rangle.
\ee

To simplify the dynamics further we assume that any transitions to excited states occur near $\theta=\pi$. This assumption is justified because for a broad class of interactions, the gap at $\theta=0$ is much larger than at $\theta=\pi$. 
Near $\theta=\pi$ the $|1,1\rangle$ state always has the highest energy, so we may project it out and write an effective two-state model to match the Landau-Zener model: 
\begin{eqnarray}
\mathcal{H}^+_{\rm eff}&=&-[\tilde{r}f(\theta)+H\cos\theta+M]\sigma^z-\sqrt{2}H\sin\theta\sigma_x+(H\cos\theta+M)\mathbb{I},\label{eq:Heff}
\end{eqnarray}
which upon rotating the Pauli matrices and expanding about $\theta=\pi$, takes the form of~\eqref{eq:LZham}.

\subsection{Reversibility}
One might wonder about the reversibility of the above sweep protocol. We investigate this by extending the time evolution to a full sweep around the sphere $t\in[0,2\pi/v]$, with the interaction symmetrized about $\pi$: $f(\theta)\rightarrow f(\pi-|\pi-\theta|)$. From the numerical time evolution over this range we compute a ``reversed Chern number'' defined as
\be
\mathcal{C}^{j}_r\equiv\frac{1}{2}\bigg(\left\langle\sigma_{z}^{j}(\theta=2\pi)\right\rangle-\left\langle\sigma_{z}^{j}(\theta=\pi)\right\rangle\bigg),
\ee
which is shown by the dashed lines of Fig.~\ref{fig:reverse}. In the limit $v\rightarrow0$, we find $\mathcal{C}^{j}_r=\mathcal{C}^{j}$ as expected. For higher speeds, the two Chern numbers agree in the small $\tilde{r}$ and large $\tilde{r}$ regimes, but deviate significantly near the transition. This too can be understood from the Landau-Zener physics as follows. Recall that in the systems we consider, the gap is smallest at $\theta=\pi$ and we can assume that diabatic transitions take place only in the vicinity of $\theta=\pi$ where the two-state effective Hamiltonian~\eqref{eq:Heff} is valid. Since the probability for such a transition only varies appreciably close to $\theta=\pi$ (in accordance with Fig.~\ref{fig:LZ}), we may take it's asymptotic value when considering the spin state after a complete cycle. In other words, the probability of measuring the two-qubit system in the first excited state at $\theta=2\pi$ is $\approx e^{-\pi\gamma}$. Now from Eq.~\eqref{eq:Htrip}, we see that the first excited state at $\theta=2\pi$ is always $|1,0\rangle$ (as long as $\tilde{r}<(H+M)/f(0)$) which has a $\sigma_z^j$ value of 0, while the ground state is $|1,1\rangle$ with a $\sigma_z^j$ value of 1. Thus,
\begin{eqnarray}
\langle \sigma^j_z(\theta=2\pi)\rangle = (1-e^{-\pi\gamma}),
\end{eqnarray}
so that
\begin{eqnarray}
\mathcal{C}^{j}_r&=&\frac{1}{2}\bigg(\langle\sigma_{z}^{j}(\theta=2\pi)\rangle-\langle\sigma_{z}^{j}(\theta=0)\rangle+\langle\sigma_{z}^{j}(\theta=0)\rangle-\left\langle\sigma_{z}^{j}(\theta=\pi)\right\rangle \bigg)\\
&\approx&\frac{1}{2}\bigg((1-e^{-\pi\gamma})-1\bigg)+\mathcal{C}^j\\
&=&\mathcal{C}^j-\frac{1}{2}e^{-\pi\gamma}.\label{eq:Crj}
\end{eqnarray}
Since
\be
\gamma=\frac{(\tilde{r}f(\pi)-H+M)^2}{\sqrt{2}Hv},
\ee
(see Eq.~(21) of the article), the deviation between the regular and reversed Chern numbers is a Gaussian cantered at the critical coupling with a variance that vanishes in the adiabatic limit. Upon plugging in Eq.~(23) from the article for $C^j$, we plot this analytic result in the dashed black curve of Supplementary Fig.~\ref{fig:reverse}.


\subsection{Stability of fractional phases in quantum spin models}

\noteKLH{ In addition, we have verified the presence of fractional phases for $N>2$ spins, related to Sec. \ref{sec:nspins} above, from the time evolution of spin models in the quantum circuit simulator Cirq. This is useful for studying larger spin systems and illustrates that these phases can indeed be seen through the action of unitary gates in a generic quantum computer. To do so, we implement our trajectory via the trotter decomposition. Dividing the time interval $[t_0,t_f]$ into $N_{\rm steps}$ with $\Delta t\equiv(t_f-t_0)/N_{\rm steps}$, and ignoring terms of $\mathcal{O}(\Delta t^2)$, the quantum circuit encoding this time-evolution is given by a product of unitary gates. For the $N$-spins Ising chain 
\be
\mathcal{H}=\sum_{i=1}^N(-\b{H_i}\cdot\bsigma^i+\tilde{r}\sigma_z^i\sigma_z^{i+1}), 
\ee
with $\b{H}_i=(H\sin\theta\cos\phi,H\sin\theta\sin\phi,H\cos\theta+M_i)$,
the circuit is defined by the sequence 
\be
U(t_f,t_0)\approx\Pi_{j=0}^{N_{\rm steps}}\Pi_{l=0}^N g_x^l(vt_j)P^{l,l+1}_z(\tilde{r})g_z^l(vt_j),
\ee
where $g^l_x (vt_j)\equiv e^{i(\Delta t H\sin(v t_j)\sigma^l_x)}$, $g^l_z (vt_j)\equiv e^{i(\Delta t (H\cos(v t_j)+M_l)\sigma^l_z)}$ are rotation gates acting on the spin (qubit) $l$ and $P^{l,l+1}_z(\tilde{r})\equiv e^{-i(\Delta t \tilde{r}\sigma_z^l\sigma_z^{l+1})}$ is the two-qubits parity gate, with the Planck constant $\hbar=h/(2\pi)$ set to unity. A situation with $N$=5 is shown in Fig.~\ref{fig:cirq}. Even with the Trotter error due to the time discretization, the fractional Chern phase $\mathcal{C}^j=3/5$ is clearly visible.}


\pagebreak


\begin{figure}[t]
	\centering
	\includegraphics[width=0.56\columnwidth]{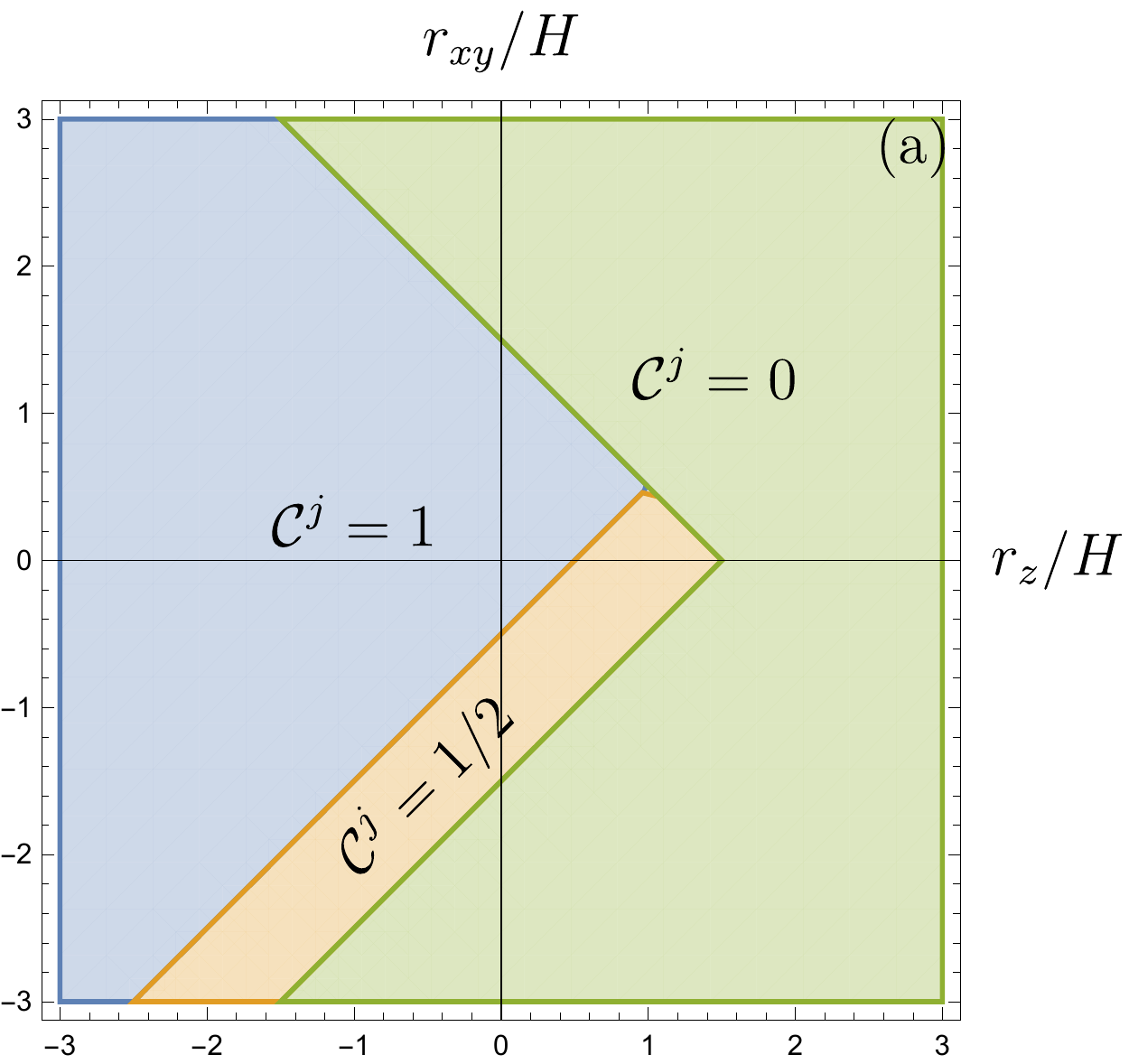}
	\includegraphics[width=0.56\columnwidth]{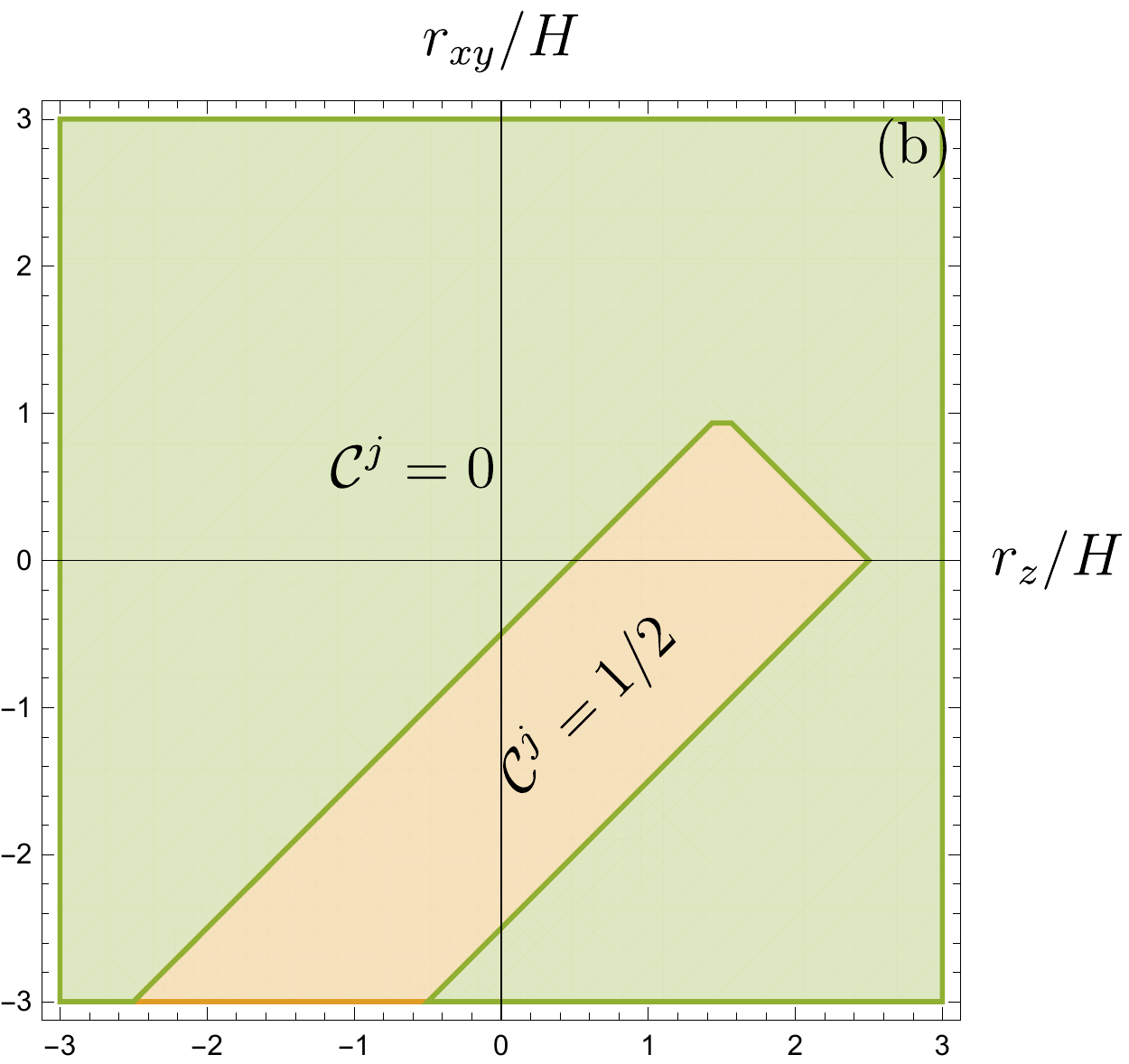}
	\caption{Phase diagrams for the generalized spin model. (a) $M=H/2$. This is qualitatively similar to the phase diagram for all $M<H$. (b) $M=3H/2$. This is qualitatively similar to the phase diagram for all $M>H$.}\label{fig:anis_phase}
\end{figure}

\begin{figure}[t]
	\centering
	\includegraphics[width=0.56\columnwidth]{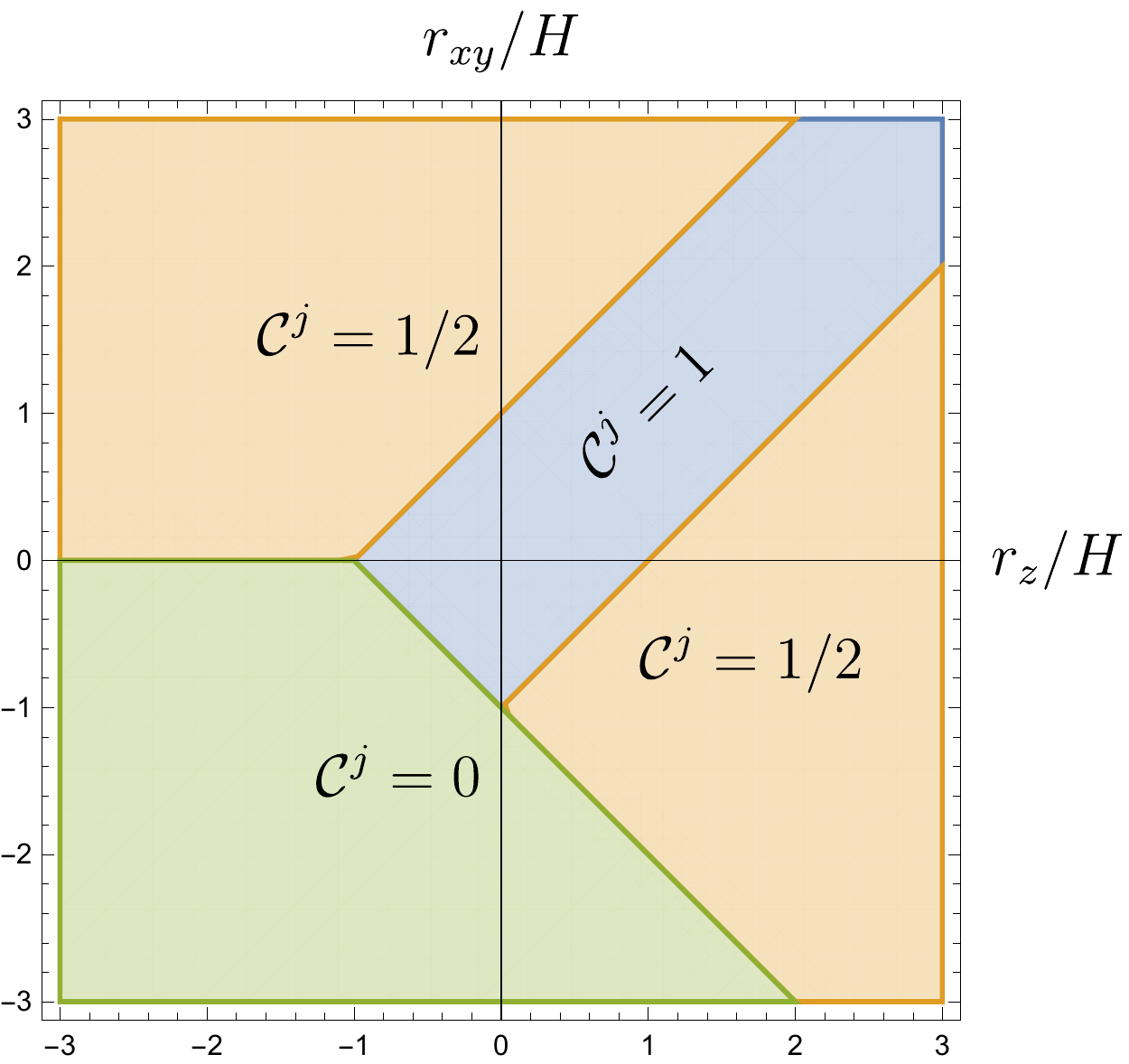}
	\caption{Phase diagram for the generalized spin model with $M=0$ and $\theta$-dependent interactions $f(\theta)=g(\theta)=\sin(\theta-\pi/2)$.}\label{fig:phasesin}
	\end{figure}

\begin{figure}[t]
	\centering
	\includegraphics[width=0.96\columnwidth]{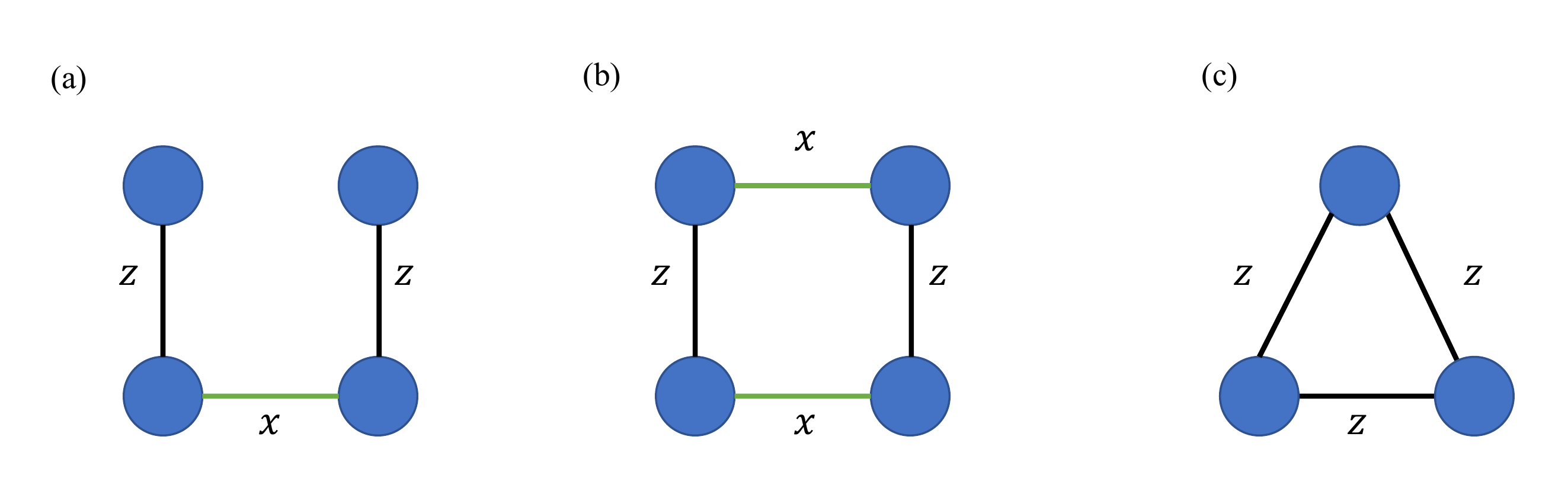}
	\caption{$4$-spin and $3$-spin configurations considered in this paper. $z$ denotes an antiferromagnetic Ising interaction: $r\sigma_i^z\sigma_j^z$, $r>0$. $x$ denotes a transverse coupling $r_x\sigma^x_i\sigma_j^x$. Configuration (a) yields an integer-valued $\mathcal{C}^j$ for all spins, while (b) and (c) produce rational values of $\mathcal{C}^j$.}\label{fig:Nspins}
	\end{figure}

\begin{figure}[b]
	\centering
	\includegraphics[width=0.68\columnwidth]{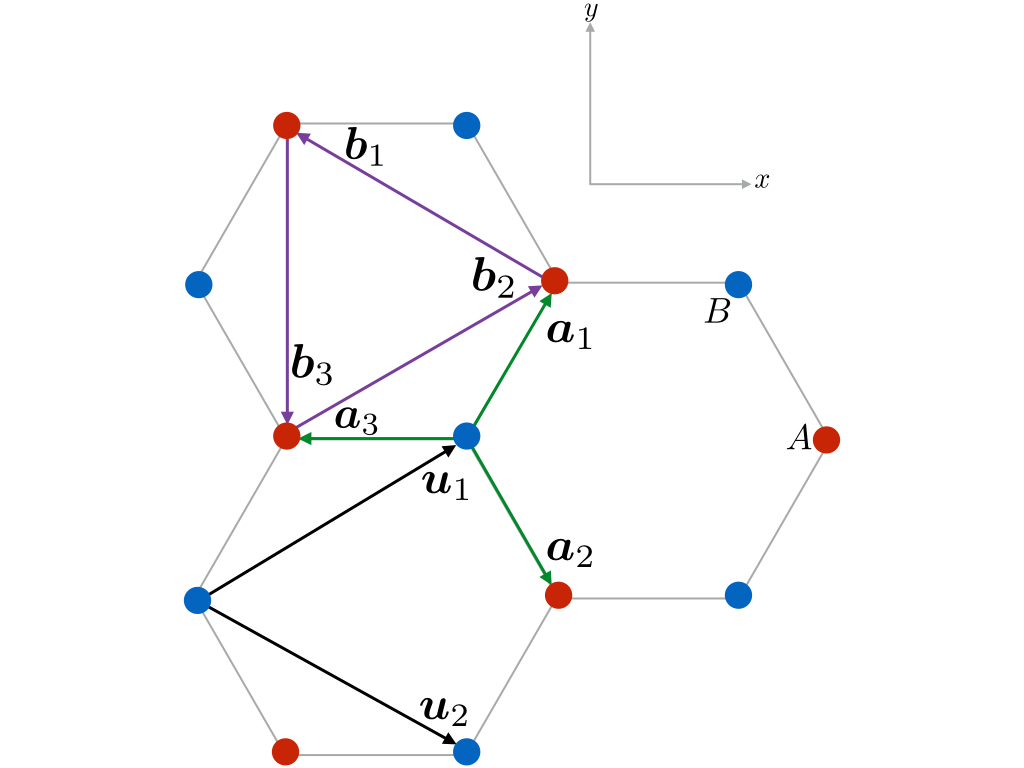}
	\caption{Orientation of vectors on the honeycomb lattice.}\label{fig:Hex}
\end{figure}

\begin{figure}[t]
	\centering
	\includegraphics[width=0.58\columnwidth]{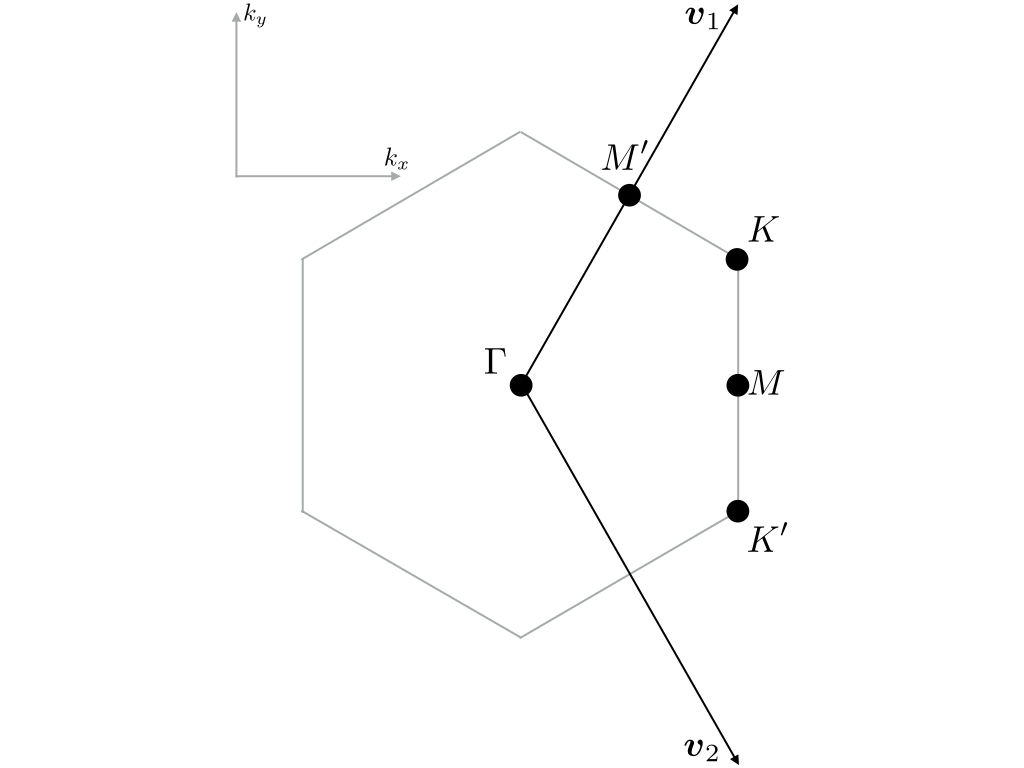}
	\caption{First Brillouin zone of the honeycomb lattice.}\label{fig:Hex2}
\end{figure}

\begin{figure}[t]
	\centering
	\includegraphics[width=0.52\columnwidth]{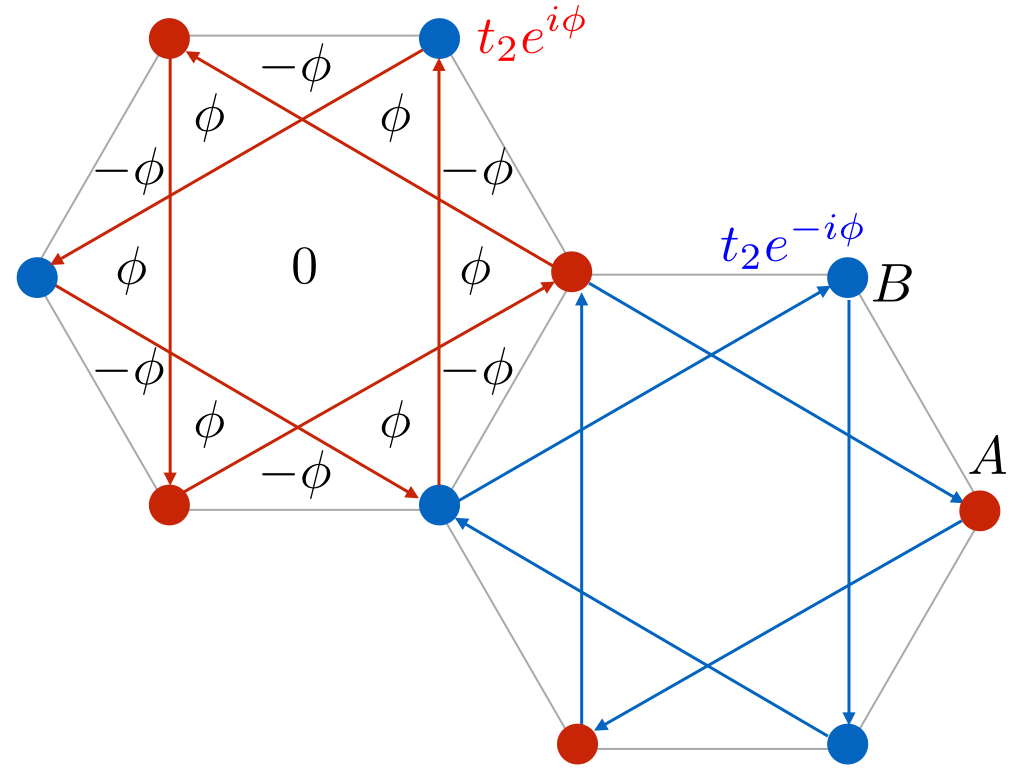}
	\caption{Next-nearest neighbour hoppings and flux orientation for the Haldane model.}\label{fig:Hex3}
\end{figure}

\begin{figure}[t]
	\centering
	\includegraphics[width=0.99\columnwidth]{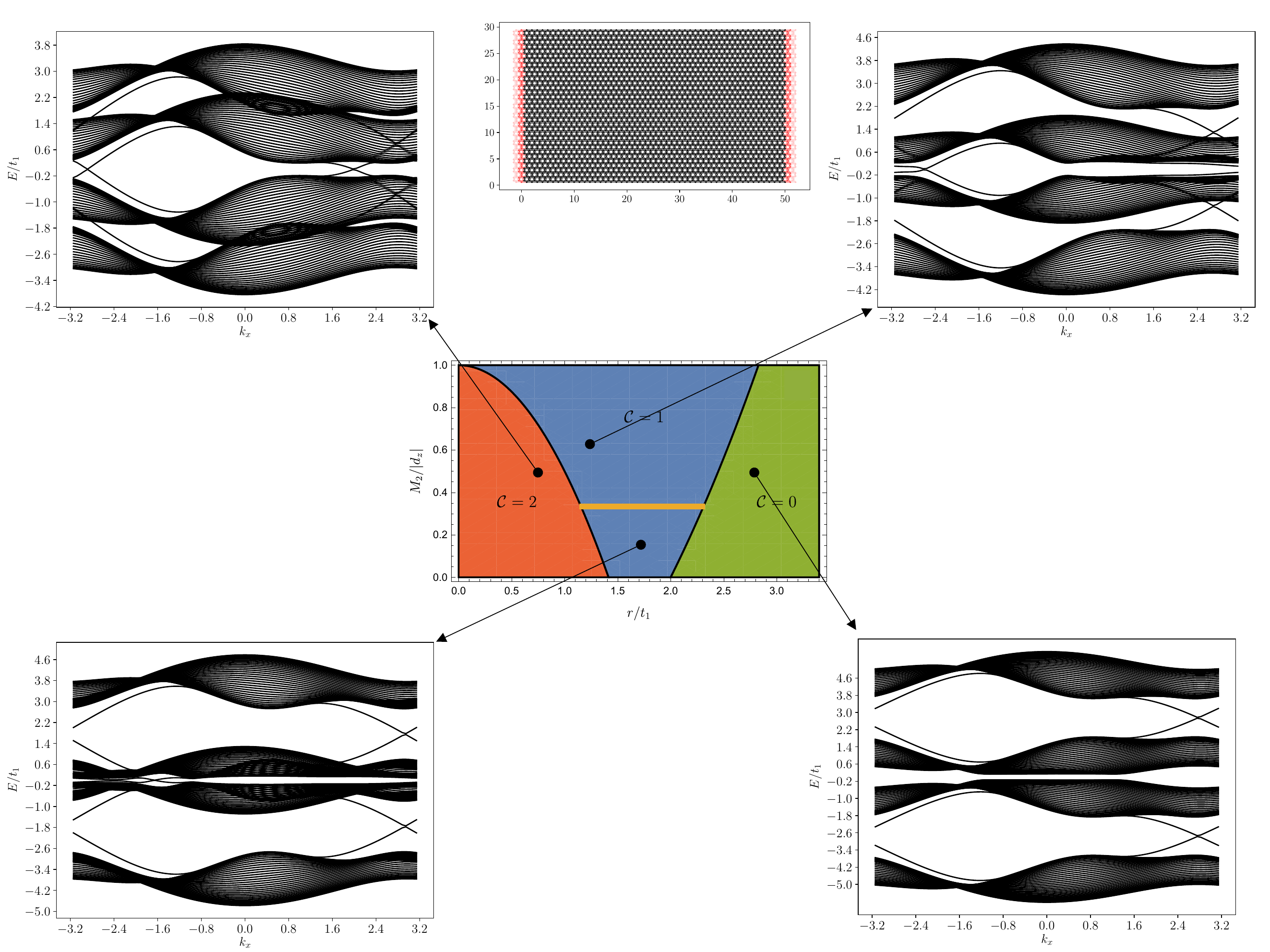}
	\caption{Spectra computed at different points in the phase diagram for the ribbon geometry shown in the top centre with $M_1=\sqrt{3}/3t_1$. The bilayer consists of $30\times50$ sites with periodic boundary conditions in the $x$-direction (indicated in red). Top left: $M_2=0.5\sqrt{3}t_1$, $r=0.75t_1$. Top right: $M_2=0.65\sqrt{3}t_1$, $r=1.25t_1$. Bottom left: $M_2=0.15\sqrt{3}t_1$, $r=1.75t_1$. Bottom right: $M_2=0.5\sqrt{3}t_1$, $r=2.75t_1$.}\label{fig:edge1}
\end{figure}

\begin{figure}[t]
	\centering
	\includegraphics[width=0.98\columnwidth]{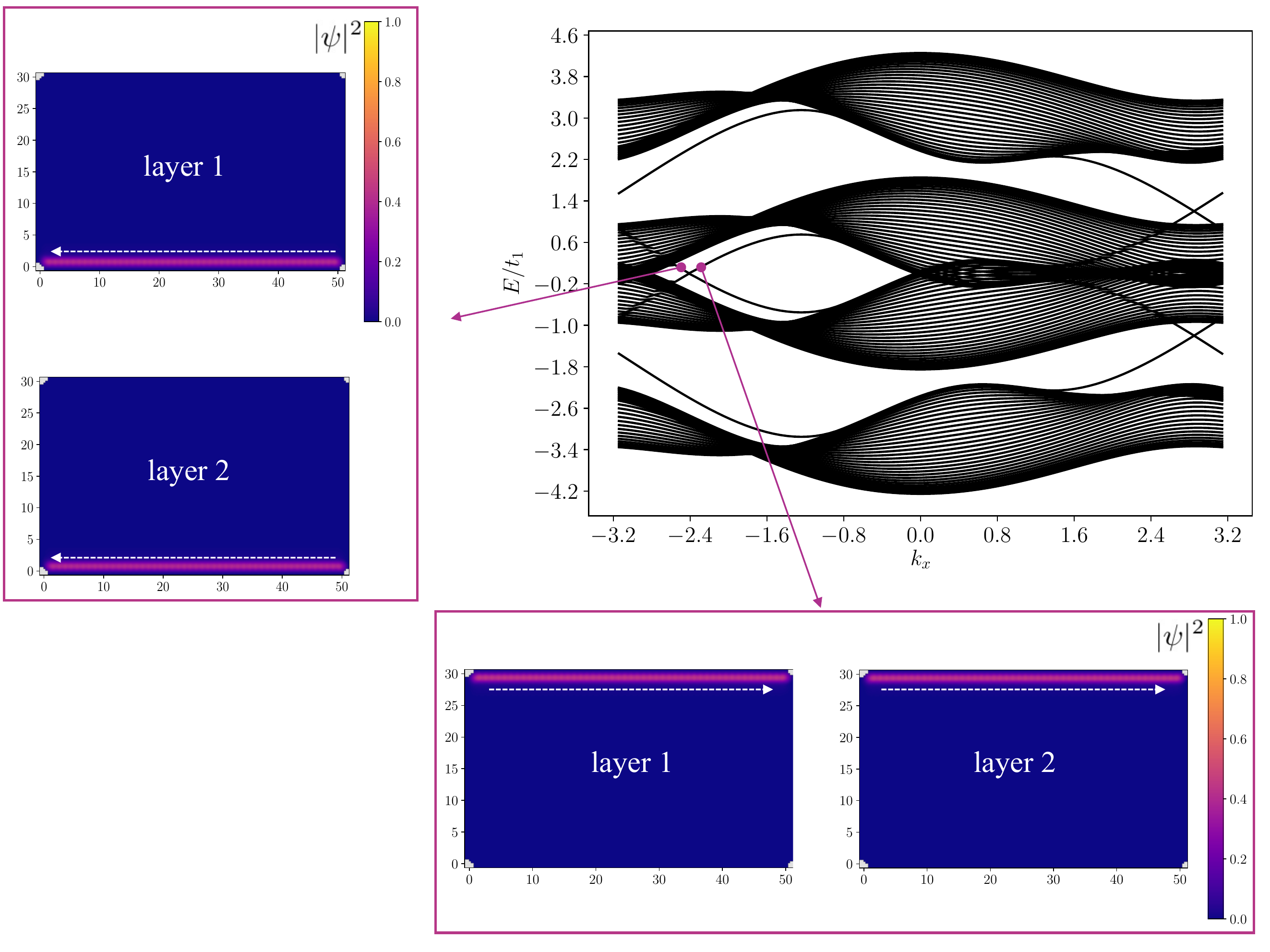}
	\caption{Top right: Spectra computed for the same ribbon geometry as Fig.~\ref{fig:edge1}, but with $M_1=M_2=\sqrt{3}/3t_1$ and $r=1.2t_1$ (in the $\tilde{\mathcal{C}}^j=1/2$ phase). Top left: Left chiral mode at $k_x=-2.5$ showing computed probability density on each layer. Bottom right: Right chiral mode at $k_x=-2.3$ showing computed probability density on each layer. White arrows indicate the current direction.}\label{fig:edge2}
\end{figure}

\begin{figure}[t]
	\centering
	\includegraphics[width=0.7\columnwidth]{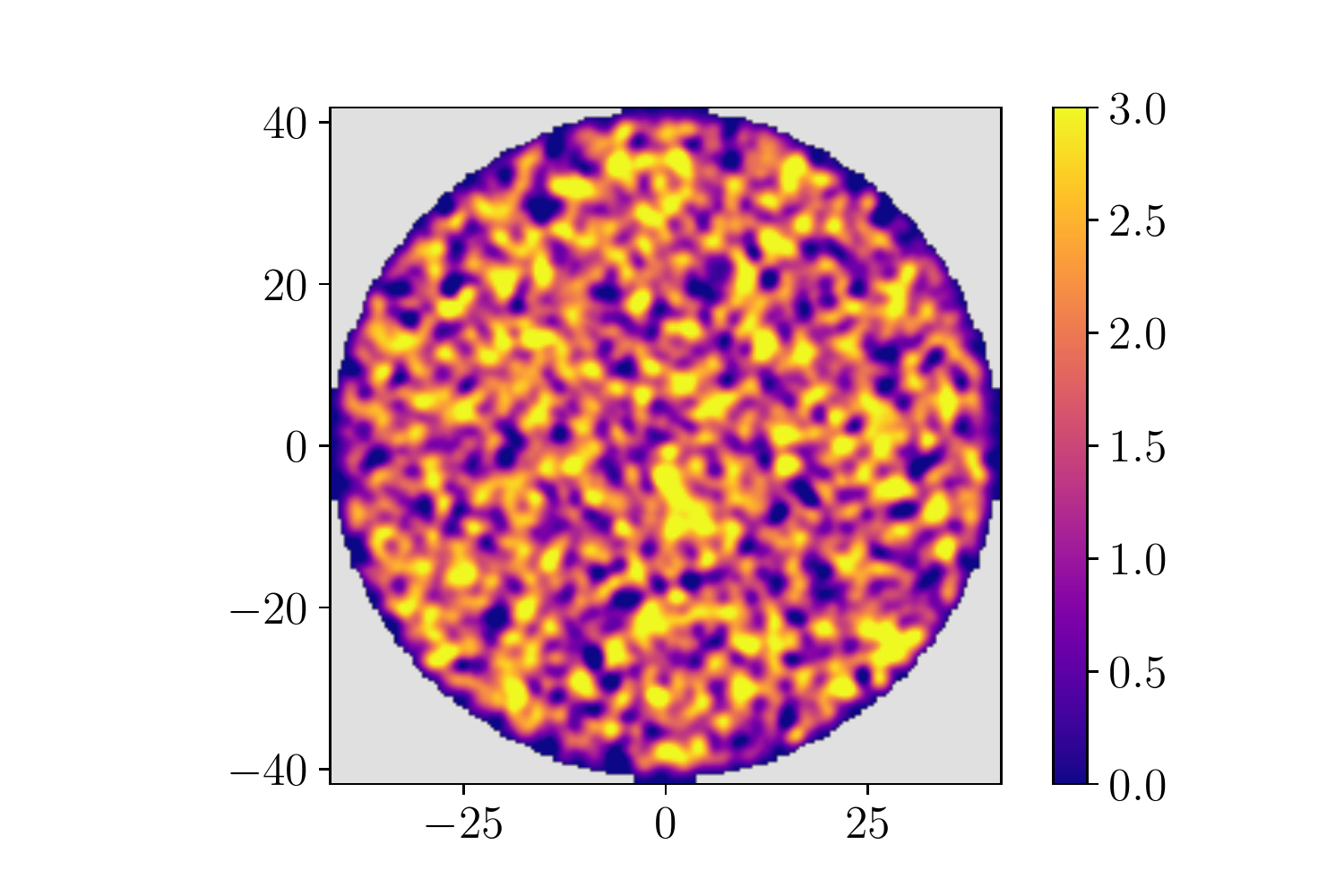}
	\includegraphics[width=0.7\columnwidth]{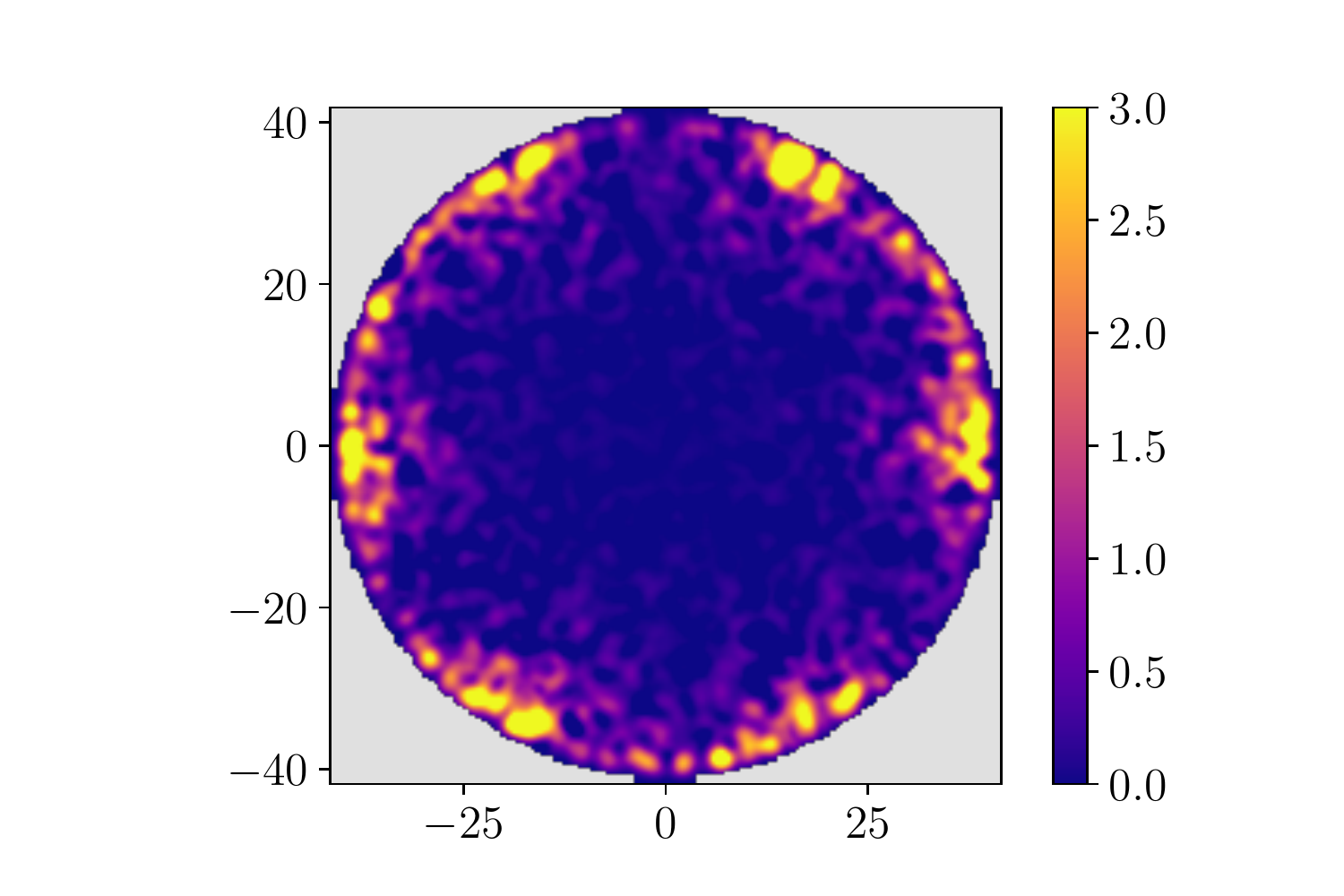}
	\caption{Top: Local density of states for a disk geometry with 30-site radius with $M_1=M_2=\sqrt{3}/3t_1$ and $r=1.4t_1$ showing the edge mode and additional bulk states coming from the nodal ring semimetal in the reciprocal space. Bottom: Local density of states shifted very slightly from the line of symmetry: $M_2=M_1+0.2$, $M_1=\sqrt{3}/3t_1$ and $r=1.4t_1$ in the blue region of the phase diagram, showing the single chiral edge mode.}\label{DOS}
\end{figure}

\begin{figure}[t]
	\centering
	\includegraphics[width=0.76\columnwidth]{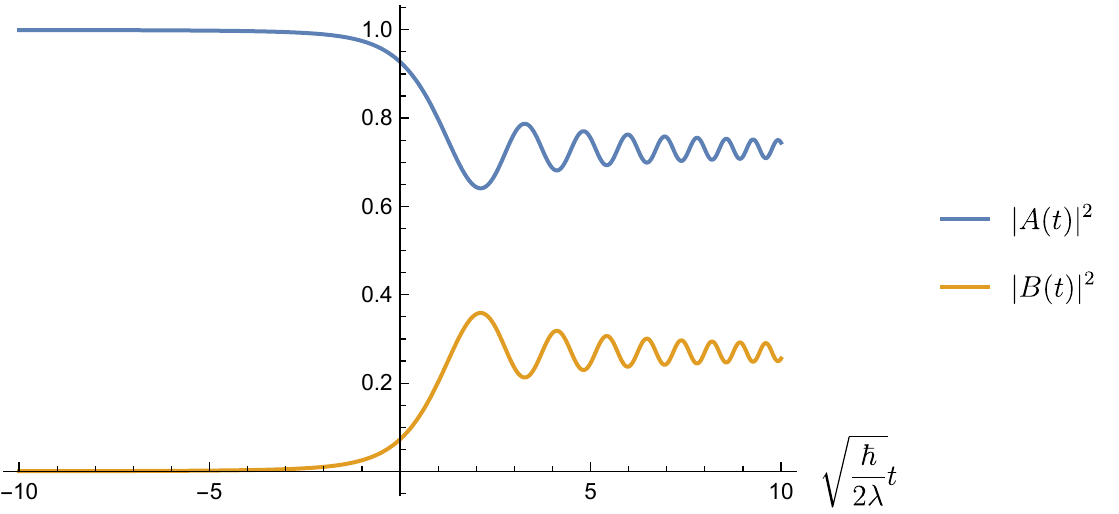}
\caption{Probabilities ($|A(t)|^2$) and ($|B(t)|^2$) with $\gamma=0.1$.}
\label{fig:LZ}
\end{figure}

\begin{figure}[t]
	\centering
	\includegraphics[width=0.76\columnwidth]{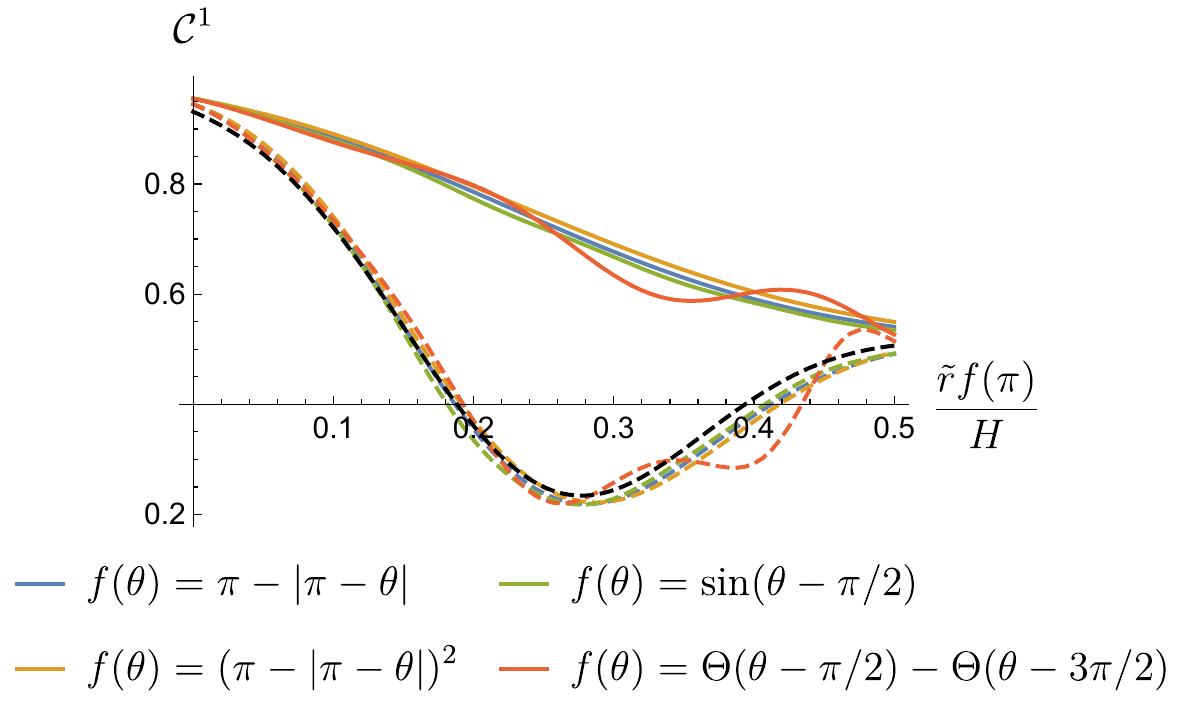}
\caption{Numerically determined Chern number $C^j$ (solid) and reversed Chern number $C_r^j$ (dashed) of a single spin vs $\tilde{r}f(\pi)/H$ for different interactions with $v=0.05H$; $\Theta$ refers to the Heaviside step function. The dashed black line shows the analytic approximation of $C_r^j$ (Eq.~\eqref{eq:Crj}) which is universal for a given speed $v$.}
\label{fig:reverse}
\end{figure}

\begin{figure}
	\centering
	\includegraphics[width=0.5\columnwidth]{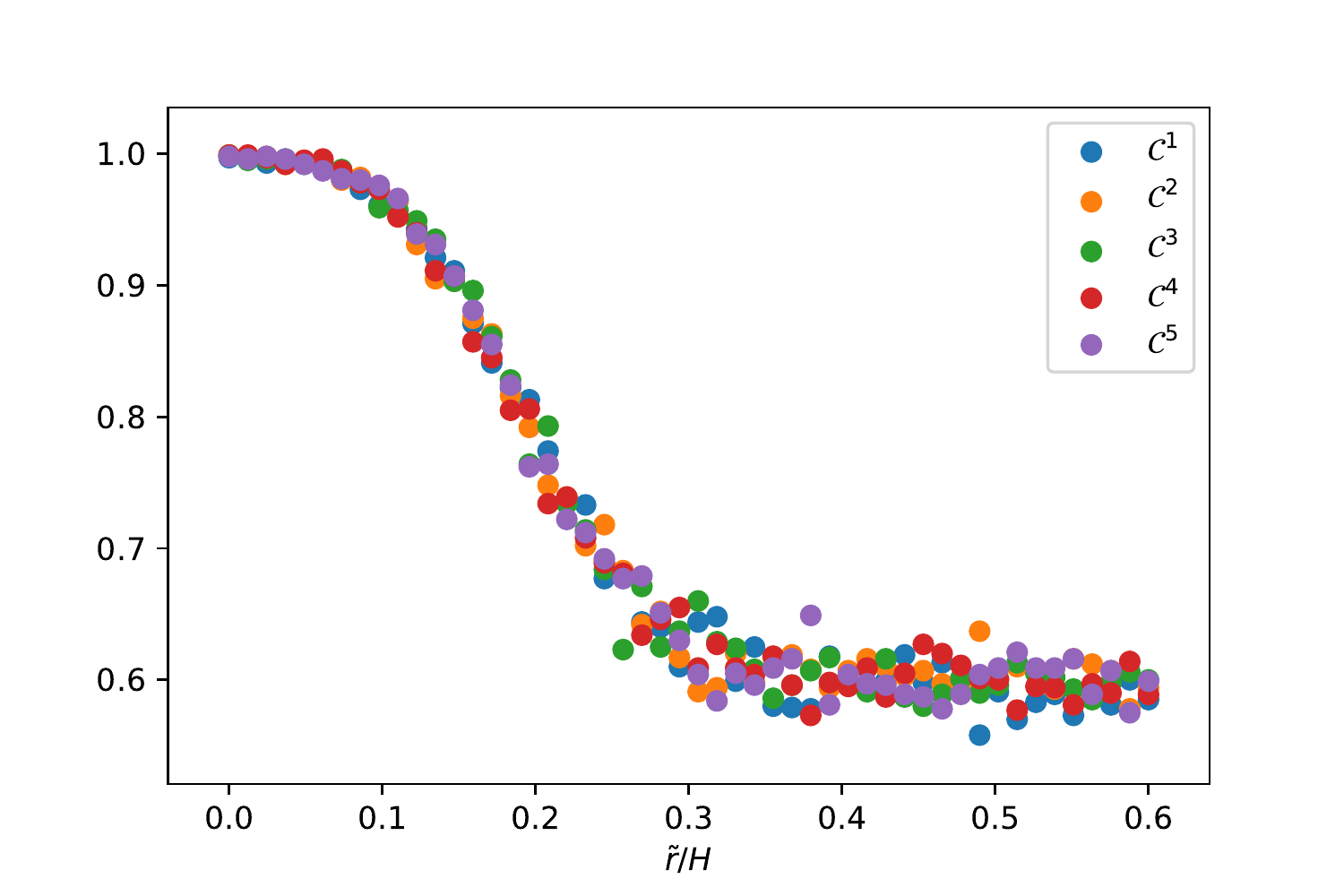}
\caption{\noteKLH{Partial Chern numbers as a function of the coupling $\tilde{r}$ measured in a five-spins quantum circuit simulation with nearest-neighbour Ising interactions and periodic boundary conditions. To time-evolve the spins (qubits), we use a Trotter decomposition with 800 time steps and sweep velocity $v=0.03H$. The bias field for all qubits is fixed to $M=0.6H$.}}\label{fig:cirq}
\end{figure}